\pdfoutput=1

\documentclass[11pt]{article}

\usepackage[final]{acl}

\usepackage{times}
\usepackage{latexsym}

\usepackage[T1]{fontenc}

\usepackage[utf8]{inputenc}
\usepackage{algpseudocode}
\usepackage{hyperref}       
\usepackage{url}            
\usepackage{booktabs}       
\usepackage{amsfonts}       
\usepackage{nicefrac}       
\usepackage{supertabular}
\usepackage{xcolor}         
\usepackage{amsmath}
\usepackage{amssymb}
\usepackage{amsthm}
\usepackage[capitalize,noabbrev]{cleveref}
\usepackage{multirow}
\usepackage{multicol}
\usepackage{graphicx}
\usepackage{wrapfig}
\usepackage{makecell}
\usepackage{enumitem}
\usepackage{color}
\usepackage{marvosym}
\usepackage{fontawesome}
\usepackage{wrapfig}
\usepackage{diagbox}
\usepackage{tabularx}
\usepackage{algorithm}
\usepackage{blindtext}
\usepackage{titlesec}
\usepackage{mathrsfs}
\usepackage{dutchcal}
\usepackage{subfigure}
\usepackage{array}
\usepackage{balance}
\usepackage{calligra}
\usepackage{colortbl}
\usepackage{courier}
\usepackage{csvsimple}
\usepackage{fancybox}
\usepackage{fontenc}
\usepackage{listings}
\usepackage{longtable}
\usepackage{lscape}
\usepackage{pifont}
\usepackage{rotating}
\usepackage{setspace}
\usepackage[most]{tcolorbox}
\usepackage{threeparttable}
\usepackage{tikz}
\usepackage{soul}
\usepackage[normalem]{ulem}
\usepackage{soul}
\usepackage{wasysym}
\usepackage{xspace}
\usepackage{amsthm}
\usepackage{mdframed}
\usepackage{threeparttable}
\usepackage{authblk}

\Crefname{assumption}{Assumption}{Assumptions}

\usepackage{scrwfile}
\TOCclone[\contentsname~(\appendixname)]{toc}{atoc}
\newcommand\StartAppendixEntries{}
\AfterTOCHead[toc]{%
  \renewcommand\StartAppendixEntries{\value{tocdepth}=-10000\relax}%
}
\AfterTOCHead[atoc]{%
  \edef\maintocdepth{\the\value{tocdepth}}%
  \value{tocdepth}=-10000\relax%
  \renewcommand\StartAppendixEntries{\value{tocdepth}=\maintocdepth\relax}%
}
\newcommand*\appendixwithtoc{%
  \cleardoublepage
  \appendix
  \addtocontents{toc}{\protect\StartAppendixEntries}
  \listofatoc
}
\usepackage{blindtext}

\usepackage{amsmath}
\usepackage{amssymb}
\usepackage{mathtools}
\usepackage{amsthm}

\usepackage{amsmath,amsfonts,bm}

\def\eqref#1{equation~\ref{#1}}

\def\1{\bm{1}}

\DeclareMathAlphabet{\mathsfit}{\encodingdefault}{\sfdefault}{m}{sl}
\SetMathAlphabet{\mathsfit}{bold}{\encodingdefault}{\sfdefault}{bx}{n}

\usepackage[capitalize,noabbrev]{cleveref}

\theoremstyle{plain}
\newtheorem{theorem}{Theorem}[section]

\newtheorem{lemma}[theorem]{Lemma}

\theoremstyle{definition}

\theoremstyle{remark}

\usepackage[textsize=tiny]{todonotes}

\newcommand{\tool}[1]{\texttt{{CodeAgent}}\xspace}

\title{\texorpdfstring{\includegraphics[height=1em]{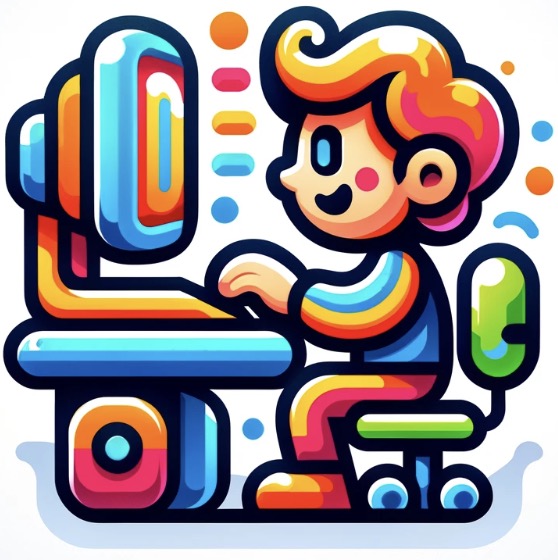}}{} \tool{}: Autonomous Communicative Agents for Code Review}

\author[1]{Xunzhu Tang}
\author[2]{Kisub Kim}
\author[1]{Yewei Song}
\author[3]{Cedric Lothritz}
\author[4]{Bei Li}
\author[5]{Saad Ezzini}
\author[6,*]{\\Haoye Tian}
\author[1]{Jacques Klein}
\author[1]{Tegawendé F. Bissyandé}
\affil[1]{University of Luxembourg}
\affil[2]{Singapore Management University}
\affil[3]{Luxembourg Institute of Science and Technology}
\affil[4]{Northeastern University}
\affil[5]{Lancaster University}
\affil[6]{The University of Melbourne}

\begin{document}
\maketitle

\renewcommand{\thefootnote}{\fnsymbol{footnote}} 
\footnotetext[1]{Corresponding author.} 
\renewcommand{\thefootnote}{\arabic{footnote}}

\begin{abstract}

Code review, which aims at ensuring the overall quality and reliability of software, is a cornerstone of software development. Unfortunately, while crucial, Code review is a labor-intensive process that the research community is looking to automate. Existing automated methods rely on single input-output generative models and thus generally struggle to emulate the collaborative nature of code review. This work introduces \tool{}, a novel multi-agent Large Language Model (LLM) system for code review automation. \tool{} incorporates a supervisory agent, QA-Checker, to ensure that all the agents' contributions address the initial review question.
We evaluated \tool{} on critical code review tasks: (1) detect inconsistencies between code changes and commit messages, (2) identify vulnerability introductions, (3) validate code style adherence, and (4) suggest code revision. The results demonstrate \tool{}'s effectiveness, contributing to a new state-of-the-art in code review automation. Our data and code are publicly available (\url{https://github.com/Code4Agent/codeagent}). 
\end{abstract}

\section{Introduction} 
Code review~\cite{bacchelli2013expectations,bosu2013impact,davila2021systematic} implements a process wherein software maintainers examine and assess code contributions to ensure quality and adherence to coding standards, and identify potential bugs or improvements.
In recent literature, various approaches~\cite{tufano2021towards,tufano2022using} have been proposed to enhance the performance of code review automation.
Unfortunately, major approaches in the field ignore a fundamental aspect: the code review process is inherently interactive and collaborative~\cite{bacchelli2013expectations}. Instead, they primarily focus on rewriting and adapting the submitted code~\cite{watson2022systematic,thongtanunam2022autotransform,staron2020using}. 
In this respect, an effective approach should not only address how to review the submitted code for some specific needs (e.g., vulnerability detection~\cite{chakraborty2021deep,  yang2024security}). Still, other non-negligible aspects of code review should also be considered, like detecting issues in code formatting or inconsistencies in code revision~\cite{oliveira2023systematic, tian2022change, panthaplackel2021deep}. 
However, processing multiple sub-tasks requires interactions among employees in different roles in a real code review scenario, which makes it challenging to design a model that performs code review automatically.

Agent-based systems are an emerging paradigm and a computational framework in which autonomous entities (aka agents) interact with each other~\cite{li2023camel,qian2023communicative,hong2023metagpt} to perform a task.
Agent-based approaches have been proposed to address a spectrum of software engineering tasks~\cite{qian2023communicative,zhang2024autocoderover,tang2023just,tian2023chatgpt}, moving beyond the conventional single input-output paradigm due to their exceptional ability to simulate and model complex interactions and behaviors in dynamic environments~\cite{xi2023rise,yang2024cref,wang2023unleashing}.
Recently, multi-agent systems have leveraged the strengths of diverse agents to simulate human-like decision-making processes~\cite{du2023improving,liang2023encouraging,park2023generative}, leading to enhanced performance across various tasks~\cite{chen2023agentverse,li2023metaagents,hong2023metagpt}. This paradigm is well-suited to the challenge of code review, where multiple reviewers, each with diverse skills and roles, collaborate to achieve a comprehensive review of the code..



{\bf This paper.} Drawing from the success of agent-based collaboration, we propose a {\bf multi-agent-based framework} \tool{} to simulate the dynamics of a collaborative team engaged in the code review process, incorporating diverse roles such as code change authors, reviewers, and decision makers.
In particular, A key contribution of \tool{} is that we address the challenge of prompt drifting~\cite{zheng2024neo,yang2024adapting}, a common issue in multi-agent systems and Chain-of-Thought (CoT) reasoning. This issue, characterized by conversations that stray from the main topic, highlights the need for strategies to maintain focus and coherence~\cite{humanfirst,chae2023dialogue}.
This drift, often triggered by the model-inspired tangents or the randomness of Large Language Models (LLMs), necessitates the integration of a supervisory agent. 
We employ an agent named QA-Checker  (for "Question-Answer Checker") that monitors the conversation flow, ensuring that questions and responses stay relevant and aligned with the dialogue's intended objective.
Such an agent not only refines queries but also realigns answers to match the original intent, employing a systematic approach grounded in a mathematical framework. 

To evaluate the performance of \tool{}, we first assess its effectiveness for typical review objectives such as detecting vulnerabilities~\ref{sec:va} and validating the consistency and alignment of the code format~\ref{sec:cafa}.
We then compare \tool{} with state-of-the-art generic and code-specific language models like ChatGPT~\cite{chatgpt} and CodeBERT~\cite{codebert}. Finally, we assess the performance of \tool{} compared to the state-of-the-art tools for code revision suggestions~\cite{tufano2021towards,thongtanunam2022autotransform,tufano2022using}.
Since each of these related works presents a specific dataset, we also employ them toward a fair comparison. Additionally, we also collect pull requests from GitHub, featuring an extensive array of commits, messages, and comments to evaluate advanced capabilities.
The experimental results reveal that \tool{} significantly outperforms the state-of-the-art, achieving a 41\% increase in hit rate for detecting vulnerabilities.
\tool{} also excels in consistency checking and format alignment, outperforming the target models.
Finally, \tool{} showcases its robustness for code revision by presenting superior average edit progress.

We summarize our contributions as follows:

\begin{itemize}
    \item To the best of our knowledge, we are the first to propose an autonomous agent-based system for practical code review in the field of software maintenance.
    \item We build a new dataset comprising 3\,545 real-world code changes and commit messages. This dataset, which includes all relevant files and details in a self-contained format, is valuable for evaluating advanced code review tasks such as vulnerability detection, code style detection, and code revision suggestions.
    \item We demonstrate the effectiveness of the QA-Checker. This agent monitors the conversation flow to ensure alignment with the original intent, effectively addressing the common prompt drifting issues in multi-agent systems. 
\end{itemize}

Experimental evaluation highlights the performance of \tool{}:  In vulnerability detection, \tool{} outperforms GPT-4 and CodeBERT by 3 to 7 percentage points in terms of the number of vulnerabilities detected. For format alignment, \tool{} outperforms ReAct by approximately 14\% in recall for inconsistency detection. On the code revision task, \tool{} surpasses the state of the art in software engineering literature, achieving an average performance improvement of about 30\% in the Edit Progress metric~\cite{zhou2023generation}. 
\section{\tool{}}
This section details the methodology behind our \tool{} framework. We discuss tasks and definition in Sec~\ref{sec:taskdefinition}, pipeline in Section~\ref{sec:pipeline}, defined role cards in Section~\ref{sec:cardinfo}, and the design of the QA-Checker in Sec~\ref{sec:qachecker}.

\begin{figure*}[ht!]
    \centering
    \includegraphics[width=1\linewidth]{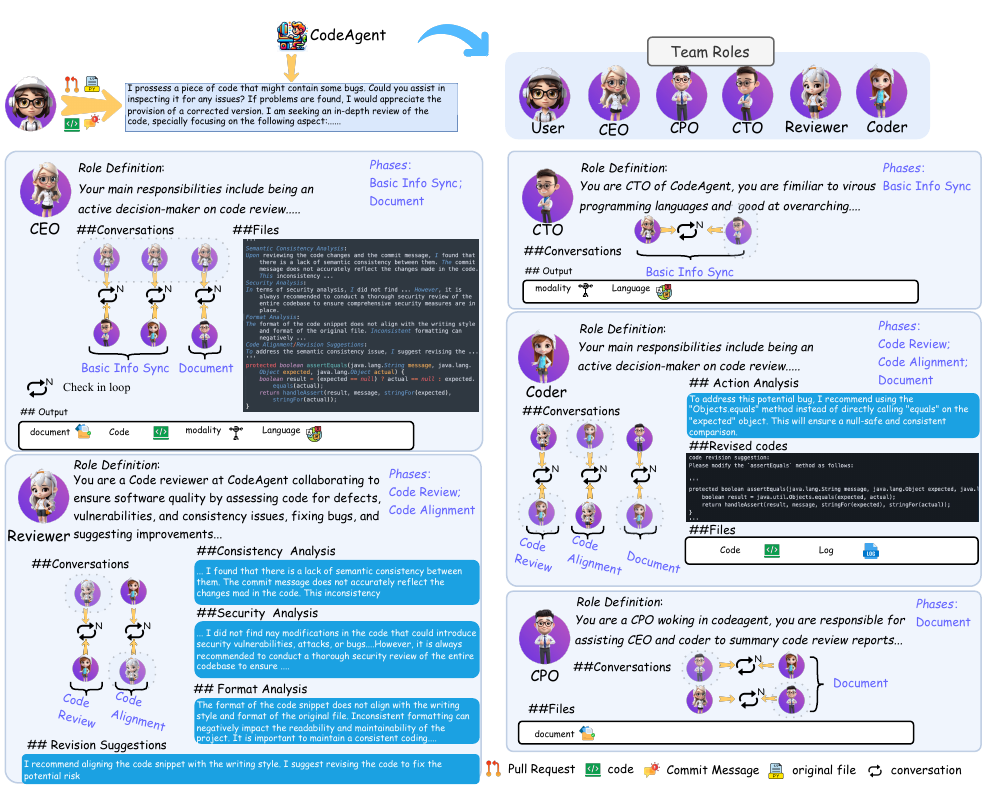}
    \caption{\textbf{A Schematic diagram of role data cards of simulated code review team and their conversations within \tool{}}. We have six characters in \tool{} across four phases, including ``Basic Info Sync", ``Code Review", ``Code Alignment", and ``Document". Code review is a kind of collaboration work, where we design conversations between every two roles for every step to complete the task.
    }
    \label{fig:userdatacard}
\end{figure*}

\begin{figure*}[ht]
    \centering
    \includegraphics[width=\linewidth]{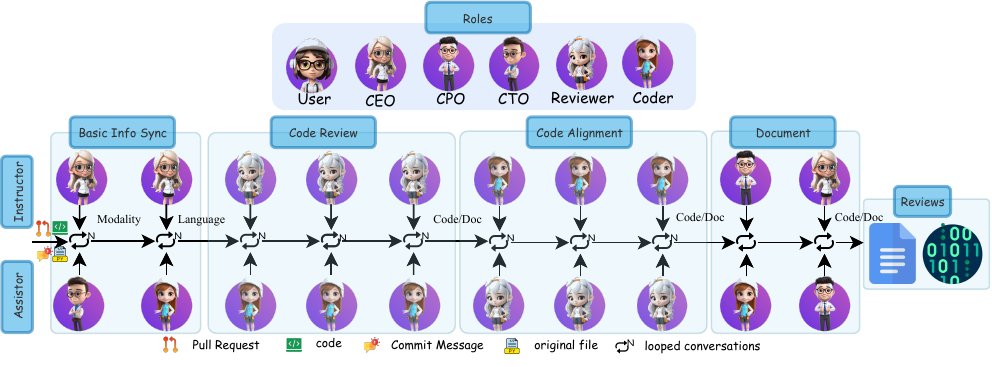}
    \caption{\tool{}'s pipeline/scenario of a full conversation during the code review process among different roles.
    ``Basic Info Sync'' demonstrates the basic information confirmation by the CEO, CTO, and Coder; ``Code Review'' shows the actual code review process; ``Code Alignment'' illustrates the potential code revision; and ``Document'' represents the summarizing and writing conclusion for all the stakeholders.
    All the conversations are being ensured by the Quality Assurance checker until they reach the maximum dialogue turns or meet all the requirements.}
    \label{fig:chatchain}
\end{figure*}
\subsection{Tasks} \label{sec:taskdefinition}
We define \textit{CA}, \textit{VA}, \textit{FA}, and \textit{CR} in as following: 

\noindent
\textbf{\textit{CA}}~\cite{zhang2022consistent}: Consistency analysis between code change and commit message; the task is to detect cases where the commit message accurate describes (in natural language) the intent of code changes (in programming language). 

\noindent
\textbf{\textit{VA}}~\cite{braz2022less}: Vulnerability analysis; the task is to identify cases where the code change introduces a vulnerability in the code. 

\noindent
\textbf{\textit{FA}}~\cite{han2020does}: Format consistency analysis between commit and original files; the task is to validate that the code change formatting style is not aligned with the target code. 

\noindent
\textbf{\textit{CR}}~\cite{zhou2023generation}: Code revisions; this task attempts to automatically suggest rewrites of the code change to address any issue discovered. 
    













\subsection{Pipeline} 
\label{sec:pipeline}

We defined six characters and four phases for the framework. 
The roles of the characters are illustrated in Figure~\ref{fig:userdatacard}.
Each phase contains multiple conversations, and each conversation happens between agents.
The four phases consist of  
1) Basic Info Sync, containing the roles of chief executive officer ({\it CEO}), chief technology officer ({\it CTO}), and Coder to conduct modality and language analysis; 
2) Code Review, asking the {\it Coder} and {\it Reviewer} for actual code review (i.e., target sub-tasks); 
3) Code Alignment, supporting the {\it Coder} and {\it Reviewer} to correct the commit through code revision and suggestions to the author; and 
4) Document, finalizing by synthesizing the opinions of the {\it CEO}, {\it CPO} (Chief Product Officer), {\it Coder}, and {\it Reviewer} to provide the final comments. 
In addition to six defined roles, the proposed architecture of \tool{} consists of phase-level and conversation-level components. The waterfall model breaks the code review process at the phase level into four sequential phases. At the conversation level, each phase is divided into atomic conversations. These atomic conversations involve task-oriented role-playing between two agents, promoting collaborative communication. One agent works as an instructor and the other as an assistant. Communication follows an instruction-following style, where agents interact to accomplish a specific subtask within each conversation, and each conversation is supervised by QA-Checker. 
QA-Checker is used to align the consistency of questions and answers between the instructor and the assistant in a conversation to avoid digression. QA-Checker will be introduced in Section~\ref{sec:qachecker}.

Figure~\ref{fig:chatchain} shows an illustrative example of the \tool{} pipeline. \tool{} receives the request to do the code review with the submitted commit, commit message, and original files. 
In the first phase, \textit{CEO}, \textit{CTO}, and \textit{Coder} will cooperate to recognize the modality of input (e.g., document, code) and language (e.g., Python, Java and Go). 
In the second phase, with the help of \textit{Coder}, \textit{Reviewer} will write an analysis report on consistency analysis, vulnerability analysis, format analysis and suggestions for code revision. 
In the third phase, based on analysis reports, \textit{Coder} will align or revise the code if any incorrect snippets are identified with assistance from \textit{Reviewer}. 
$Coder$ cooperates with \textit{CPO} and \textit{CEO} to summarize the document and codes about the whole code review in the final phase.

\subsection{Role Card Definition} \label{sec:cardinfo}
As shown in Figure~\ref{fig:userdatacard}, we define six characters in our simulation system (\tool{}), including {\it User}, {\it CEO},  {\it CPO}, {\it CTO}, {\it Reviewer}, {\it Coder}, and they are defined for different specific tasks. 


All tasks are processed by the collaborative work of two agents in their multi-round conversations. For example, as a role {\it Reviewer}, her responsibility is to do the code review for given codes and files in three aspects (tasks \textit{CA}, \textit{VA}, and \textit{FA} in Sec~\ref{sec:taskdefinition}) and provide a detailed description of observation. {\it Reviewer}'s code review activity is under the assistance with {\it Coder} as shown in Figure~\ref{fig:chatchain}. Meanwhile, with the Reviewer's assistance, {\it Coder} can process the code revision as shown in the `Revised codes' part in the {\it Coder} card in Figure~\ref{fig:userdatacard}. Apart from {\it Reviewer}, {\it Coder} also cooperates with {\it CTO} and {\it CEO} in the simulated team. 

Each role and conversation, input and output of each conversation is designed in Figure~\ref{fig:userdatacard}. Further information about role definition details is provided in our Appendix-Section~\ref{sec:roledefine}. 

\subsection{Self-Improving CoT with QA Checker}\label{sec:qachecker}
\begin{figure}[ht]
    \centering
    \includegraphics[width=1\linewidth]{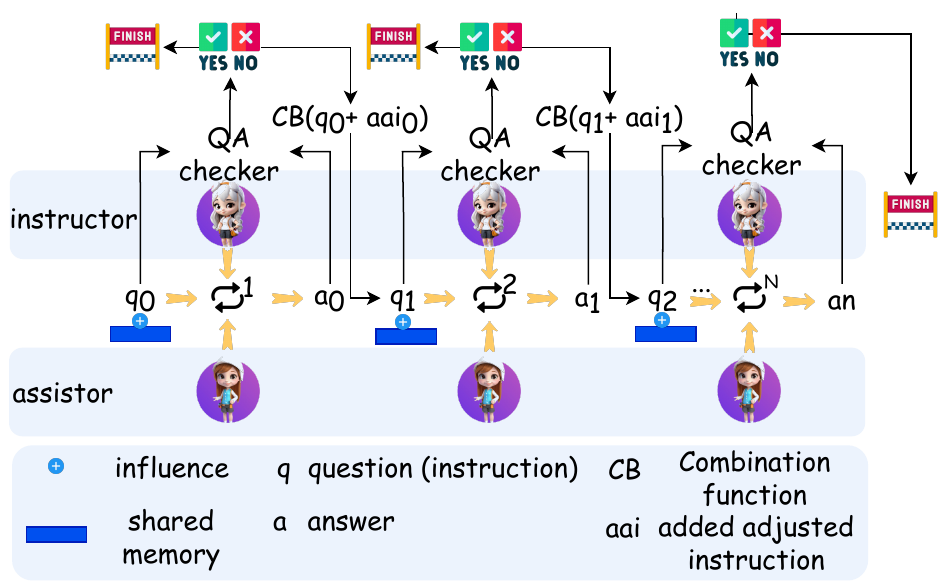}
    \caption{This diagram shows the architecture of our designed Chain-of-Thought (CoT): Question-Answer Checker (QA-Checker). 
    }
    \label{fig:qachecker}
\end{figure}
QA-Checker is an instruct-driven agent, designed to fine-tune the question inside a conversation to drive the generated answer related to the question. 
As shown in Figure~\ref{fig:qachecker}, the initial question (task instruction) is represented as $q_0$, and the first answer of the conversation between $Reviewer$ and $Coder$ is represented as $a_0$. 
If QA-Checker identifies that $a_0$ is inappropriate for $q_0$, it generates additional instructions attached to the original question (task instruction) and combines them to ask agents to further generate a different answer. 
The combination in Figure~\ref{fig:qachecker} is defined as $q_1 = CB(q_0+aai_0)$, where $aai_0$ is the additional instruction attached. 
The conversation between two agents is held until the generated answer is judged as appropriate by QA-Checker or it reaches the maximum number of dialogue turns.

\paragraph{Theoretical Analysis of QA-Checker in Dialogue Refinement}

The QA-Checker is an instruction-driven agent, crucial in refining questions and answers within a conversation to ensure relevance and precision. Its operation can be understood through the following lemma and proof in Appendix \ref{sec:qaalgorithm}.

\section{Experimental Setup}
We evaluate the performance of \tool{} through various qualitative and quantitative experiments across nine programming languages, using four distinct metrics. In this section, we will discuss experimental settings, including datasets, metrics, and baselines. For more information, please see Appendix~\ref{sec:expsetdetail}.

\subsection{Datasets}
To conduct a fair and reliable comparison for the code revision task, we employ the same datasets (i.e., Trans-Review$_{data}$, AutoTransform$_{data}$, and T5-Review$_{data}$) as the state-of-the-art study~\cite{zhou2023generation}.
Furthermore, we collect and curate an additional dataset targeting the advanced tasks. 
Table~\ref{tab:ca_fa_comparison} shows our new dataset which includes over 3,545 commits and 2,933 pull requests from more than 180 projects, spanning nine programming languages: Python, Java, Go, C++, JavaScript, C, C\#, PHP, and Ruby.
It focuses on consistency and format detection, featuring both positive and negative samples segmented by the merged and closed status of pull requests across various languages. 
The detailed information about the dataset can be seen in Appendix-Section~\ref{sec:dataset}.

\begin{table}[htbp]
\centering
\caption{Comparison of Positive and Negative Samples in CA and FA (CA and FA are defined in Section~\ref{sec:taskdefinition}).}
\label{tab:ca_fa_comparison}
\renewcommand\tabcolsep{2.8pt}
\renewcommand\arraystretch{1.2}
\centering
\resizebox{0.49\textwidth}{!}{%
\begin{tabular}{l|cc|cc}
\toprule
        Samples       & \multicolumn{2}{c|}{CA} & \multicolumn{2}{c}{FA} \\ \cline{2-5} 
                 & Merged    & Closed     & Merged    & Closed     \\ \midrule
Positive (consistency)       & 2,089      & 820        & 2,238      & 861        \\
Negative (inconsistency)       & 501       & 135        & 352       & 94         \\ 
\bottomrule
\end{tabular}
}

\end{table}

\subsection{Metrics}
\begin{itemize}
    \item \textbf{F1-Score and Recall.} We utilized the F1-Score and recall to evaluate our method's effectiveness on tasks~\textit{CA} and \textit{FA}. The F1-Score, a balance between precision and recall, is crucial for distinguishing between false positives and negatives. Recall measures the proportion of actual positives correctly identified~\cite{hossin2015review}.
    \item \textbf{Edit Progress (EP).} EP evaluates the improvement in code transitioning from erroneous to correct by measuring the reduction in edit distance between the original code and the prediction on task~\textit{CR}. A higher EP indicates better efficiency in code generation~\cite{dibia2022aligning,elgohary-etal-2021-nl,zhou2023generation}.
    \item \textbf{Hit Rate (Rate)} We also use hit rate to evaluate the rate of confirmed vulnerability issues out of the found issues by approaches on task~\textit{VA}.
\end{itemize}

\begin{table*}[!t]
\caption{The number of vulnerabilities found by \tool{} and other approaches. As described in Appendix-Section~\ref{sec:dataset}, we have 3,545 items to evaluate. 
Rate$_{cr}$ represents the confirmed number divided by the number of findings while 
Rate$_{ca}$ is the confirmed number divided by the total evaluated number. 
\tool{}$_{w/o}$ indicates the version without QA-Checker. }
\label{tab:agent-va-finding}
\centering
\resizebox{0.9\linewidth}{!}{
\begin{threeparttable}
\begin{tabular}{l|cccllcc}
\toprule
 &  CodeBERT & GPT-3.5 & GPT-4.0 & COT & ReAct & \tool{} &\tool{}$_{w/o}$  \\ 
\midrule
Find        & 1,063    & 864         & 671         & 752     & 693     & 483                  & 564                              \\
Confirm     & 212      & 317         & 345         & 371     & 359     & 449                  & 413                              \\
Rate$_{cr}$ & 19.94\%  & 36.69\%     & 51.42\%     & 49.34\% & 51.80\% & \colorbox{gray!30}{92.96\%}            & 73.23\%                          \\
Rate$_{ca}$ & 5.98\%   & 8.94\%      & 9.73\%      & 10.46\% & 10.13\% & \colorbox{gray!30}{12.67\%}            & 11.65\%                          \\ 
\bottomrule
\end{tabular}%
{The values in gray (\colorbox{gray!30}{nn.nn}) denote the greatest values for the confirmed number of vulnerabilities and the rates.}
\end{threeparttable}
}
\end{table*}

\subsection{State-of-the-Art Tools and Models}
Our study evaluates various tools and models for code revision and modeling. \textbf{Trans-Review}~\cite{tufano2021towards} employs src2abs for code abstraction, effectively reducing vocabulary size. \textbf{AutoTransform}~\cite{thongtanunam2022autotransform} uses Byte-Pair Encoding for efficient vocabulary management in pre-review code revision. \textbf{T5-Review}~\cite{tufano2022using} leverages the T5 architecture, emphasizing improvement in code review through pre-training on code and text data. In handling both natural and programming languages, \textbf{CodeBERT}~\cite{codebert} adopts a bimodal approach, while \textbf{GraphCodeBERT}~\cite{GraphCodeBERT} incorporates code structure into its modeling. \textbf{CodeT5}~\cite{codeT5}, based on the T5 framework, is optimized for identifier type awareness, aiding in generation-based tasks. Additionally, we compare these tools with \textbf{GPT}~\cite{chatgpt} by OpenAI, notable for its human-like text generation capabilities in natural language processing. Finally, we involve \textbf{COT}~\cite{wei2022chain} and \textbf{ReAct}~\cite{yao2022react}, of which \textbf{COT} is a method where language models are guided to solve complex problems by generating and following a series of intermediate reasoning steps and \textbf{ReAct} synergistically enhances language models by interleaving reasoning and action generation, improving task performance and interpretability across various decision-making and language tasks.

\section{Experimental Result Analysis}
This section discusses the performance of \tool{} in the four tasks considered for our experiments. 
In Appendix~Section~\ref{sec:capana}, we provide further analyses: we discuss the difference in the execution time of \tool{} in different languages and perform a capability analysis between \tool{} and recent approaches.

\subsection{Vulnerability Analysis} 
\label{sec:va}

Compared to \textit{CA} and \textit{FA}, \textit{VA} is a more complex code review subtask, covering more than 25 different aspects (please see the Appendix-Section~\ref{sec:vulreason}), including buffer overflows, sensitive data exposure, configuration errors, data leakage, etc. 
Vulnerability analysis being a costly, time-consuming, resource-intensive and sensitive activity, only a low proportion of commits are labeled. We therefore propose a proactive method for data annotion: we execute \tool{} on the 3,545 samples (covering nine languages) and manual verify the identified cases to build a ground truth. Then, we applied CodeBERT~\cite{codebert} and GPT on the dataset with the task of vulnerability binary prediction.

\paragraph{Comparison} 
As shown in Table~\ref{tab:agent-va-finding},  \tool{} successfully identified 483 potential vulnerabilities within a data set of 3,545 samples, with an impressive 449 of these finally confirmed as high-risk vulnerabilities\footnote{The verification process involved a rigorous manual examination, extending beyond 120 working hours. Each sample being validated by at least 2 people: a researcher and an engineer}. CodeBERT, a key pre-trained model for code-related tasks, with its parameters frozen for this experiment, initially identified 1,063 items as vulnerable, yet only 212 passed the stringent verification criteria. Similar trends were observed with GPT-3.5 and GPT-4.0, which confirmed 317 and 345 vulnerabilities out of 864 and 671 identified items, respectively. These outcomes are further quantified by the confirmation rates (Rate$_{cr}$) of 19.94\% for CodeBERT, 36.69\% for GPT-3.5, and 51.42\% for GPT-4.0, while \tool{} demonstrated a remarkable Rate$_{cr}$ of 92.96\%. Additionally, the analysis of confirmed vulnerabilities against all analyzed items (Rate$_{ca}$) yielded 5.98\%, 8.94\%, 9.73\%, and 12.67\% for CodeBERT, GPT-3.5, GPT-4.0, and \tool{}, respectively. 
Evidently, Table~\ref{tab:agent-va-finding} not only highlights \tool{}'s high precision in identifying vulnerable commits but also reveals the progressive improvement from GPT-3.5 to GPT-4.0, likely due to the latter's capacity to handle longer input sequences, with token limits of 4,096 and 32,768, respectively. The integration of sophisticated algorithms like CoT and QA-Checker in \tool{} has significantly enhanced its capabilities in vulnerability detection, surpassing the individual input-output efficiencies of GPT and CodeBERT. Appendix-Sections~\ref{sec:diff} and ~\ref{sec:abs} highlight further details regarding the importance of the QA-checker.  Moreover, more experimental results in 9 languages are accessible in Appendix-Section~\ref{sec:vadetail}.

In addition, the analysis of vulnerabilities identified by various models reveals interesting overlaps
among the models.
CodeBERT confirmed 212 vulnerabilities, whereas GPT-3.5, GPT-4.0, and \tool{} confirmed 317, 345, and 449 vulnerabilities, respectively. 
Notably, the intersection of vulnerabilities confirmed by CodeBERT and GPT-3.5 is 169, indicating a substantial overlap in their findings. 
Similarly, the intersection between CodeBERT and GPT-4.0 is 170, while a larger overlap of 212 vulnerabilities is observed between GPT-3.5 and GPT-4.0. 
The combined intersection among CodeBERT, GPT-3.5, and GPT-4.0 is 137, underscoring the commonalities in vulnerabilities detected across these models. 
Further, the intersections of vulnerabilities confirmed by CodeBERT, GPT-3.5, and GPT-4.0 with \tool{} are 212, 317, and 334, respectively, highlighting the comprehensive coverage and detection capabilities of \tool{}.

\begin{figure}[H]
    \centering
    \includegraphics[width=0.65\linewidth]{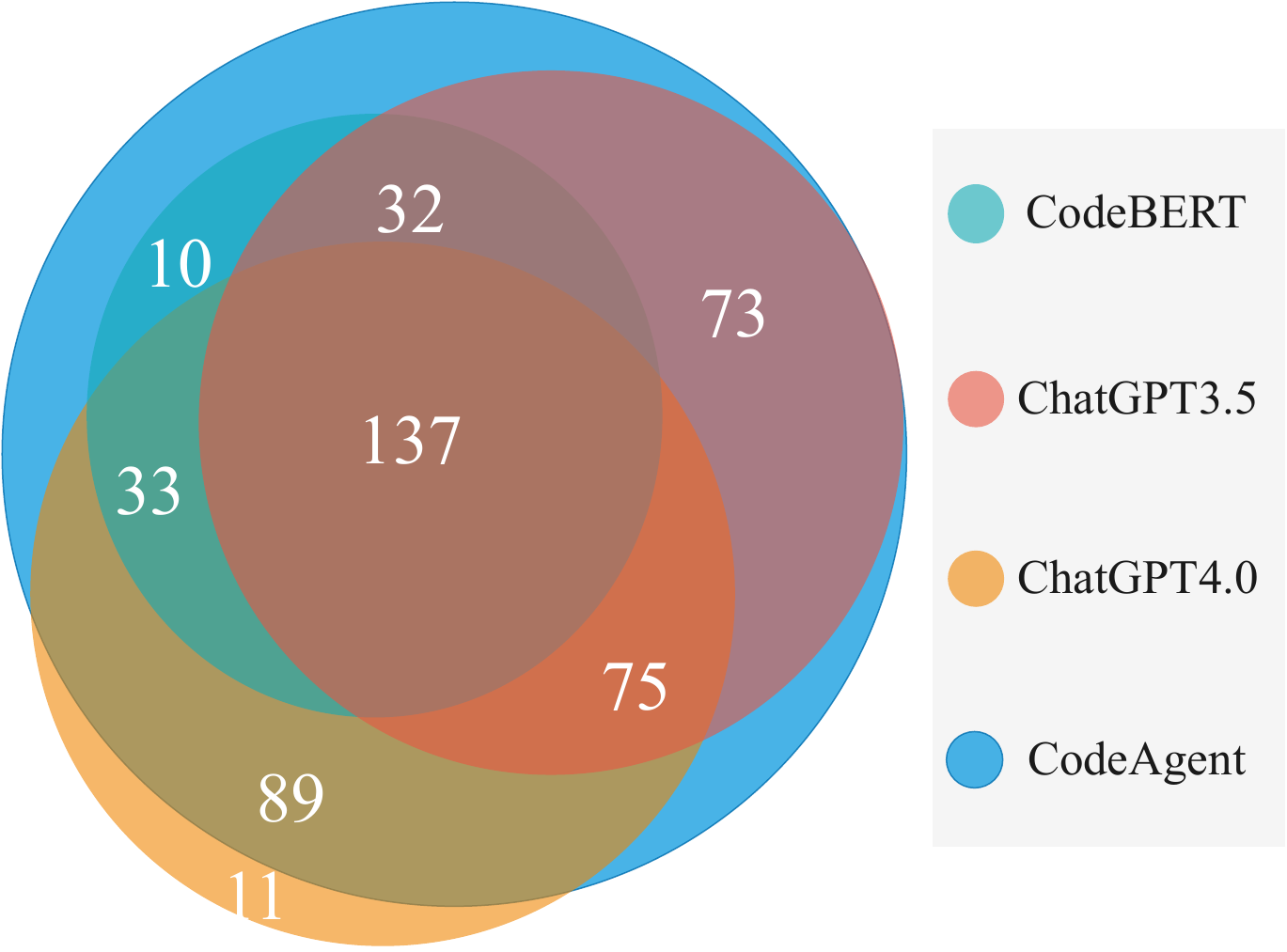}
    \caption{Overlap of vulnerability detection by CodeBERT, GPT-3.5, GPT-4.0, and \tool{}.}
    \label{fig:venn}
\end{figure}

\paragraph{Ablation Study.}

As shown in Table~\ref{tab:agent-va-finding}, we conducted an ablation study to evaluate the effectiveness of the QA-Checker in \tool{}. 
Specifically, we created a version of our tool without the QA-Checker, referred to as \tool{}$_{w/o}$. 
We then compared this version to the full version of \tool{} that includes the QA-Checker.
The results demonstrate that \tool{}$_{w/o}$ is substantially less effective in identifying vulnerable issues, yielding lower hit rates (Rate$_{cr}$ and Rate$_{ca}$). 
This reduction in performance highlights the critical role of the QA-Checker in enhancing \tool{}'s overall effectiveness. More detailed information about the ablation study can be found in Appendix-Section~\ref{sec:abs}.

\subsection{Consistency and Format Detection} 
\label{sec:cafa}
In this section, we will discuss the performance of \tool{} and baselines on metrics like the F1-Score and recall score of task \textit{CA} and \textit{FA}. For \textit{CA} and \textit{FA}, the dataset we have is shown in Table~\ref{tab:ca_fa_comparison} and more detailed data information is shown in Figure~\ref{fig:datadetail} in Appendix.



\paragraph{Code Change and Commit Message Consistency Detection.}
As illustrated in Table~\ref{tab:ca_mergeclose_exp}, we assess the efficacy of \tool{} in detecting the consistency between code changes and commit messages, contrasting its performance with other prevalent methods like CodeBERT, GPT-3.5, and GPT-4.0. This evaluation specifically focuses on merged and closed commits in nine languages. In particular, \tool{} exhibits remarkable performance, outperforming other methods in both merged and closed scenarios. In terms of Recall, \tool{} achieved an impressive 90.11\% for merged commits and 87.15\% for closed ones, marking a considerable average improvement of 5.62\% over the other models. Similarly, the F1-Score of \tool{} stands at 93.89\% for merged and 92.40\% for closed commits, surpassing its counterparts with an average improvement of 3.79\%. More comparable details in different languages are shown in Appendix-Section.~\ref{sec:cafadetailedexp}.

\begin{table}[ht]
\centering
\caption{Comparison of \tool{} with other methods on merged and closed commits across 9 languages on \textbf{CA task}. `Imp' represents the improvement.}
\resizebox{0.49\textwidth}{!}{%
\begin{tabular}{c|c|c|c|c|c|c|c}
\toprule
     \cellcolor{gray!50} \textbf{Merged }         & CodeBERT & GPT-3.5 & GPT-4.0 & COT & ReAct & CodeAgent & Imp (pp) \\ \midrule
Recall       & 63.64     & 80.08        & 84.27        & 80.73 & 82.04 & \cellcolor{gray!30}90.11 & 5.84        \\ \hline
F1           & 75.00     & 87.20        & 90.12        & 87.62 & 88.93 & \cellcolor{gray!30}93.89 & 3.77        \\ 
\bottomrule
\end{tabular}}

\resizebox{0.49\textwidth}{!}{%
\begin{tabular}{c|c|c|c|c|c|c|c}
\toprule
     \cellcolor{gray!50} \textbf{Closed }         & CodeBERT & GPT-3.5 & GPT-4.0 & COT & ReAct & CodeAgent & Imp (pp) \\ \midrule
Recall       & 64.80     & 79.05        & 81.75        & 81.77 &83.42 & \cellcolor{gray!30}87.15 & 5.21        \\ \hline
F1           & 77.20     & 87.35        & 89.61        & 89.30 & 89.81 & \cellcolor{gray!30}92.40 & 3.35        \\ 
\bottomrule
\end{tabular}}

\resizebox{0.49\textwidth}{!}{%
\begin{tabular}{c|c|c|c|c|c|c|c}
\toprule
     \cellcolor{gray!50} \textbf{Average }         & CodeBERT & GPT-3.5 & GPT-4.0 & COT & ReAct & CodeAgent & Imp (pp) \\ \midrule
Recall       & 64.22     & 79.57        & 83.01        & 81.25 & 82.73 & \cellcolor{gray!30}88.63 & 5.62        \\ \hline
F1           & 76.01     & 87.28        & 89.61       & 88.46 & 89.37 & \cellcolor{gray!30}93.16 & 3.79        \\ 
\bottomrule
\end{tabular}}
\label{tab:ca_mergeclose_exp}
\end{table}

\paragraph{Format Consistency Detection.}
In our detailed evaluation of format consistency between commits and original files, \tool{}'s performance was benchmarked against established models like CodeBERT and GPT variants across nine different languages. This comparative analysis, presented in Table~\ref{tab:fa_mergeclose_exp}, was centered around pivotal metrics such as Recall and F1-Score. \tool{} demonstrated a significant edge over the state-of-the-art, particularly in the merged category, with an impressive Recall of 89.34\% and an F1-Score of 94.01\%. These figures represent an average improvement of 10.81\% in Recall and 6.94\% in F1-Score over other models. 
In the closed category, \tool{} continued to outperform, achieving a Recall of 89.57\% and an F1-Score of 94.13\%, surpassing its counterparts with an improvement of 15.56\% in Recall and 9.94\% in F1-Score. 
The overall average performance of \tool{} further accentuates its superiority, with a Recall of 89.46\% and an F1-Score of 94.07\%, marking an average improvement of 13.39\% in Recall and 10.45\% in F1-Score. These results underscore \tool{}'s exceptional capability in accurately detecting format consistency between commits and their original files.

\begin{table}[ht]
\centering
\caption{Comparison of \tool{} with other methods on merged and closed commits across the 9 languages on \textbf{FA task}. `Imp' represents the improvement.}
\resizebox{0.49\textwidth}{!}{%
\begin{tabular}{c|c|c|c|c|c|c|c}
\toprule
\cellcolor{gray!50} \textbf{Merged} & CodeBERT & GPT-3.5 & GPT-4.0 & COT & ReAct &  \tool{} & Imp (pp) \\ \midrule
Recall & 60.59 & 60.72 & 78.53 & 70.39 & 71.21 & \cellcolor{gray!30}89.34 & 10.81 \\ \hline
F1 & 74.14 & 74.88 & 87.07 & 80.69 & 82.18 & \cellcolor{gray!30}94.01 & 6.94 \\ 
\bottomrule
\end{tabular}}

\resizebox{0.49\textwidth}{!}{%
\begin{tabular}{c|c|c|c|c|c|c|c}
\toprule
\cellcolor{gray!50} \textbf{Closed} & CodeBERT & GPT-3.5 & GPT-4.0 & COT & ReAct &  \tool{} & Imp (pp) \\ \midrule
Recall & 69.95 & 73.61 & 68.46  & 73.39 & 74.01 & \cellcolor{gray!30}89.57 & 15.56 \\ \hline
F1 & 80.49 & 84.19 & 80.16 & 83.65 & 83.90 & \cellcolor{gray!30}94.13 & 9.94 \\ 
\bottomrule
\end{tabular}}

\resizebox{0.49\textwidth}{!}{%
\begin{tabular}{c|c|c|c|c|c|c|c}
\toprule
\cellcolor{gray!50} \textbf{Average} & CodeBERT & GPT-3.5 & GPT-4.0 & COT & ReAct & \tool{} & Imp (pp) \\ \midrule
Recall & 65.27 & 67.17 & 73.50 & 71.89 & 72.61 & \cellcolor{gray!30}89.46 & 15.96 \\ \hline
F1 & 77.32 & 79.54 & 83.62 & 82.17 & 83.04 & \cellcolor{gray!30}94.07 & 10.45 \\ 
\bottomrule
\end{tabular}}
\label{tab:fa_mergeclose_exp}
\vspace{-0.5cm}
\end{table}

\subsection{Code Revision} \label{sec:cr}
We evaluate the effectiveness of \tool{} in revision suggestion (i.e., bug fixing) based on Edit Progress (EP) metric. We consider Trans-Review, AutoTransform, T5-Review, CodeBERT, GraphCodeBERT, CodeT5 as comparable state of the art. As detailed in Table~\ref{tab:coderevision}, these approaches exhibit a varied performance across different datasets. In particular, \tool{} shows remarkable performance in the T5-Review dataset, achieving the highest EP of 37.6\%. This is a significant improvement over other methods, which underlines the effectiveness of \tool{} in handling complex code revision tasks. Furthermore, with an average EP of 31.6\%, \tool{} consistently outperforms its counterparts, positioning itself as a leading solution in automated code revision. Its ability to excel in the T5-Review, a challenging benchmark data, indicates a strong capability to address complex bugs. In addition, its overall average performance surpasses other state-of-the-art models, highlighting its robustness and reliability. 

\begin{table}[htbp]
\caption{Experimental Results for the Code Revision (\textbf{CR task}) of \tool{} and the state-of-the-art works. Bold indicates the best performers.}
\label{tab:coderevision}
\centering
\resizebox{0.49\textwidth}{!}{%
\begin{tabular}{@{}l|ccc|cc@{}}
\hline

\hline

\hline

\hline

\multirow{2}{*}{\textbf{\begin{tabular}[c]{@{}l@{}} Approach\\ \end{tabular}}} & \multicolumn{1}{c}{\textbf{$\text{Trans-Review}_{\text{data}}$}} & \multicolumn{1}{c}{\textbf{$\text{AutoTransform}_{\text{data}}$}} & \multicolumn{1}{c|}{\textbf{$\text{T5-Review}_{\text{data}}$}} & \multicolumn{1}{c}{\textbf{Average}} \\ \cmidrule(l){2-2} \cmidrule(l){3-3}  \cmidrule(l){4-4} \cmidrule(l){5-5}  
                                                                                         & \textbf{EP}         & \textbf{EP}          & \textbf{EP}         & \textbf{EP}       \\ \midrule
 \textbf{Trans-Review}                                        & -1.1\%              & -16.6\%             & -151.2\%           & -56.3\%           \\
 \textbf{AutoTransform}                                       & 49.7\%              & \cellcolor{gray!30}\textbf{29.9\%}     & 9.7\%              & 29.8\%   \\
 \textbf{T5-Review}                                           & -14.9\%             & -71.5\%             & 13.8\%             & -24.2\%           \\
\textbf{CodeBERT}                                                                    & 49.8\%              & -75.3\%             & 22.3\%             & -1.1\%            \\
\textbf{GraphCodeBERT}                                                               & \cellcolor{gray!30}\textbf{50.6\%}     & -80.9\%             & 22.6\%             & -2.6\%            \\
\textbf{CodeT5}                                                                      & 41.8\%              & -67.8\%             & 25.6\%    & -0.1\%            \\ 
\textbf{\tool{}}                                                                      & 42.7\%              & 14.4\%             & \cellcolor{gray!30}\textbf{37.6\%}    & \cellcolor{gray!30}\textbf{31.6\%}            \\ 
\hline

\hline

\hline

\hline
\end{tabular}}
\end{table}






\section{Related Work}
\noindent\textbf{Automating Code Review Activities.} 
Our work contributes to automating code review activities, focusing on detecting source code vulnerabilities and maintaining code consistency. 
Related studies include Hellendoorn et al.~\cite{hellendoorn2021towards}, who addressed code change anticipation, and Siow et al.~\cite{siow2020core}, who introduced CORE for code modification semantics. 
Hong et al.~\cite{hong2022commentfinder} proposed COMMENTFINDER for comment suggestions, while Tufano et al.~\cite{tufano2021towards} and Li et al.~\cite{li2022codereviewer} developed tools for code review automation using models like T5CR and CodeReviewer, respectively. 
Recently, Lu et al.~\cite{lu2023llama} incorporated large language models for code review, enhancing fine-tuning techniques.

\noindent\textbf{Collaborative AI.} 
Collaborative AI, involving AI systems working towards shared goals, has seen advancements in multi-agent LLMs~\cite{talebirad2023multiagent, qian2023communicative}, focusing on collective thinking, conversation dataset curation~\cite{wei2023multi, li2023camel}, and sociological phenomenon exploration~\cite{park2023generative}. 
Research by Akata et al.~\cite{akata2023playing} and Cai et al.~\cite{cai2023large} further explores LLM cooperation and efficiency. 
However, there remains a gap in integrating these advancements with structured software engineering practices~\cite{li2023camel, qian2023communicative}, a challenge our approach addresses by incorporating advanced human processes in multi-agent systems.
For a complete overview of related work, please refer to our Appendix-Section~\ref{sec:completerw}.

\section{Conclusion}

In this paper, we introduced \tool{}, a novel multi-agent framework that automates code reviews.  \tool{} leverages its novel QA-Checker system to maintain focus on the review's objectives and ensure alignment.  Our experiments demonstrate \tool{}'s effectiveness in detecting vulnerabilities, enforcing code-message consistency, and promoting uniform code style.  Furthermore, \tool{} outperforms existing state-of-the-art solutions in code revision suggestions.  By incorporating human-like conversational elements and considering the specific characteristics of code review, \tool{} significantly improves both efficiency and accuracy.  We believe this work opens exciting new avenues for research and collaboration practices in software development.
\newpage


\section*{Limitations}
Firstly, the generalizability of the system across different software development environments or industries may require further validation and testing. While the system has shown promising results in the provided datasets, its applicability to other contexts remains uncertain without additional empirical evidence. This limitation suggests that the findings may not be fully transferable to all settings within the software development domain. Secondly, the baseline test used in the study might be insufficient. The current testing approach may not fully capture the system's performance, particularly in edge cases or more complex scenarios. This could result in an overestimation of the system's capabilities and an underestimation of its limitations. Further, more comprehensive testing is needed to establish a more robust baseline and to ensure that the system performs reliably across a wider range of conditions.
\section*{Ethics Statements}
This study was conducted in compliance with ethical guidelines and standards for research. The research did not involve human participants, and therefore, did not require informed consent or ethical review from an institutional review board. All data used in this study were publicly available, and no personal or sensitive information was accessed or processed. The development and evaluation of the \tool{} system were performed with a focus on transparency, reproducibility, and the potential positive impact on the software development community. 
\bibliography{custom}

\appendix

\appendixwithtoc
\newpage

\section{Details of QA-Checker Algorithm}\label{sec:qaalgorithm}

\begin{lemma}
    Let \( \mathcal{Q}(Q_i, A_i) \) denote the quality assessment function of the QA-Checker for the question-answer pair \( (Q_i, A_i) \) in a conversation at the \( i \)-th iteration. Assume \( \mathcal{Q} \) is twice differentiable and its Hessian matrix \( H(\mathcal{Q}) \) is positive definite. If the QA-Checker modifies the question \( Q_i \) to \( Q_{i+1} \) by attaching an additional instruction \( aai_i \), and this leads to a refined answer \( A_{i+1} \), then the sequence \( \{(Q_i, A_i)\} \) converges to an optimal question-answer pair \( (Q^*, A^*) \), under specific regularity conditions.
\end{lemma}

\begin{proof}
    The QA-Checker refines the question and answers using the rule:
    \begin{align*}
        Q_{i+1} &= Q_i + aai_i, \\
        A_{i+1} &= A_i - \alpha H(\mathcal{Q}(Q_i, A_i))^{-1} \nabla \mathcal{Q}(Q_i, A_i),
    \end{align*}
    where \( \alpha \) is the learning rate. To analyze convergence, we consider the Taylor expansion of \( \mathcal{Q} \) around \( (Q_i, A_i) \):
    \begin{small}
        \begin{align*}
        \mathcal{Q}(Q_{i+1}, A_{i+1}) &\approx \mathcal{Q}(Q_i, A_i) + \nabla \mathcal{Q}(Q_i, A_i)\\
        &\cdot (Q_{i+1} - Q_i, A_{i+1} - A_i) \\
        &\quad + \frac{1}{2} (Q_{i+1} - Q_i, A_{i+1} - A_i)^T \\
        &H(\mathcal{Q}(Q_i, A_i)) (Q_{i+1} - Q_i, A_{i+1} - A_i).
    \end{align*}
    \end{small}
    Substituting the update rule and rearranging, we get:
    \begin{small}
        \begin{align*}
        \mathcal{Q}(Q_{i+1}, A_{i+1}) &\approx \mathcal{Q}(Q_i, A_i) \\
        &- \alpha \nabla \mathcal{Q}(Q_i, A_i)^T H(\mathcal{Q}(Q_i, A_i))^{-1} \\
        &\nabla \mathcal{Q}(Q_i, A_i) \\
        & + \frac{\alpha^2}{2} \nabla \mathcal{Q}(Q_i, A_i)^T H(\mathcal{Q}(Q_i, A_i))^{-1} \\
        &\nabla \mathcal{Q}(Q_i, A_i).
    \end{align*}
    \end{small}
    For sufficiently small \( \alpha \), this model suggests an increase in \( \mathcal{Q} \), implying convergence to an optimal question-answer pair \( (Q^*, A^*) \) as \( i \to \infty \). The convergence relies on the positive definiteness of \( H(\mathcal{Q}) \) and the appropriate choice of \( \alpha \), ensuring each iteration moves towards an improved quality of the question-answer pair.
\end{proof}

In practical terms, this lemma and its proof underpin the QA-Checker's ability to refine answers iteratively. The QA-Checker assesses the quality of each answer concerning the posed question, employing advanced optimization techniques that are modeled by the modified Newton-Raphson method to enhance answer quality. This framework ensures that, with each iteration, the system moves closer to the optimal answer, leveraging both first and second-order derivatives for efficient and effective learning.

\paragraph{Further Discussion}
The QA-Checker computes \( \mathcal{Q}(Q_i, A_i) \) at each iteration \( i \) and compares it to a predefined quality threshold \( \tau \). If \( \mathcal{Q}(Q_i, A_i) < \tau \), the QA-Checker generates an additional instruction \( aai_i \) to refine the question to \( Q_{i+1} = Q_i + aai_i \), prompting the agents to generate an improved answer \( A_{i+1} \).
\vspace{+0.5em}

First, we assume that the quality assessment function \( \mathcal{Q}(Q_i, A_i) \) is twice differentiable with respect to the question \( Q_i \). This assumption is reasonable given the smooth nature of the component functions (relevance, specificity, and coherence) and the use of continuous word embeddings. Next, we apply the second-order Taylor approximation to \( \mathcal{Q}(Q_{i+1}, A_{i+1}) \) around the point \( (Q_i, A_i) \):
\begin{small}
   \begin{align}
    \mathcal{Q}(Q_{i+1}, A_{i+1}) &\approx \mathcal{Q}(Q_i, A_i) + \nabla \mathcal{Q}(Q_i, A_i)^T \Delta Q_i \nonumber \\ &+ \frac{1}{2} \Delta Q_i^T H(\mathcal{Q}(Q_i, A_i)) \Delta Q_i + R_2(\Delta Q_i)\nonumber
\end{align} 
\end{small}

where \( \Delta Q_i = Q_{i+1} - Q_i \), \( H(\mathcal{Q}(Q_i, A_i)) \) is the Hessian matrix of \( \mathcal{Q} \) evaluated at \( (Q_i, A_i) \), and \( R_2(\Delta Q_i) \) is the remainder term.

Assuming that the remainder term \( R_2(\Delta Q_i) \) is negligible and that the Hessian matrix is positive definite, we can approximate the optimal step \( \Delta Q_i^* \) as:

\[
\Delta Q_i^* \approx -H(\mathcal{Q}(Q_i, A_i))^{-1} \nabla \mathcal{Q}(Q_i, A_i).
\]

Substituting this approximation into the Taylor expansion and using the fact that \( Q_{i+1} = Q_i + \alpha \Delta Q_i^* \) (where \( \alpha \) is the learning rate), we obtain:
\begin{small}
    \begin{align}
    \mathcal{Q}(Q_{i+1}, A_{i+1}) \approx& \mathcal{Q}(Q_i, A_i) - \alpha \nabla \mathcal{Q}(Q_i, A_i)^T \nonumber\\& \cdot H(\mathcal{Q}(Q_i, A_i))^{-1} \nabla \mathcal{Q}(Q_i, A_i) \nonumber\\ & + \frac{\alpha^2}{2} \nabla \mathcal{Q}(Q_i, A_i)^T H(\mathcal{Q}(Q_i, A_i))^{-1}\nonumber\\& \cdot \nabla \mathcal{Q}(Q_i, A_i).\nonumber
\end{align}
\end{small}

The assumptions of twice differentiability, negligible remainder term, and positive definite Hessian matrix provide a more solid foundation for the approximation in Lemma 3.1. For sufficiently small \( \alpha \), this approximation suggests an increase in \( \mathcal{Q} \), implying convergence to an optimal question-answer pair \( (Q^*, A^*) \) as \( i \to \infty \). The convergence relies on the positive definiteness of \( H(\mathcal{Q}) \) and the appropriate choice of \( \alpha \), ensuring each iteration moves towards an improved quality of the question-answer pair.

The quality assessment function \( \mathcal{Q} \) used by the QA-Checker is defined as:

\begin{align}
    \mathcal{Q}(Q_i, A_i) &= \alpha \cdot \text{Relevance}(Q_i, A_i)\nonumber\\ &\quad+ \beta \cdot \text{Specificity}(A_i)\nonumber\\ &\quad+ \gamma \cdot \text{Coherence}(A_i)\nonumber
\end{align}

where:

\begin{itemize}
    \item \( Q_i \) and \( A_i \) represent the question and answer at the \( i \)-th iteration of the conversation.
    \item \( \text{Relevance}(Q_i, A_i) \) measures how well the answer \( A_i \) addresses the key points and intent of the question \( Q_i \), computed as:
    \[
    \text{Relevance}(Q_i, A_i) = \frac{\vec{Q_i} \cdot \vec{A_i}}{|\vec{Q_i}| |\vec{A_i}|}
    \]
    where \( \vec{Q_i} \) and \( \vec{A_i} \) are vector representations of \( Q_i \) and \( A_i \).
    \item \( \text{Specificity}(A_i) \) assesses how specific and detailed the answer \( A_i \) is, calculated as:
    \begin{small}
        \begin{equation}
        A_i = \frac{\sum_{t \in \text{ContentWords}(A_i)} \text{TechnicalityScore}(t)}{\text{Length}(A_i)}\nonumber
        \end{equation}
    \end{small}

    where \( \text{ContentWords}(A_i) \) is the set of substantive content words in \( A_i \), \( \text{TechnicalityScore}(t) \) is a measure of how technical or domain-specific the term \( t \) is, and \( \text{Length}(A_i) \) is the total number of words in \( A_i \).
    \item \( \text{Coherence}(A_i) \) evaluates the logical flow and structural coherence of the answer \( A_i \), computed as:
    \begin{small}
        \begin{align}
        \text{Coherence}(A_i) =& \alpha \cdot \text{DiscourseConnectives}(A_i)\nonumber \\ &  + \beta \cdot \text{CoreferenceConsistency}(A_i)\nonumber \\ & + \gamma \cdot \text{AnswerPatternAdherence}(A_i)\nonumber
    \end{align}
    \end{small}

    where \( \text{DiscourseConnectives}(A_i) \) is the density of discourse connectives in \( A_i \), \( \text{CoreferenceConsistency}(A_i) \) measures the consistency of coreference chains in \( A_i \), and \( \text{AnswerPatternAdherence}(A_i) \) assesses how well \( A_i \) follows the expected structural patterns for the given question type.
\end{itemize}

\(\alpha\), \(\beta\), and \(\gamma\) are non-negative weights that sum to 1, with \(\alpha = \beta = \gamma\).

\section{Complete Related Work}\label{sec:completerw}

\noindent\textbf{Automating Code Review Activities} 
Our focus included detecting source code vulnerabilities, ensuring style alignment, and maintaining commit message and code consistency. 
Other studies explore various aspects of code review. 
Hellendoorn et al.~\cite{hellendoorn2021towards} addressed the challenge of anticipating code change positions. 
Siow et al.~\cite{siow2020core} introduced CORE, employing multi-level embeddings for code modification semantics and retrieval-based review suggestions. 
Hong et al.~\cite{hong2022commentfinder} proposed COMMENTFINDER, a retrieval-based method for suggesting comments during code reviews. 
Tufano et al.~\cite{tufano2021towards} designed T5CR with SentencePiece, enabling work with raw source code without abstraction. 
Li et al.~\cite{li2022codereviewer} developed CodeReviewer, focusing on code diff quality, review comment generation, and code refinement using the T5 model. 
Recently, large language models have been incorporated; Lu et al.~\cite{lu2023llama} fine-tuned LLama with prefix tuning for LLaMA-Reviewer, using parameter-efficient fine-tuning and instruction tuning in a code-centric domain.

\noindent\textbf{Collaborative AI} 
Collaborative AI refers to artificial intelligent systems designed to achieve shared goals with humans or other AI systems. 
Previous research extensively explores the use of multiple LLMs in collaborative settings, as demonstrated by Talebirad et al.~\cite{talebirad2023multiagent} and Qian et al.~\cite{qian2023communicative}. 
These approaches rely on the idea that inter-agent interactions enable LLMs to collectively enhance their capabilities, leading to improved overall performance. 
The research covers various aspects of multi-agent scenarios, including collective thinking, conversation dataset curation, sociological phenomenon exploration, and collaboration for efficiency. 
Collective thinking aims to boost problem-solving abilities by orchestrating discussions among multiple agents. 
Researchers like Wei et al.~\cite{wei2023multi} and Li et al.~\cite{li2023camel} have created conversational datasets through role-playing methodologies. 
Sociological phenomenon investigations, such as Park et al.~\cite{park2023generative}'s work, involve creating virtual communities with rudimentary language interactions and limited cooperative endeavors. 
In contrast, Akata et al.~\cite{akata2023playing} scrutinized LLM cooperation through orchestrated repeated games.
Collaboration for efficiency, proposed by Cai et al.~\cite{cai2023large}, introduces a model for cost reduction through large models as tool-makers and small models as tool-users.
Zhang et al.~\cite{zhang2023building} established a framework for verbal communication and collaboration, enhancing overall efficiency.
However, Li et al.~\cite{li2023camel} and Qian et al.~\cite{qian2023communicative}, presenting a multi-agent framework for software development, primarily relied on natural language conversations, not standardized software engineering documentation, and lacked advanced human process management expertise.
Challenges in multi-agent cooperation include maintaining coherence, avoiding unproductive loops, and fostering beneficial interactions. 
Our approach emphasizes integrating advanced human processes, like code review in software maintenance, within multi-agent systems.

\section{Experimental Details} \label{sec:expsetdetail}
In our work, the maximum number of conversation rounds is set as 10.
\subsection{Role Definition} \label{sec:roledefine}
Six roles are defined as shown in Figure~\ref{fig:roleplay}. 

\begin{figure*}[ht]
    \centering
    \includegraphics[width=\linewidth]{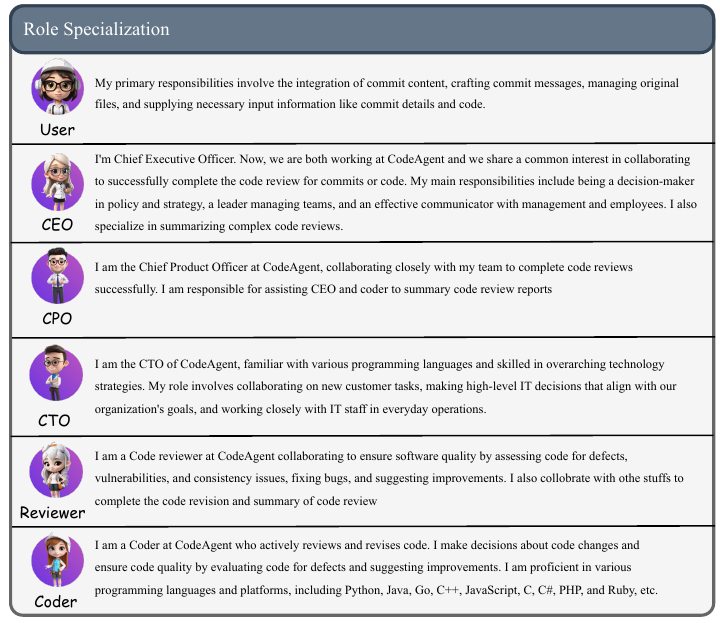}
    \caption{Specialization of six main characters in \tool{}.}
    \label{fig:roleplay}
\end{figure*}

Apart from that, for the QA-checker in \tool{}, we define an initial prompt for it, which is shown as follows:

\begin{tcolorbox}[colback=gray!20]
\includegraphics[height=2em]{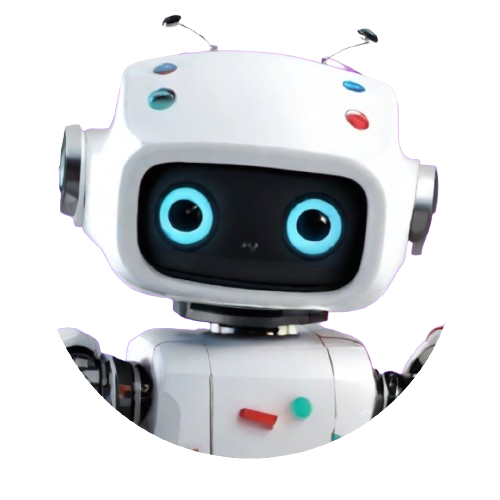} I'm the QA-Checker, an AI-driven agent specializing in ensuring quality and coherence in conversational dynamics, particularly in code review discussions at CodeAgent. My primary role involves analyzing and aligning conversations to maintain topic relevance, ensuring that all discussions about code commits and reviews stay focused and on track. As a sophisticated component of the AI system, I apply advanced algorithms, including Chain-of-Thought reasoning and optimization techniques, to evaluate and guide conversational flow. I am adept at identifying and correcting topic drifts, ensuring that every conversation adheres to its intended purpose. My capabilities extend to facilitating clear and effective communication between team members, making me an essential asset in streamlining code review processes and enhancing overall team collaboration and decision-making.
\end{tcolorbox}

\subsection{Execute Time Across Languages}

As depicted in the data, we observe a significant trend in the average execution time for code reviews in \tool{} across various programming languages. The analysis includes nine languages: Python, Java, Go, C++, JavaScript, C, C\#, PHP, and Ruby. For each language, the average execution time of code reviews for both merged and closed pull requests (PRs) is measured. The results, presented in Figure~\ref{fig:executetime}, indicate that, on average, the execution time for merged PRs is longer than that for closed PRs by approximately 44.92 seconds. This considerable time difference can be attributed to several potential reasons. One primary explanation is that merged PRs likely undergo a more rigorous and detailed review process. They are intended to be integrated into the main codebase, and as such, contributors might be requested to update their commits in the PRs more frequently to adhere to the project's high-quality standards. On the other hand, closed PRs, which are not meant for merging, might not require such extensive review processes, leading to shorter review times on average, which may also be the reason they are not merged into main projects. 
\begin{figure}[H]
    \centering
    \includegraphics[width=\linewidth]{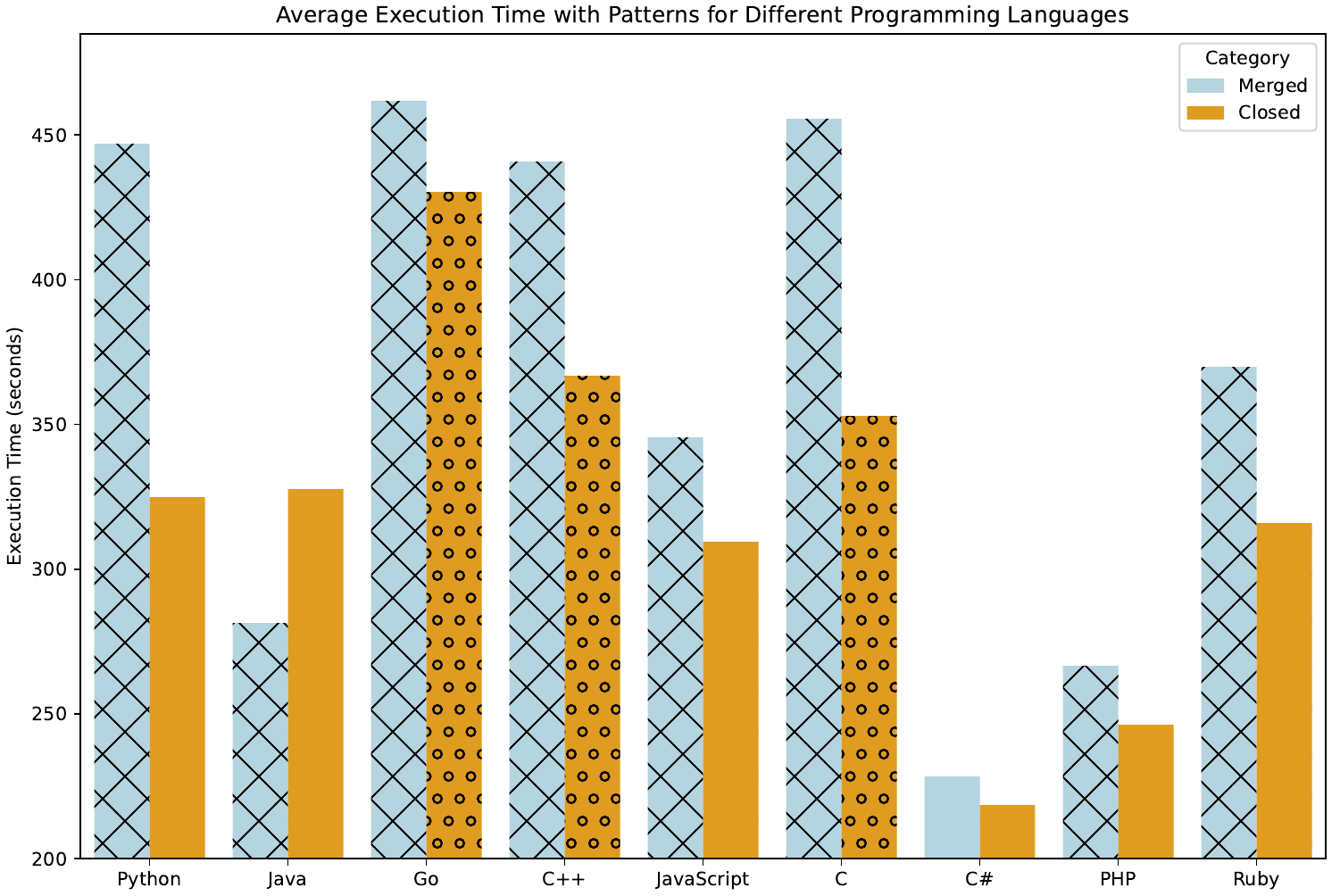}
    \caption{Execution time with \tool{} across different language (count unit: second).}
    \label{fig:executetime}
\end{figure}

\section{Comparative Analysis of QA-Checker AI System and Recursive Self-Improvement Systems} \label{sec:diff}
In this section, we will delve into the differences between QA-Checker and self-improvement systems~\cite{hong2023metagpt}, and underscore the importance of the QA-Checker in role conversations.

\subsection{Comparison Table}
We begin with a comparative overview presented in Table \ref{tab:comparison}.









\begin{table*}[]
\caption{Comparative Overview of QA-Checker AI System and Recursive Self-Improvement Systems}
\label{tab:comparison}
\resizebox{2\columnwidth}{!}{%
\begin{tabular}{l|l|l}
\hline
Feature/System &
  QA-Checker AI System &
  Recursive Self-Improvement System \\ \hline
Application Focus &
  \begin{tabular}[c]{@{}l@{}}Specialized for QA tasks with \\ precise task execution\end{tabular} &
  \begin{tabular}[c]{@{}l@{}}Broad scope, covering various dimensions like\\  software development and learning algorithms\end{tabular} \\ \hline
Learning Mechanism &
  \begin{tabular}[c]{@{}l@{}}Advanced optimization techniques \\ for iterative improvement in QA\end{tabular} &
  \begin{tabular}[c]{@{}l@{}}Multi-level learning: learning, meta-learning, \\ and recursive self-improvement\end{tabular} \\ \hline
Scope of Improvement &
  \begin{tabular}[c]{@{}l@{}}Focused on individual capability \\ in specific QA tasks\end{tabular} &
  \begin{tabular}[c]{@{}l@{}}Enhances the entire system, including multi-agent \\ interactions and communication protocols\end{tabular} \\ \hline
Experience Integration &
  \begin{tabular}[c]{@{}l@{}}Based on mathematical models \\ to optimize answer quality\end{tabular} &
  \begin{tabular}[c]{@{}l@{}}Utilizes experiences from past projects to improve \\ overall performance\end{tabular} \\ \hline
\end{tabular}%
}
\end{table*}

\subsection{Differences and Implications}
The key differences between these systems lie in their application scope, learning mechanisms, and improvement scopes. The QA-Checker is highly specialized, focusing on QA tasks with efficiency and precision. In contrast, recursive self-improvement systems boast a broader application range and adaptability, integrating experiences from diverse projects for systemic improvements.

\subsection{Importance of QA-Checker in Role Conversations}
In the context of role conversations, the QA-Checker plays a pivotal role. Its specialized nature makes it exceptionally adept at handling specific conversational aspects, such as accuracy, relevance, and clarity in responses. This specialization is crucial in domains where the quality of information is paramount, ensuring that responses are not only correct but also contextually appropriate and informative.

Furthermore, the efficiency of the QA-Checker in refining responses based on advanced optimization techniques makes it an invaluable tool in dynamic conversational environments. It can quickly adapt to the nuances of a conversation, providing high-quality responses that are aligned with the evolving nature of dialogue.

\subsection{Conclusion}
While recursive self-improvement systems offer broad adaptability and systemic learning, the QA-Checker stands out in its specialized role in QA tasks, particularly in role conversations. Its focused approach to improving answer quality and its efficiency in handling conversational nuances make it an essential component in AI-driven communication systems.

\section{Capabilities Analysis between \tool{} and Other Methods} \label{sec:capana}
Compared to open-source baseline methods such as AutoGPT and autonomous agents such as ChatDev and MetaGPT, \tool{} offers functions for code review tasks: consistency analysis, vulnerability analysis, and format analysis. As shown in Table~\ref{tab:cap}, our \tool{} encompasses a wide range of abilities to handle complex code review tasks efficiently. Incorporating the QA-Checker self-improved module can significantly improve the conversation generation between agents and contribute to the improvement of code review. Compared to COT, the difference and the advantages of \tool{} with QA-Checker are shown in Section~\ref{sec:diff}. 

\begin{table*}[htb]
\caption{Comparison of capabilities for \tool{} and other approaches. `\checkmark' indicates the presence of a specific feature in the corresponding framework, `\ding{55} is absence. ChatDev and MetaGPT are two representative multi-agent frameworks, GPT is a kind of single-agent framework, and CodeBert is a representative pre-trained model. }
\label{tab:cap}
\renewcommand\tabcolsep{2.8pt}
\renewcommand\arraystretch{1.2}
\resizebox{\textwidth}{!}{%
\begin{tabular}{l|cccccccc}
\hline

\hline

\hline

\hline
Approaches&  Consistency Analysis & Vulnerability Analysis & Format Analysis & Code Revision & COT & QA-Checker \\ 
\hline

\hline
ChatDev~\cite{qian2023communicative}& \ding{55} & \ding{55} & \ding{55} & \ding{55} & \checkmark & \ding{55}  \\ 
MetaGPT~\cite{hong2023metagpt}& \ding{55} & \ding{55}  & \ding{55} & \ding{55} & \checkmark & \ding{55} \\ 
GPT~\cite{chatgpt}& \checkmark & \checkmark  & \checkmark & \checkmark & \ding{55}& \ding{55} \\  
CodeBert~\cite{codebert}& \checkmark & \checkmark & \checkmark & \checkmark & \ding{55} & \ding{55} \\ 
\rowcolor[gray]{.9} 
\tool{}& \checkmark & \checkmark  & \checkmark & \checkmark & \checkmark & \checkmark \\  
\hline

\hline

\hline

\hline
\end{tabular}%
 }
\end{table*}

\section{Dataset} \label{sec:dataset}

\paragraph{Previous Dataset}
As shown in ~\citet{zhou2023generation}, our study incorporates three distinct datasets for evaluating the performance of \tool{}: $\text{Trans-Review}_\text{data}$, $\text{AutoTransform}_\text{data}$, and $\text{T5-Review}_\text{data}$. $\text{Trans-Review}_\text{data}$, compiled by Tufano et al.~\cite{tufano2021towards}, derives from Gerrit and GitHub projects, excluding noisy or overly lengthy comments and review data with new tokens in revised code not present in the initial submission. $\text{AutoTransform}_\text{data}$, collected by Thongtanunam et al.~\cite{thongtanunam2022autotransform} from three Gerrit repositories, comprises only submitted and revised codes without review comments. Lastly, $\text{T5-Review}_\text{data}$, gathered by Tufano et al.~\cite{tufano2022using} from Java projects on GitHub, filters out noisy, non-English, and duplicate comments. These datasets are employed for Code Revision Before Review (CRB) and Code Revision After Review (CRA) tasks, with the exception of $\text{AutoTransform}_\text{data}$ for CRA and Review Comment Generation (RCG) due to its lack of review comments.

\paragraph{New Dataset Design and Collection}
To enhance our model evaluation and avoid data leakage, we curated a new dataset, exclusively collecting data from repositories created after April 2023. This approach ensures the evaluation of our CodeAgent model on contemporary and relevant data, free from historical biases. The new dataset is extensive, covering a broad spectrum of software projects across nine programming languages. 

\begin{figure*}[h]
    \centering
    \subfigure[Positive and negative data of both merged and closed commits across 9 languages on \textit{CA} task (Sec~\ref{sec:taskdefinition}).]{
        \includegraphics[width=.48\linewidth]{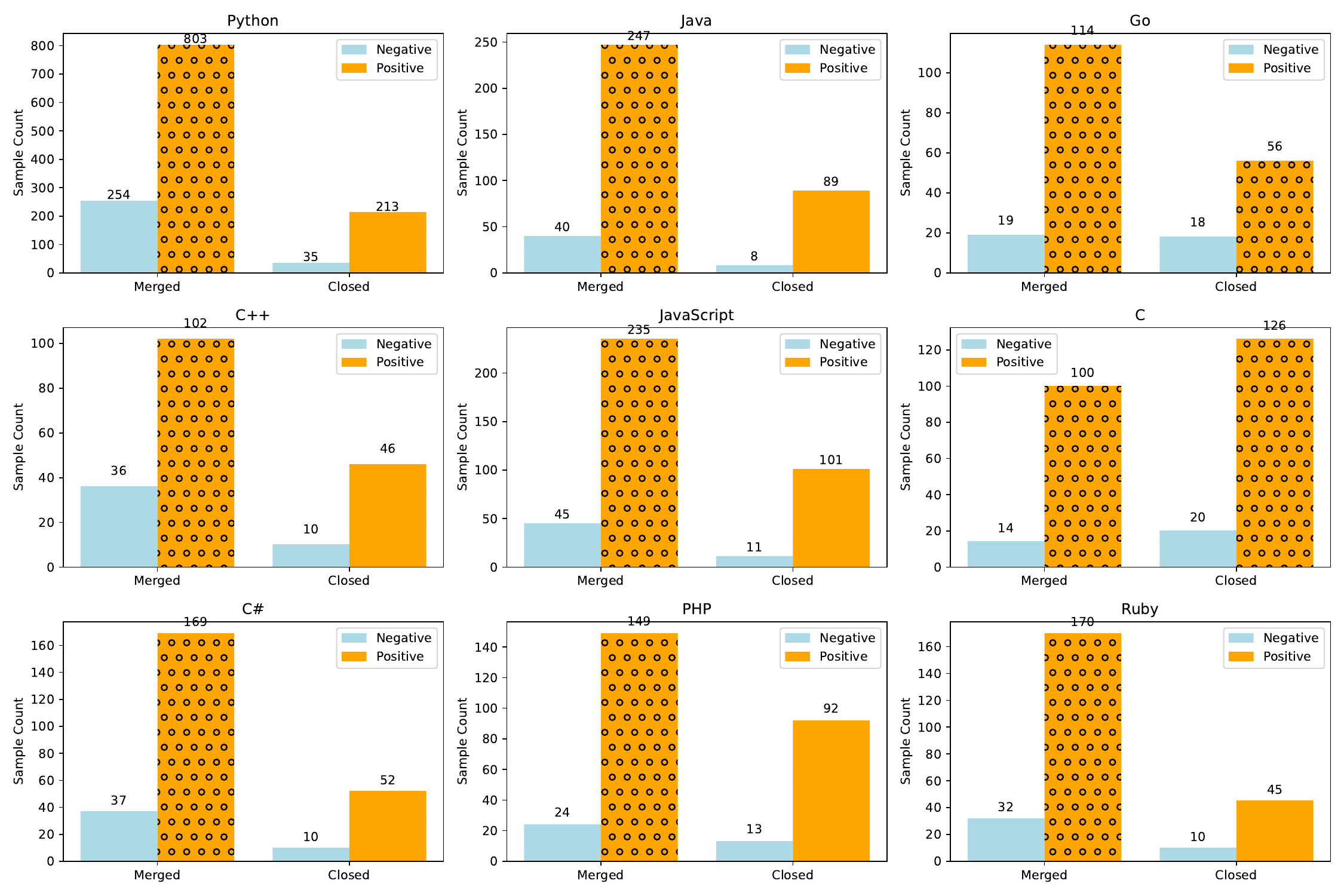}
        \label{fig:dataca}
    }%
    \subfigure[Positive and negative data of both merged and closed commits across 9 languages on \textit{FA} task (Sec~\ref{sec:taskdefinition}).]{
        \includegraphics[width=.48\linewidth]{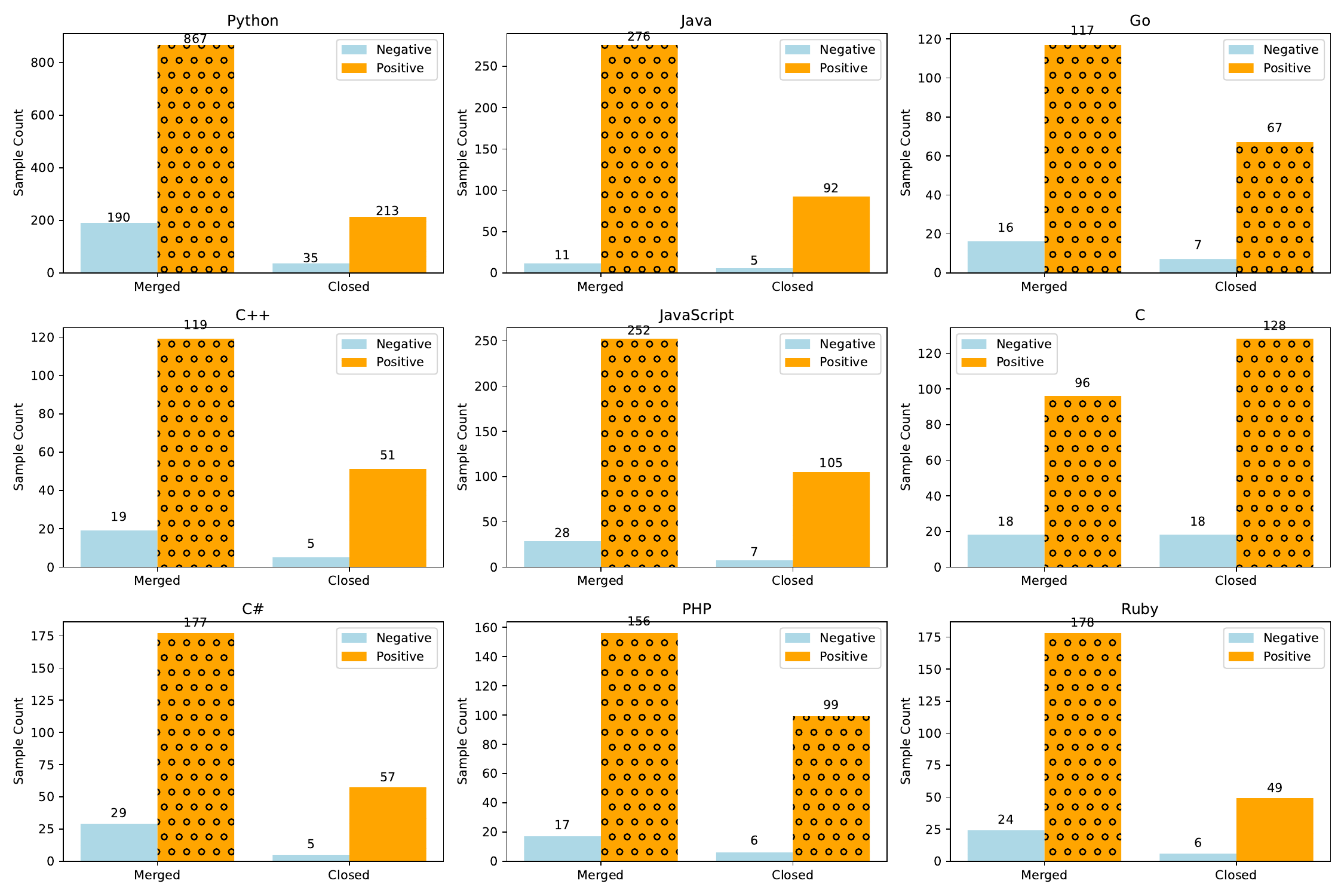}
        \label{fig:datafa}
    }
    \caption{Distribution of positive, negative of both merged and closed data across 9 languages, including `python', `java', `go', `c++', `javascript', `c', `c\#', `php', `ruby'.}
    \label{fig:datadetail}
\end{figure*}

    
    

\paragraph{Dataset Description}

\begin{figure}[htbp]
    \centering
    \includegraphics[width=0.45\textwidth]{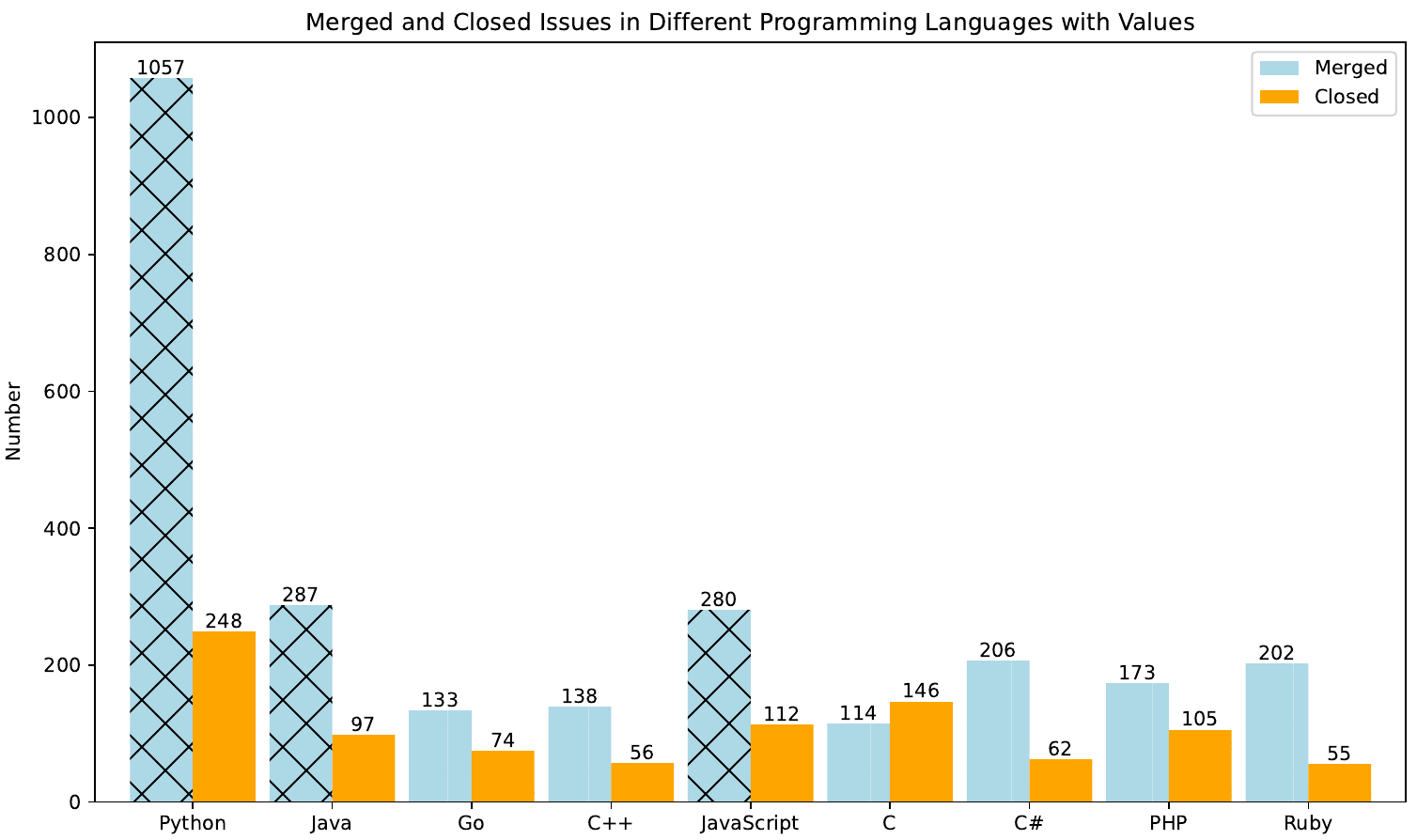} 
    \caption{Comparative Visualization of Merged and Closed Commit Counts Across Various Programming Languages}
    \label{fig:newdata}
\end{figure}
Our dataset, illustrated in Fig.~\ref{fig:newdata}, encapsulates a detailed analysis of consistency and format detection in software development, spanning various programming languages. It includes CA (consistency between commit and commit message (See Sec~\ref{sec:taskdefinition})) and FA (format consistency between commit and original (See Sec~\ref{sec:taskdefinition})) data, segmented into positive and negative samples based on the merged and closed status of pull requests. For example, in Python, the dataset comprises 254 merged and 35 closed negative CA samples, alongside 803 merged and 213 closed positive CA samples, with corresponding distributions for other languages like Java, Go, C++, and more. Similarly, the FA data follows this pattern of positive and negative samples across languages. Figure~\ref{fig:datadetail} graphically represents this data, highlighting the distribution and comparison of merged versus closed samples in both CA and FA categories for each language. This comprehensive dataset, covering over 3,545 commits and nearly 2,933 pull requests from more than 180 projects, was meticulously compiled using a custom crawler designed for GitHub API interactions, targeting post-April 2023 repositories to ensure up-to-date and diverse data for an in-depth analysis of current software development trends.


\begin{table}[ht!]
\centering
\caption{Statistics of Studied Datasets.}
\label{tab:datasets}
\begin{tabular}{@{}lrrr@{}}
\hline

\hline

\hline

\hline
\textbf{Dataset Statistics}                 & \multicolumn{1}{c}{\textbf{\#Train}} & \multicolumn{1}{c}{\textbf{\#Valid}} & \multicolumn{1}{c}{\textbf{\#Test}} \\ 
\midrule
\multicolumn{1}{l}{\textbf{Trans-Review}}  & \multicolumn{1}{r}{13,756}          & \multicolumn{1}{r}{1,719}           & \multicolumn{1}{r}{1,719}         \\ 
\midrule
\multicolumn{1}{l}{\textbf{AutoTransform}} & \multicolumn{1}{r}{118,039}         & \multicolumn{1}{r}{14,750}          & \multicolumn{1}{r}{14,750}        \\ 
\midrule
\multicolumn{1}{l}{\textbf{T5-Review}}     & \multicolumn{1}{r}{134,239}         & \multicolumn{1}{r}{16,780}          & \multicolumn{1}{r}{16,780}        \\ 
\hline

\hline

\hline

\hline
\end{tabular}
\end{table}

\section{Key Factors Leading to Vulnerabilities}\label{sec:vulreason}

The following table outlines various key factors that can lead to vulnerabilities in software systems, along with their descriptions. These factors should be carefully considered and addressed to enhance the security of the system.
\begin{table*}[]
    \centering
\begin{supertabular}{p{0.6cm}|p{5cm}|p{10cm}}
\hline

\hline

\hline

\hline
\textbf{No.} & \textbf{Vulnerability Factor} & \textbf{Description} \\
\hline

1 & Insufficient Input Validation & Check for vulnerabilities like SQL injection, Cross-Site Scripting (XSS), and command injection in new or modified code, especially where user input is processed. \\ \hline
2 & Buffer Overflows & Particularly in lower-level languages, ensure that memory management is handled securely to prevent overflows. \\ \hline
3 & Authentication and Authorization Flaws & Evaluate any changes in authentication and authorization logic for potential weaknesses that could allow unauthorized access or privilege escalation. \\ \hline
4 & Sensitive Data Exposure & Assess handling and storage of sensitive information like passwords, private keys, or personal data to prevent exposure. \\ \hline
5 & Improper Error and Exception Handling & Ensure that errors and exceptions are handled appropriately without revealing sensitive information or causing service disruption. \\ \hline
6 & Vulnerabilities in Dependency Libraries or Components & Review updates or changes in third-party libraries or components for known vulnerabilities. \\ \hline
7 & Cross-Site Request Forgery (CSRF) & Verify that adequate protection mechanisms are in place against CSRF attacks. \\ \hline
8 & Unsafe Use of APIs & Check for the use of insecure encryption algorithms or other risky API practices. \\ \hline
9 & Code Injection & Look for vulnerabilities related to dynamic code execution. \\ \hline
10 & Configuration Errors & Ensure that no insecure configurations or settings like open debug ports or default passwords have been introduced. \\ \hline
11 & Race Conditions & Analyze for potential data corruption or security issues arising from race conditions. \\ \hline
12 & Memory Leaks & Identify any changes that could potentially lead to memory leaks and resource exhaustion. \\ \hline
13 & Improper Resource Management & Check resource management, such as proper closure of file handles or database connections. \\ \hline
14 & Inadequate Security Configurations & Assess for any insecure default settings or unencrypted communications. \\ \hline
15 & Path Traversal and File Inclusion Vulnerabilities & Examine for risks that could allow unauthorized file access or execution. \\ \hline
16 & Unsafe Deserialization & Look for issues that could allow the execution of malicious code or tampering with application logic. \\ \hline
17 & XML External Entity (XXE) Attacks & Check if XML processing is secure against XXE attacks. \\ \hline
18 & Inconsistent Error Handling & Review error messages to ensure they do not leak sensitive system details. \\ \hline
19 & Server-Side Request Forgery (SSRF) & Analyze for vulnerabilities that could be exploited to attack internal systems. \\ \hline
20 & Unsafe Redirects and Forwards & Check for vulnerabilities leading to phishing or redirection attacks. \\ \hline
21 & Use of Deprecated or Unsafe Functions and Commands & Identify usage of any such functions and commands in the code. \\ \hline
22 & Code Leakages and Hardcoded Sensitive Information & Look for hardcoded passwords, keys, or other sensitive data in the code. \\ \hline
23 & Unencrypted Communications & Verify that data transmissions are securely encrypted to prevent interception and tampering. \\ \hline
24 & Mobile Code Security Issues & For mobile applications, ensure proper handling of permission requests and secure data storage. \\ \hline
25 & Cloud Service Configuration Errors & Review any cloud-based configurations for potential data leaks or unauthorized access. \\
\hline

\hline

\hline

\hline
\end{supertabular}
\end{table*}


\section{Data Leakage Statement}
As the new dataset introduced in Section~\ref{sec:dataset}, the time of the collected dataset is after April 2023, avoiding data leakage while we evaluate \tool{} on \textbf{codeData} dataset.

\section{Algorithmic Description of \tool{} Pipeline with QA-Checker}

This algorithm demonstrates the integration of QA-Checker within the \tool{} pipeline, employing mathematical equations to describe the QA-Checker's iterative refinement process.

\begin{algorithm}[H]
\caption{Integrated Workflow of \tool{} with QA-Checker}
\begin{algorithmic}
    \State \textbf{Input:} Code submission, commit message, original files
    \State \textbf{Output:} Refined code review document

    \State Initialize phase $p = 1$
    \While{$p \leq 4$}
        \State \textbf{Switch:} Phase $p$
        \State \textbf{Case 1: Basic Info Sync}
            \State Conduct initial information analysis
            \State Update: $p = 2$
            
        \State \textbf{Case 2: Code Review}
            \State Perform code review with Coder and Reviewer
            \State Update: $p = 3$
            
        \State \textbf{Case 3: Code Alignment}
            \State Apply code revisions based on feedback
            \State Update: $p = 4$
            
        \State \textbf{Case 4: Document}
            \State Finalize review document
            \State Update: $p = 5$ (End)
            
        \State \textbf{QA-Checker Refinement} (Applies in Cases 2 and 3)
            \State Let $Q_i$ be the current question and $A_i$ the current answer
            \State Evaluate response quality: $qScore = \mathcal{Q}(Q_i, A_i)$
            \If{$qScore$ below threshold}
                \State Generate additional instruction $aai$
                \State Update question: $Q_{i+1} = Q_i + aai$
                \State Request new response: $A_{i+1}$
            \EndIf
    \EndWhile

    \State \textbf{Return:} Refined code review document
\end{algorithmic}
\end{algorithm}

In this algorithm, $\mathcal{Q}(Q_i, A_i)$ represents the quality assessment function of the QA-Checker, which evaluates the relevance and accuracy of the answer $A_i$ to the question $Q_i$. If the quality score $qScore$ is below a predefined threshold, the QA-Checker intervenes by generating an additional instruction $aai$ to refine the question, prompting a more accurate response in the next iteration.

\section{Detailed Performance of \tool{} in Various Languages on \textit{VA} task} \label{sec:vadetail}
In our comprehensive analysis using \tool{}, as detailed in Table~\ref{tab:languagediff}, we observe a diverse landscape of confirmed vulnerabilities across different programming languages. The table categorizes these vulnerabilities into `merged' and `closed' statuses for languages such as Python, Java, Go, C++, JavaScript, C, C\#, PHP, and Ruby. A significant finding is a markedly high number of `merged' vulnerabilities in Python, potentially reflective of its extensive application or intrinsic complexities leading to security gaps. Conversely, languages like Go, Ruby, and C exhibit notably lower counts in both categories, perhaps indicating lesser engagement in complex applications or more robust security protocols. Table~\ref{tab:languagediff} that the `closed' category consistently presents lower vulnerabilities than `merged' across most languages, signifying effective resolution mechanisms. However, an exception is noted in C, where `closed' counts surpass those of `merged', possibly indicating either delayed vulnerability identification or efficient mitigation strategies. Remarkably, the Rate$_{close}$ is generally observed to be higher than Rate$_{merge}$ across the languages, exemplifying a significant reduction in vulnerabilities post-resolution. For example, Python demonstrates a Rate$_{merge}$ of 14.00\% against a higher Rate$_{close}$ of 18.16\%. This trend is consistent in most languages, emphasizing the importance of proactive vulnerability management. The Rate$_{avg}$, representing the proportion of confirmed vulnerabilities against the total of both merged and closed items, further elucidates this point, with C++ showing the highest Rate$_{avg}$ at 16.49\%. These insights not only underline the diverse vulnerability landscape across programming languages but also highlight the adeptness of \tool{} in pinpointing and verifying vulnerabilities in these varied contexts.

\begin{table*}[htbp]
\caption{Vulnerable problems (\#) found by \tool{}. Rate$_{merge}$ means the value of confirmed divided by the total number in the merged and  Rate$_{close}$ is the value of confirmed divided by the total number in the closed. Rate$_{avg}$ is the value of the confirmed number divided by the total number of the merged and closed.}
\label{tab:languagediff}
\resizebox{1\textwidth}{!}{%
\begin{tabular}{l|ccccccccc}
\hline

\hline

\hline

\hline
\tool{}&  Python & Java & Go & C++ & JavaScript & C & C\# & PHP & Ruby\\ 
\hline

\hline
merged (total\#) &  1,057 & 287 & 133 & 138 & 280 & 114 & 206 & 173 & 202\\
merged (confirmed\#) &  148 & 17 & 11 & 19 & 34 & 9 & 21 & 28 & 20\\
Rate$_{merge}$ &  14.00\% & 5.92\% & 8.27\% & 13.77\% & 12.14\% & 7.89\% & 10.19\% & \cellcolor{gray!30}16.18\% & 9.90\%\\
closed (total\#)&  248 & 97 & 74 & 56 & 112 & 146 &62 & 105 & 55\\ 
closed (confirmed\#) &  45 & 10 & 5 & 13 & 16 & 26 &7 & 15 & 5\\ 
Rate$_{close}$  &  18.16\% & 10.31\% & 6.76\% & \cellcolor{gray!30}23.2\% & 14.29\% & 17.81\% &11.29\% & 14.29\% & 9.09\% \\ 

Total number (\#) &  1,305 & 384 & 207 & 194 & 392 & 260 &268 & 278 & 257\\ 
Total confirmed (\#) & 193  & 27 & 16 & 32 & 50 & 35 &28 & 43 & 25\\ 
Rate$_{avg}$ & 14.79\%  & 7.03\% & 7.73\% & \cellcolor{gray!30}16.49\% & 12.76\% & 13.46\% &10.45\% & 14.47\% & 9.73\%\\

\hline

\hline

\hline

\hline
\end{tabular}%
 }
\end{table*}

\section{More detailed experimental results on CA and FA tasks}\label{sec:cafadetailedexp}

Detailed experimental results of CA are shown in Figure~\ref{fig:mergedmetricdetail} and Figure~\ref{fig:closedmetricdetail}. Detailed experimental results of FA are shown in Figure \ref{famergedmetricdetail} and Figure \ref{faclosedmetricdetail}.

\begin{figure*}[htbp]
    \centering
    \subfigure{
        \includegraphics[width=.3\linewidth]{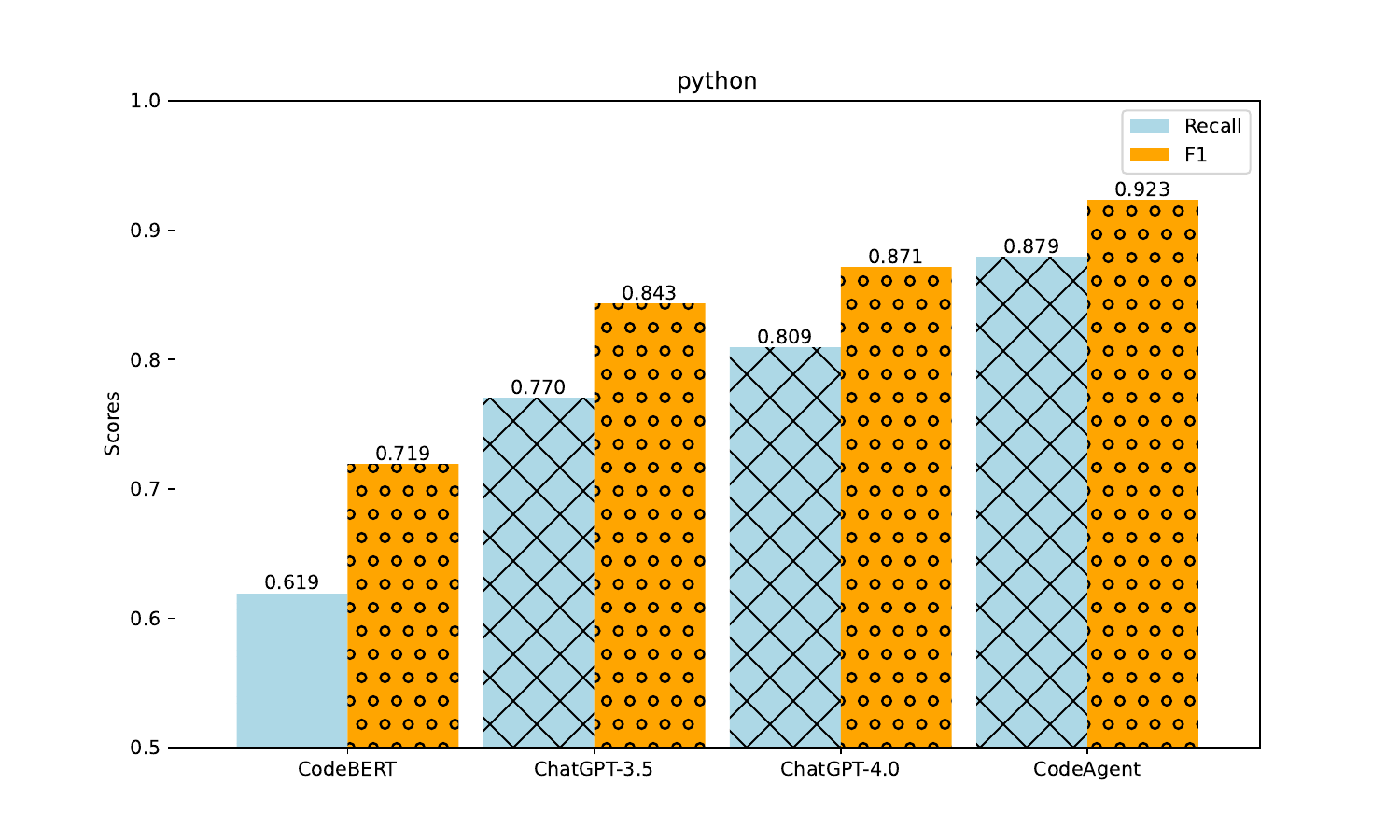}
        \label{fig:mergedpython}
    }%
    \subfigure{
        \includegraphics[width=.3\linewidth]{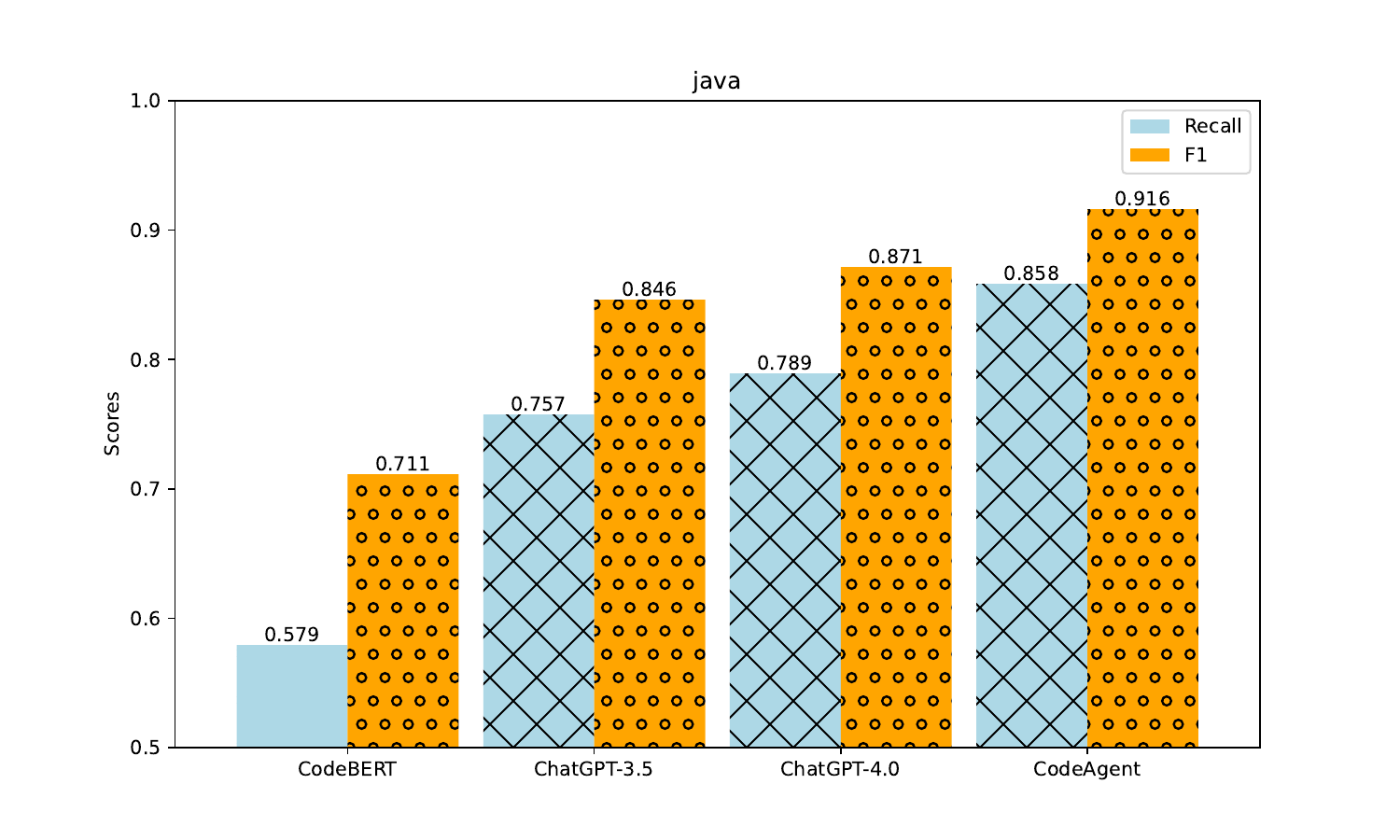}
        \label{fig:mergedjava}
    }
    \subfigure{
        \includegraphics[width=.3\linewidth]{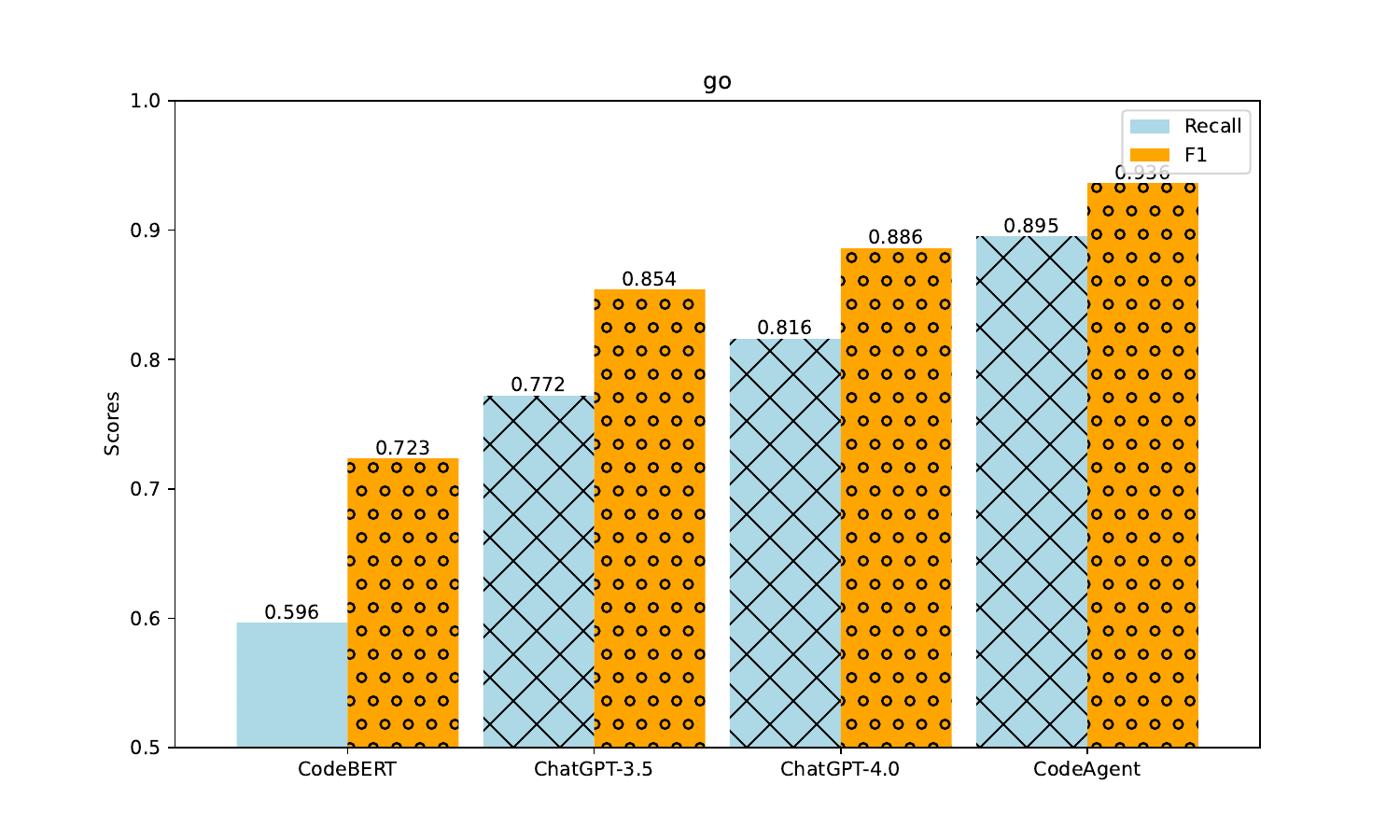}
        \label{fig:mergedgo}
    }
    \subfigure{
        \includegraphics[width=.3\linewidth]{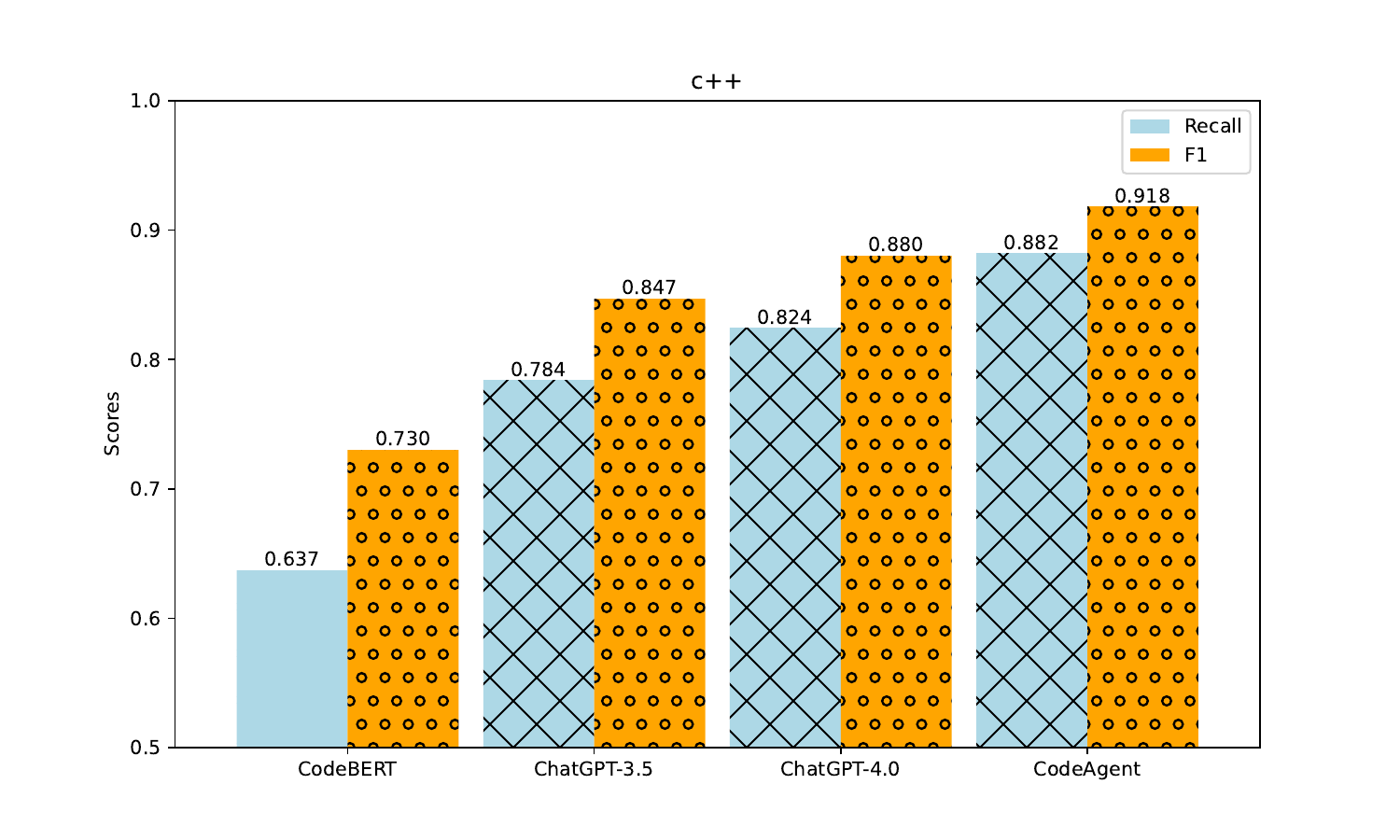}
        \label{fig:mergedc++}
    }
    \subfigure{
        \includegraphics[width=.3\linewidth]{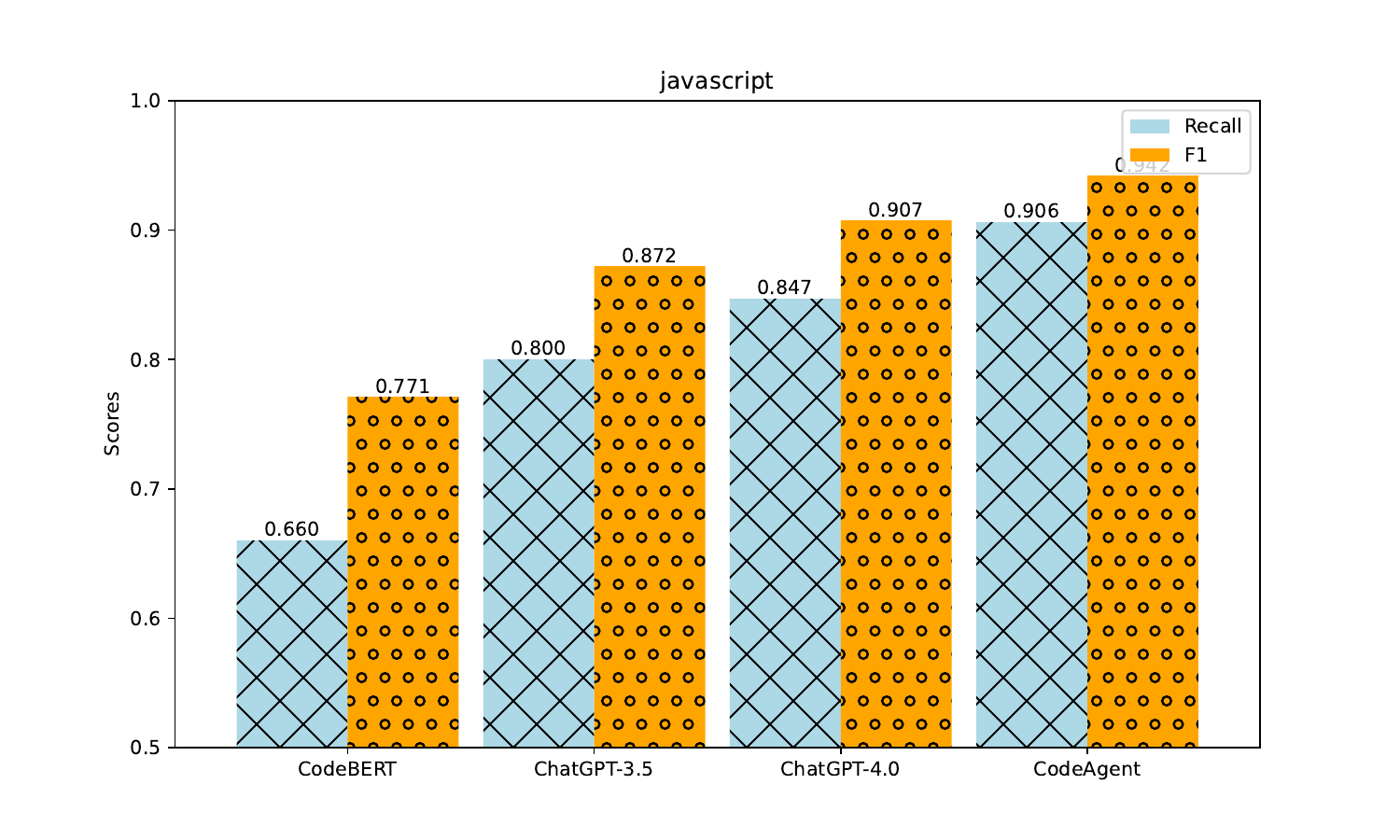}
        \label{fig:mergedjavascript}
    }
    \subfigure{
        \includegraphics[width=.3\linewidth]{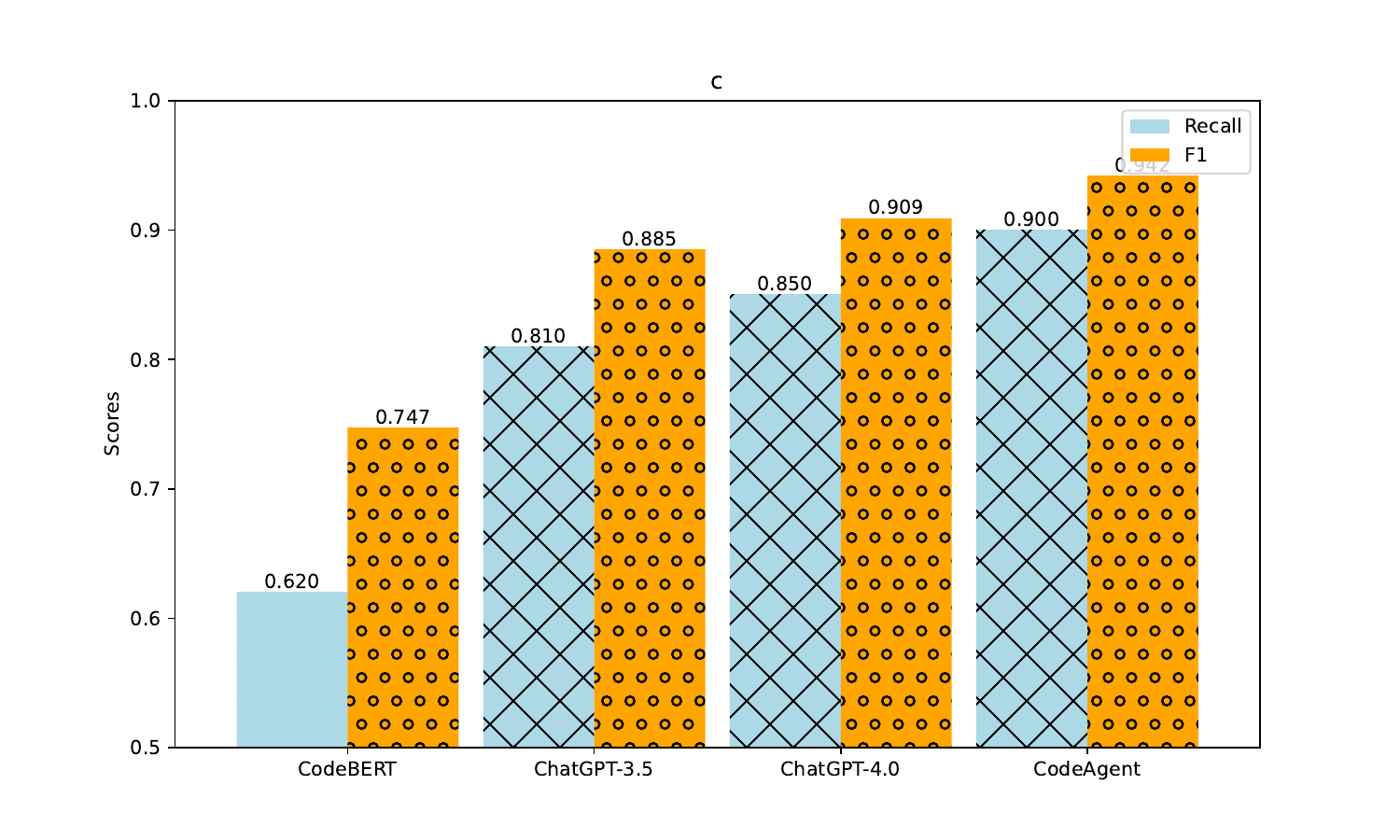}
        \label{fig:mergedc}
    }
    \subfigure{
        \includegraphics[width=.3\linewidth]{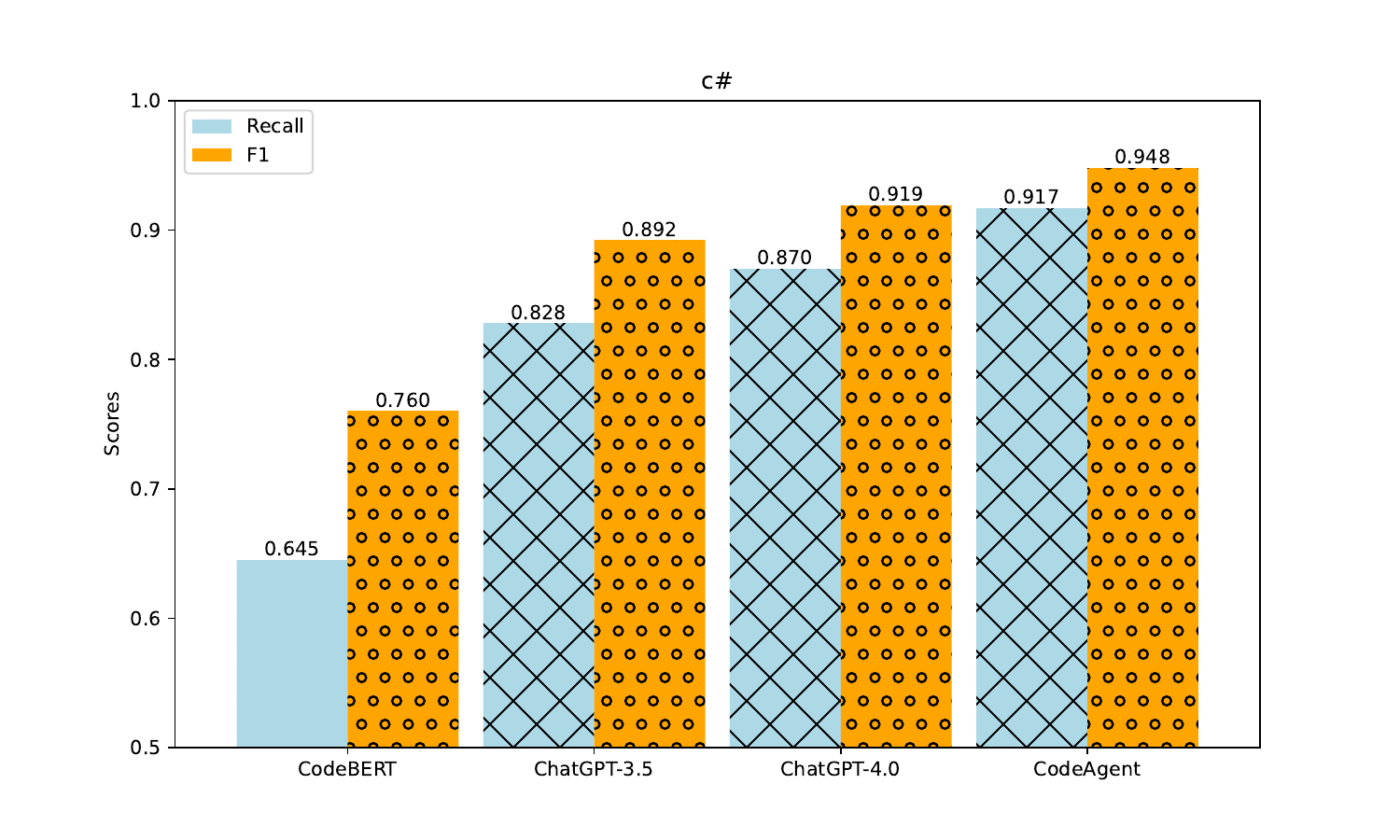}
        \label{fig:mergedcsharp}
    }
    \subfigure{
        \includegraphics[width=.3\linewidth]{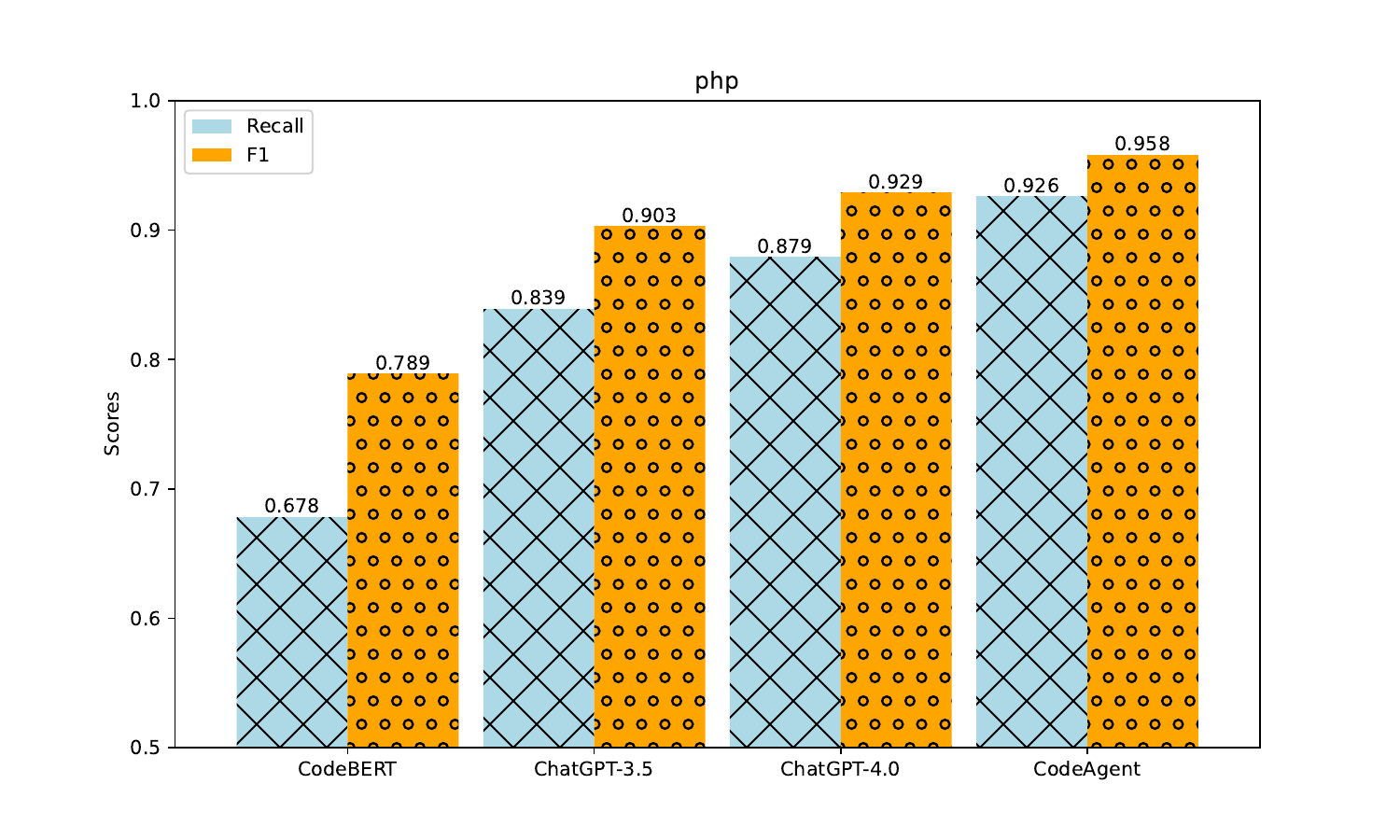}
        \label{fig:mergedphp}
    }
    \subfigure{
        \includegraphics[width=.3\linewidth]{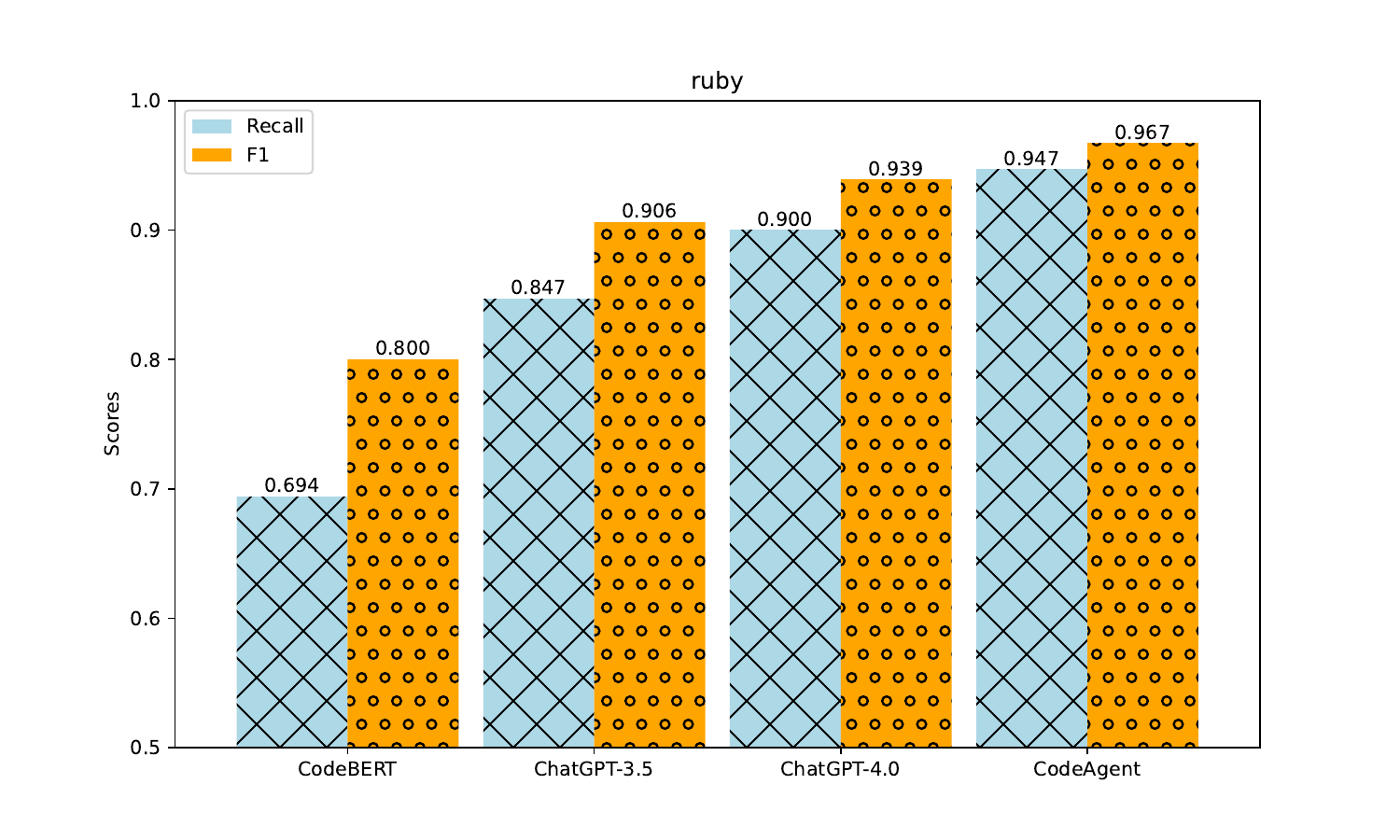}
        \label{fig:mergedruby}
    }
    \caption{Comparison of models on the \textbf{merged} data across 9 languages on \textbf{CA task}.}
    \label{fig:mergedmetricdetail}
\end{figure*}

\begin{figure*}[htbp]
    \centering
    \subfigure{
        \includegraphics[width=.3\linewidth]{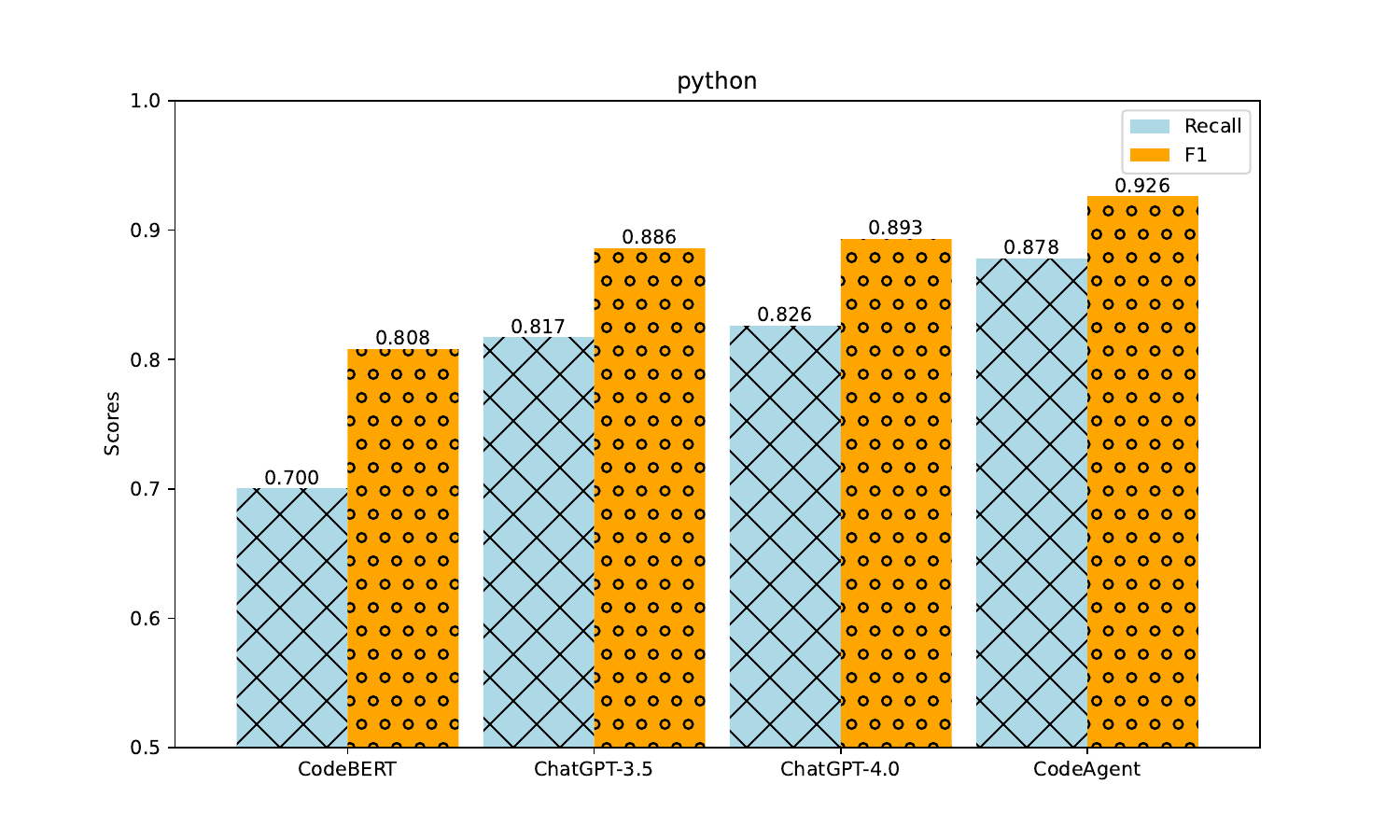}
    }%
    \subfigure{
        \includegraphics[width=.3\linewidth]{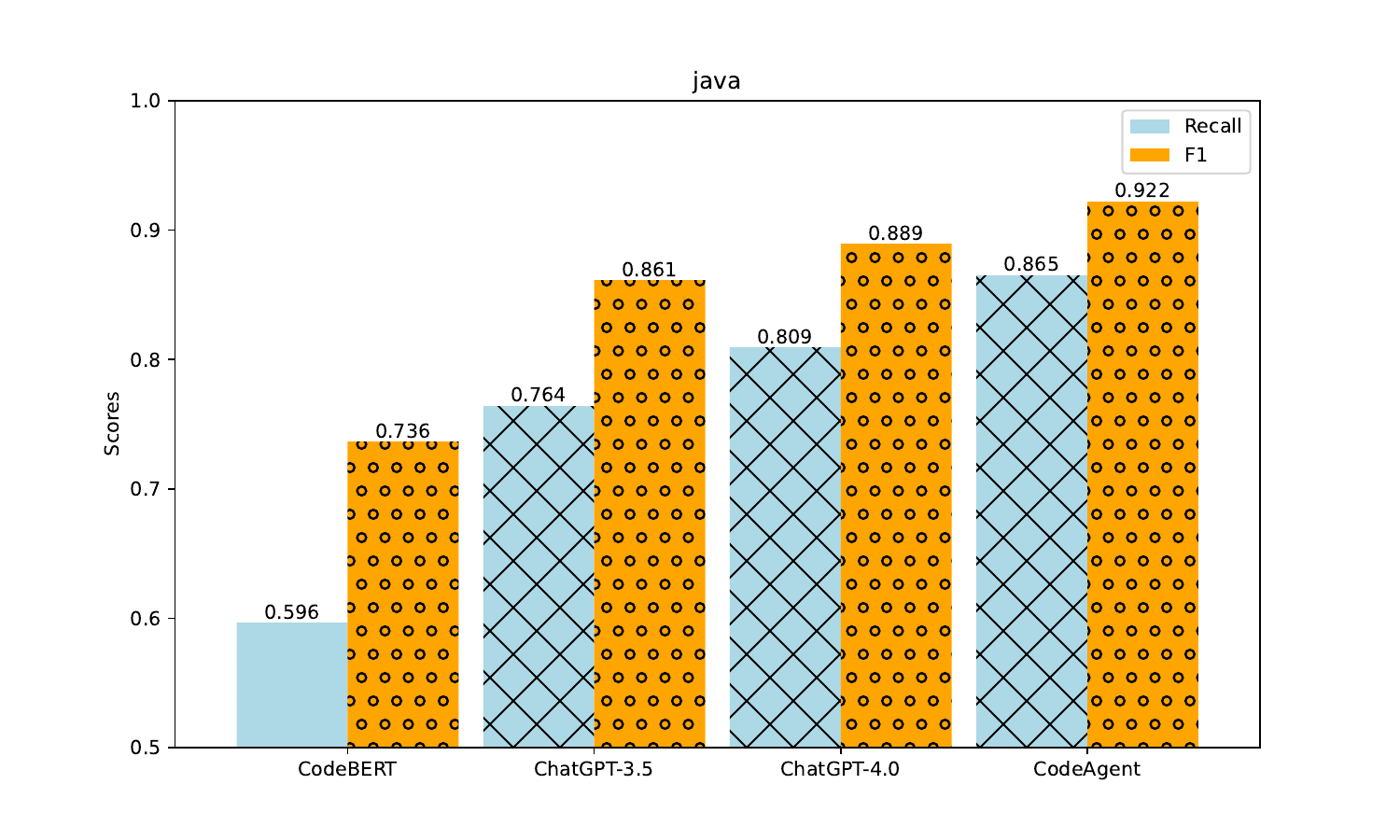}
    }
    \subfigure{
        \includegraphics[width=.3\linewidth]{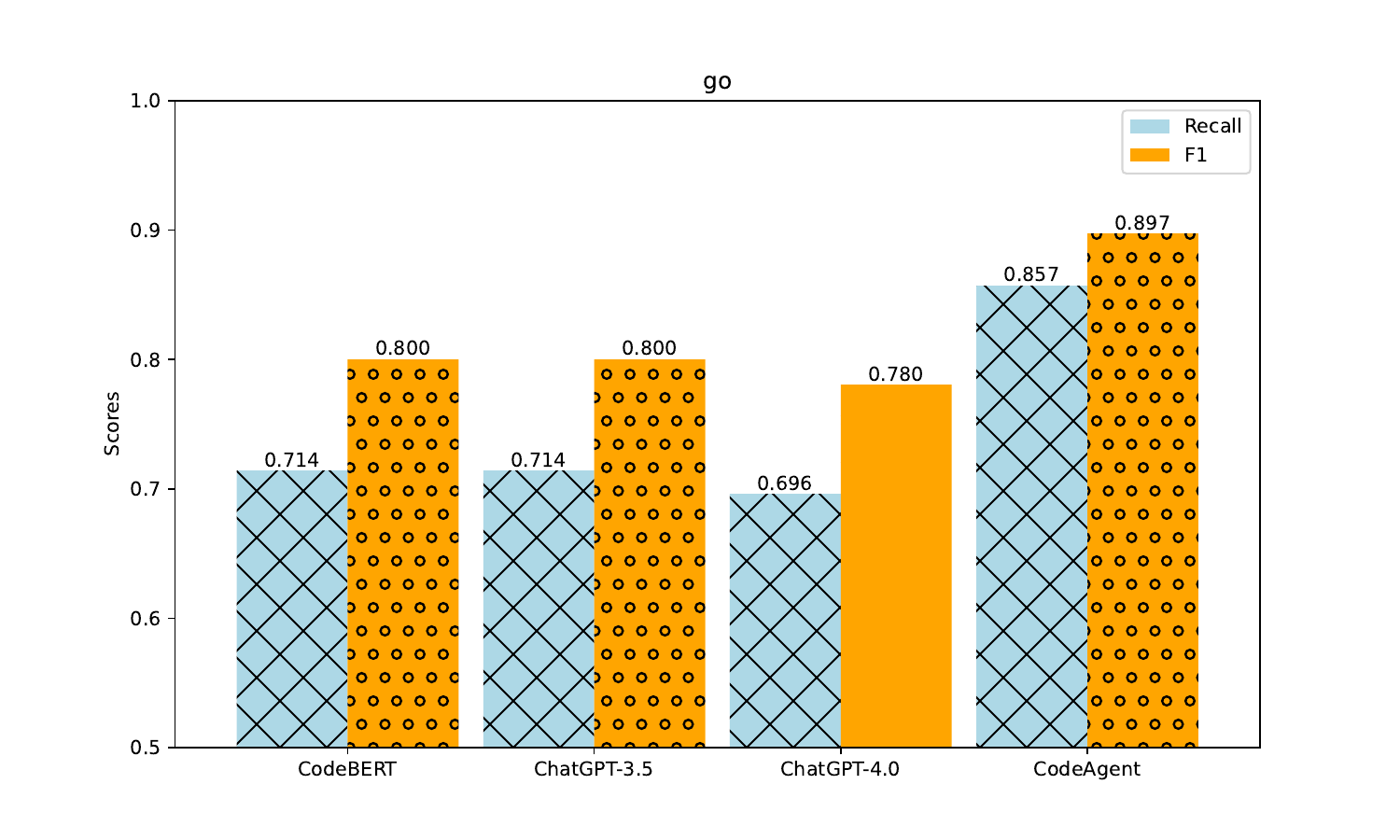}
    }
    \subfigure{
        \includegraphics[width=.3\linewidth]{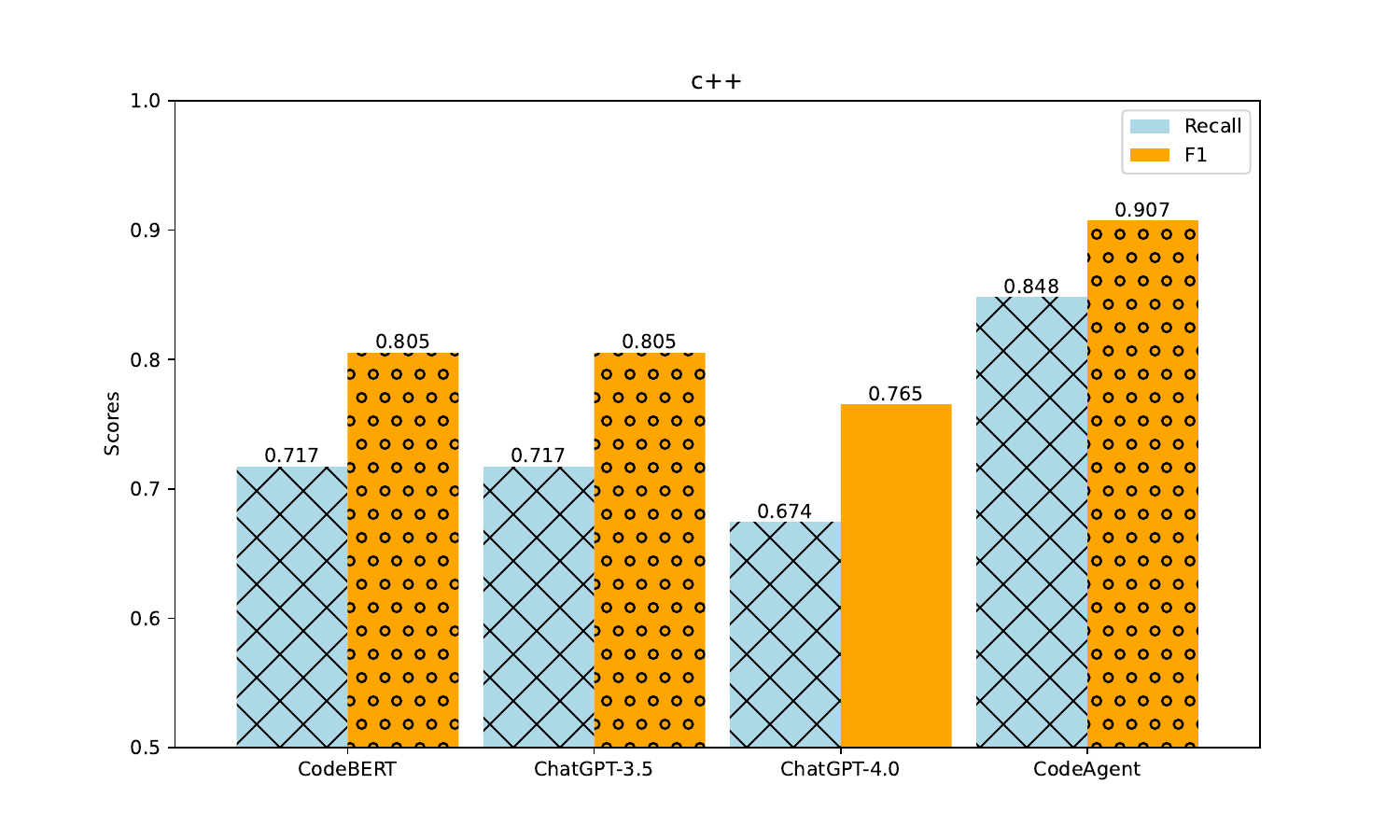}
    }
    \subfigure{
        \includegraphics[width=.3\linewidth]{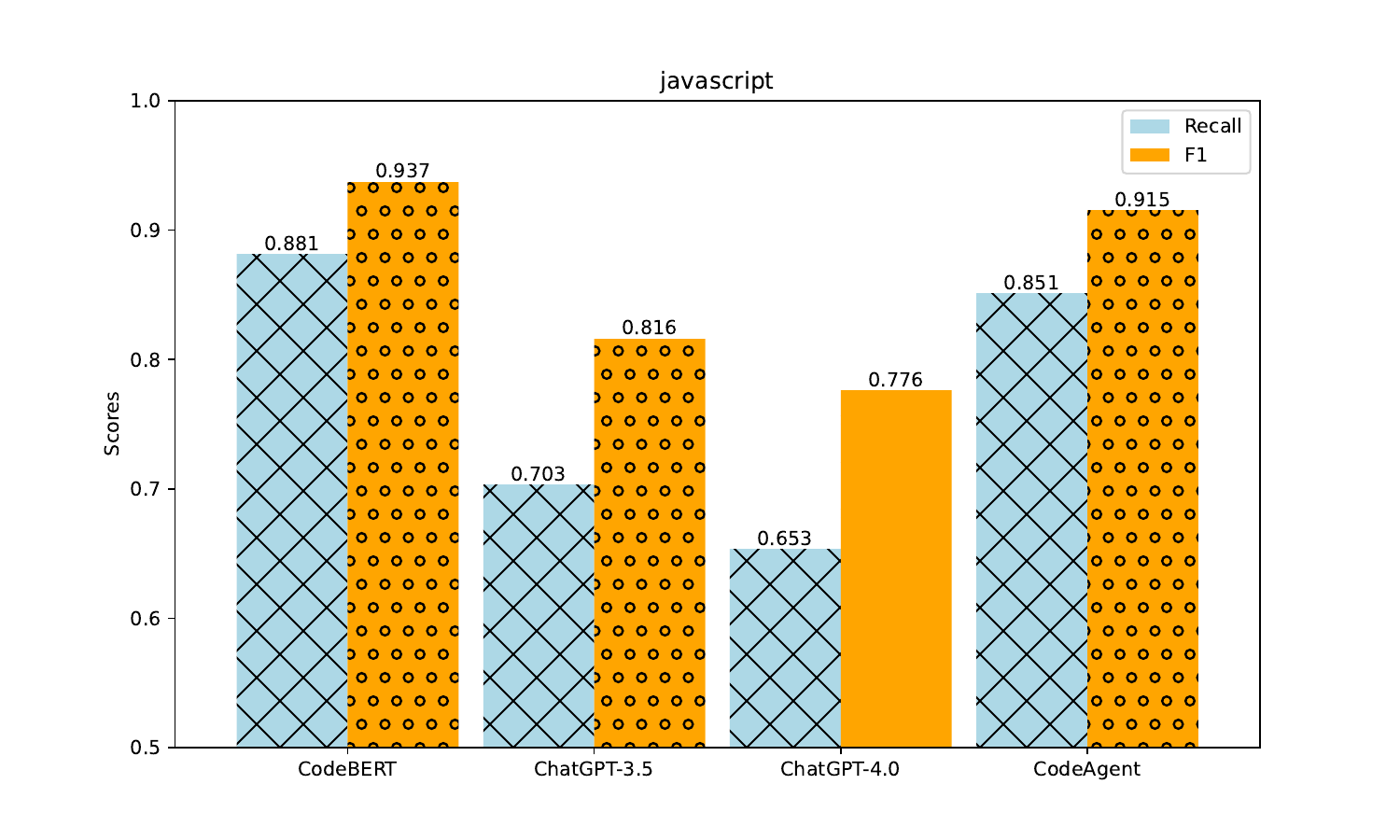}
    }
    \subfigure{
        \includegraphics[width=.3\linewidth]{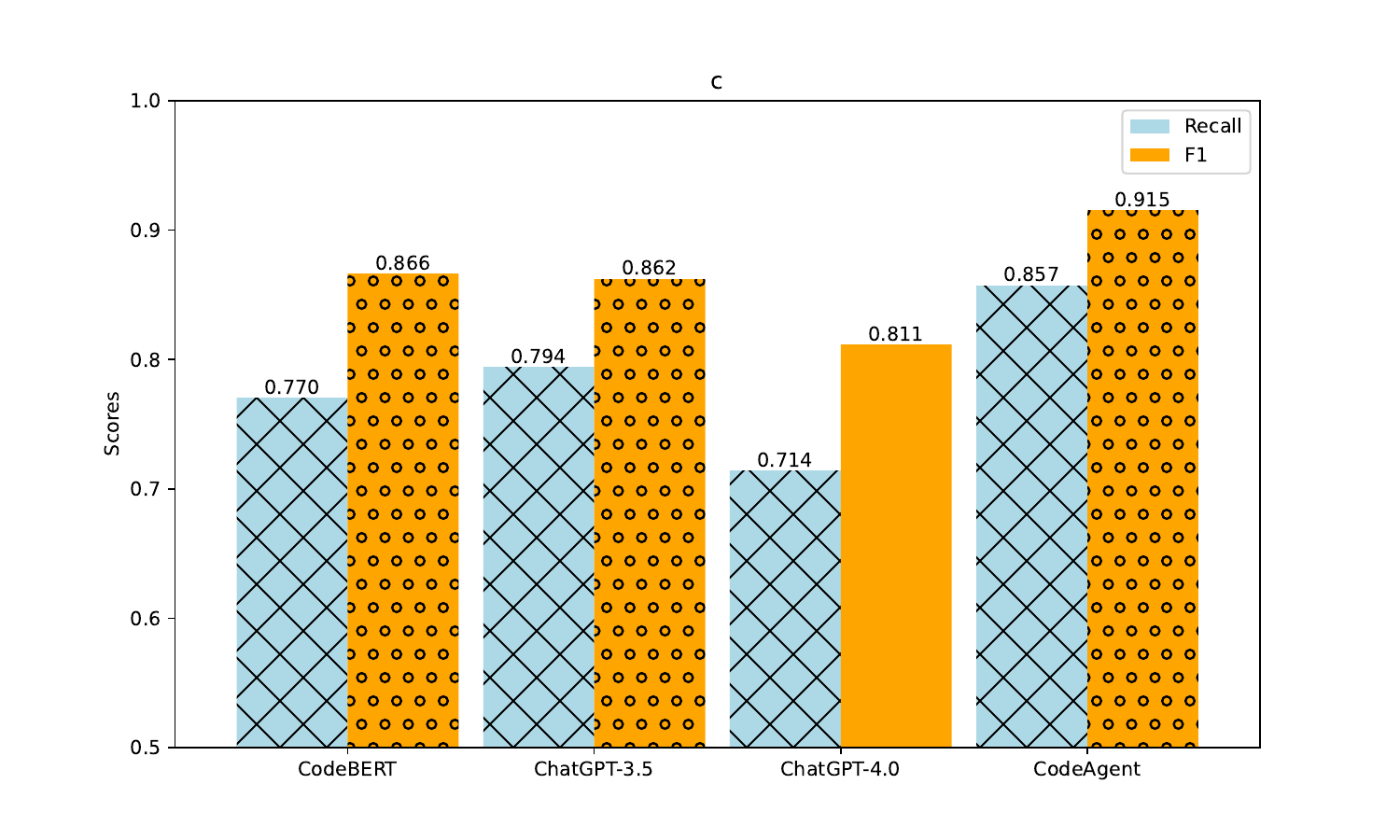}
    }
    \subfigure{
        \includegraphics[width=.3\linewidth]{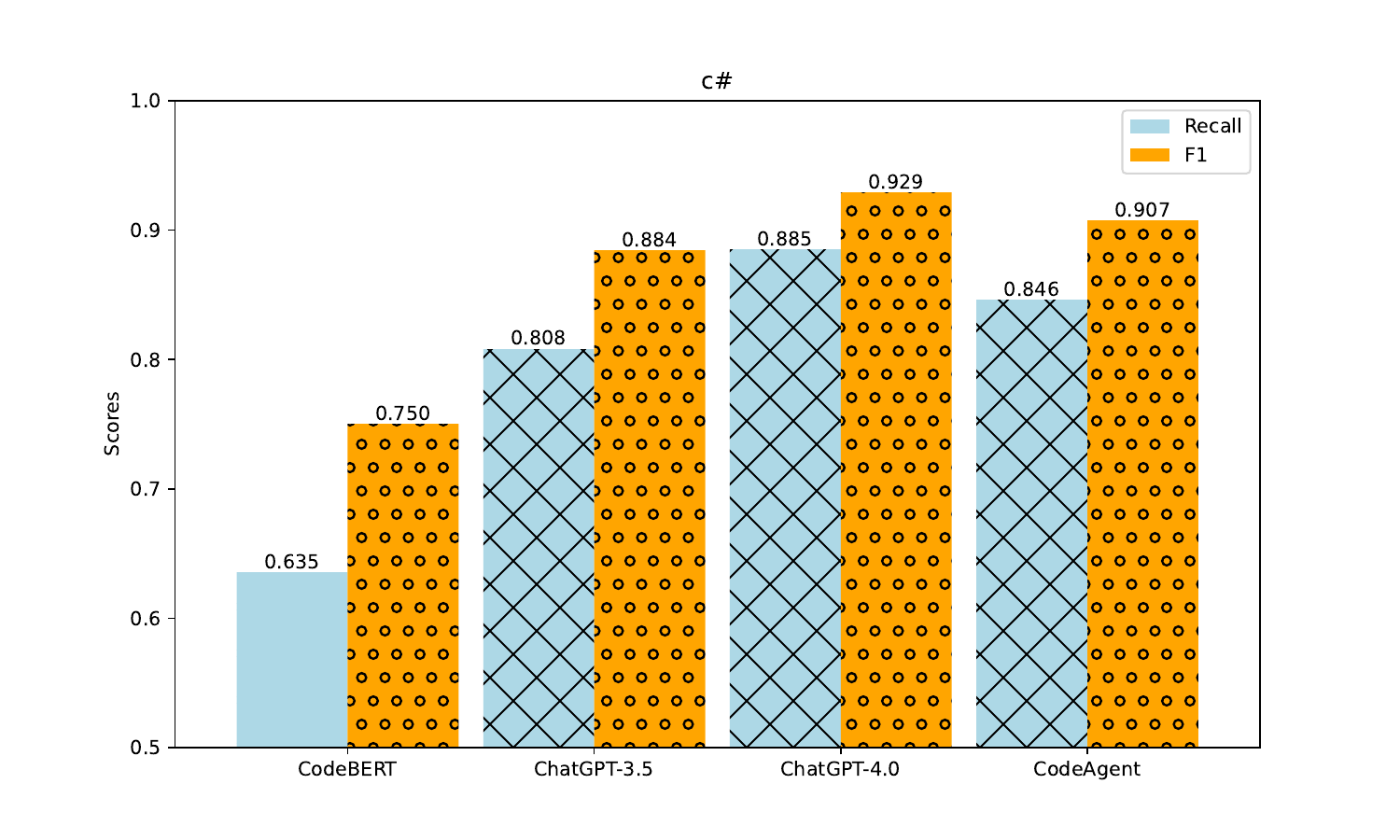}
    }
    \subfigure{
        \includegraphics[width=.3\linewidth]{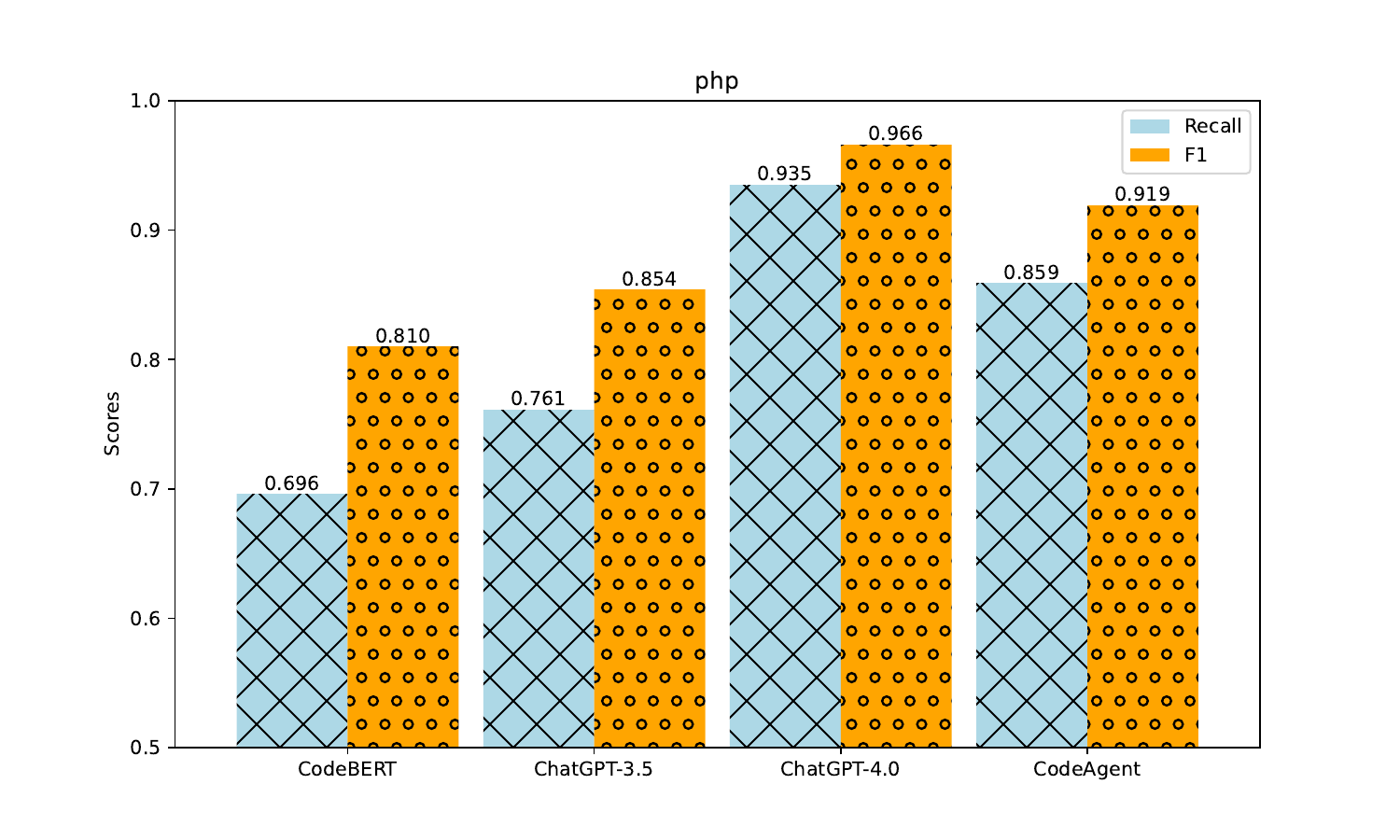}
    }
    \subfigure{
        \includegraphics[width=.3\linewidth]{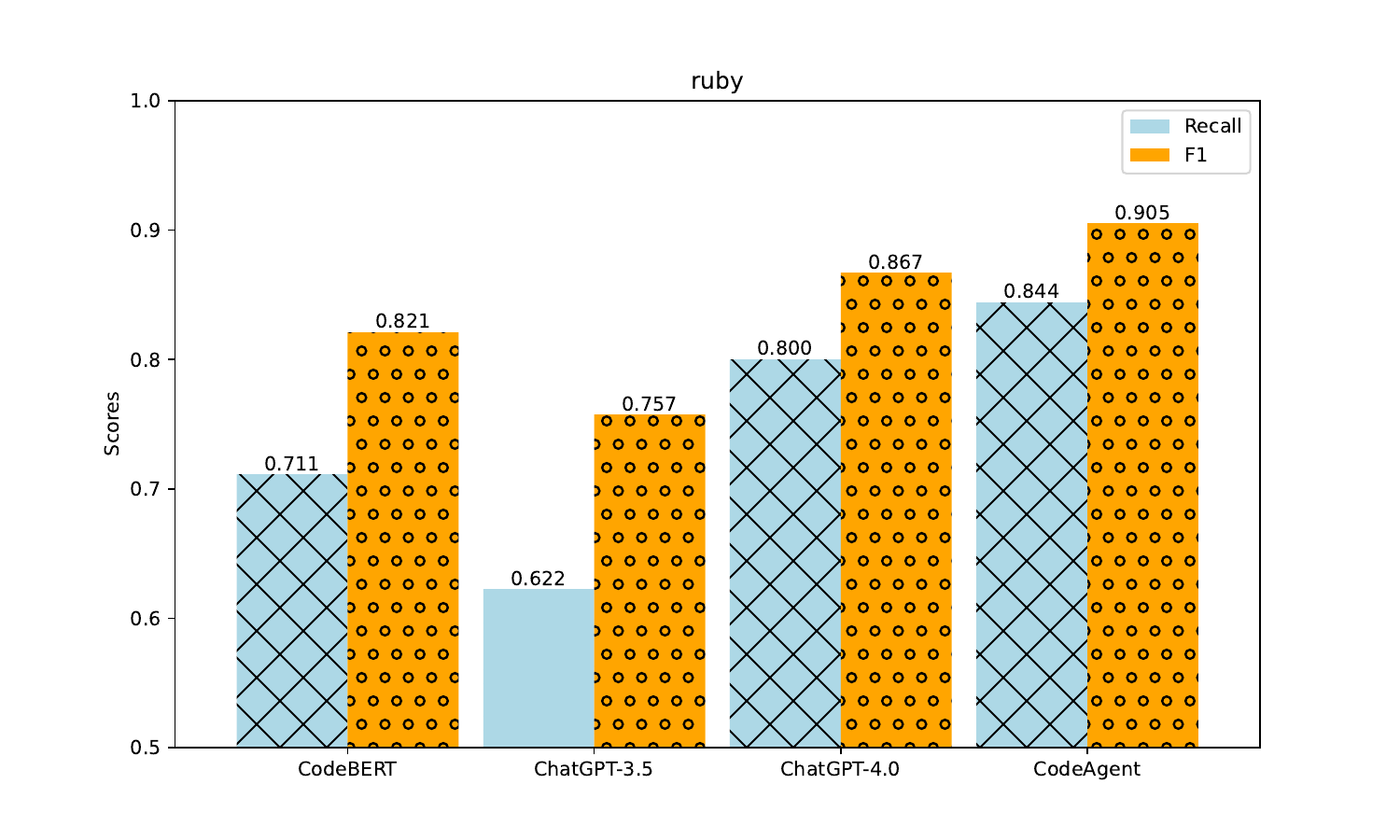}
    }
    \caption{Comparison of models on the \textbf{closed} data across 9 languages on \textbf{CA task}.}
    \label{fig:closedmetricdetail}
\end{figure*}

\begin{figure*}[htbp]
    
    \centering
    \subfigure{
        \includegraphics[width=.3\linewidth]{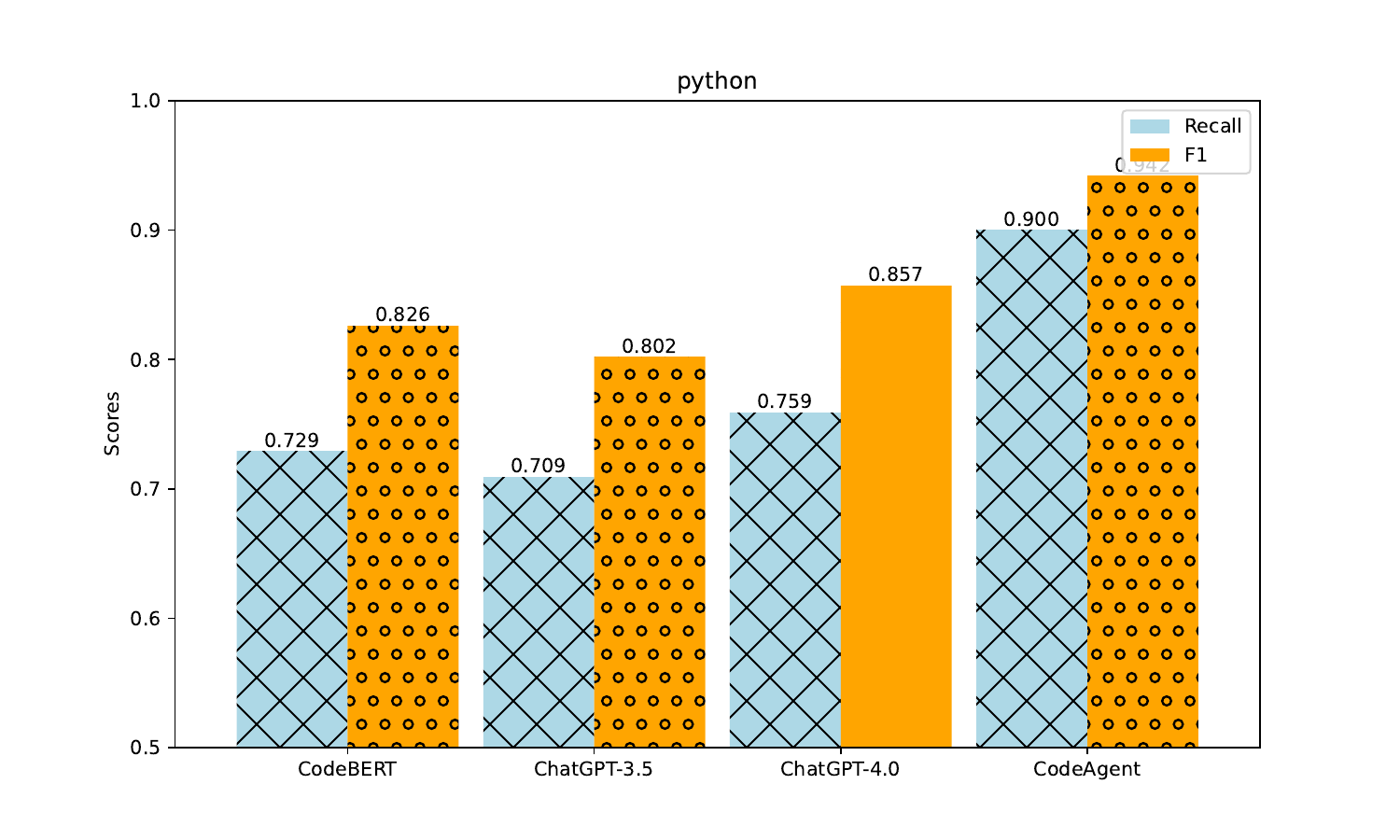}
        \label{fig:famergedpython}
    }%
    \subfigure{
        \includegraphics[width=.3\linewidth]{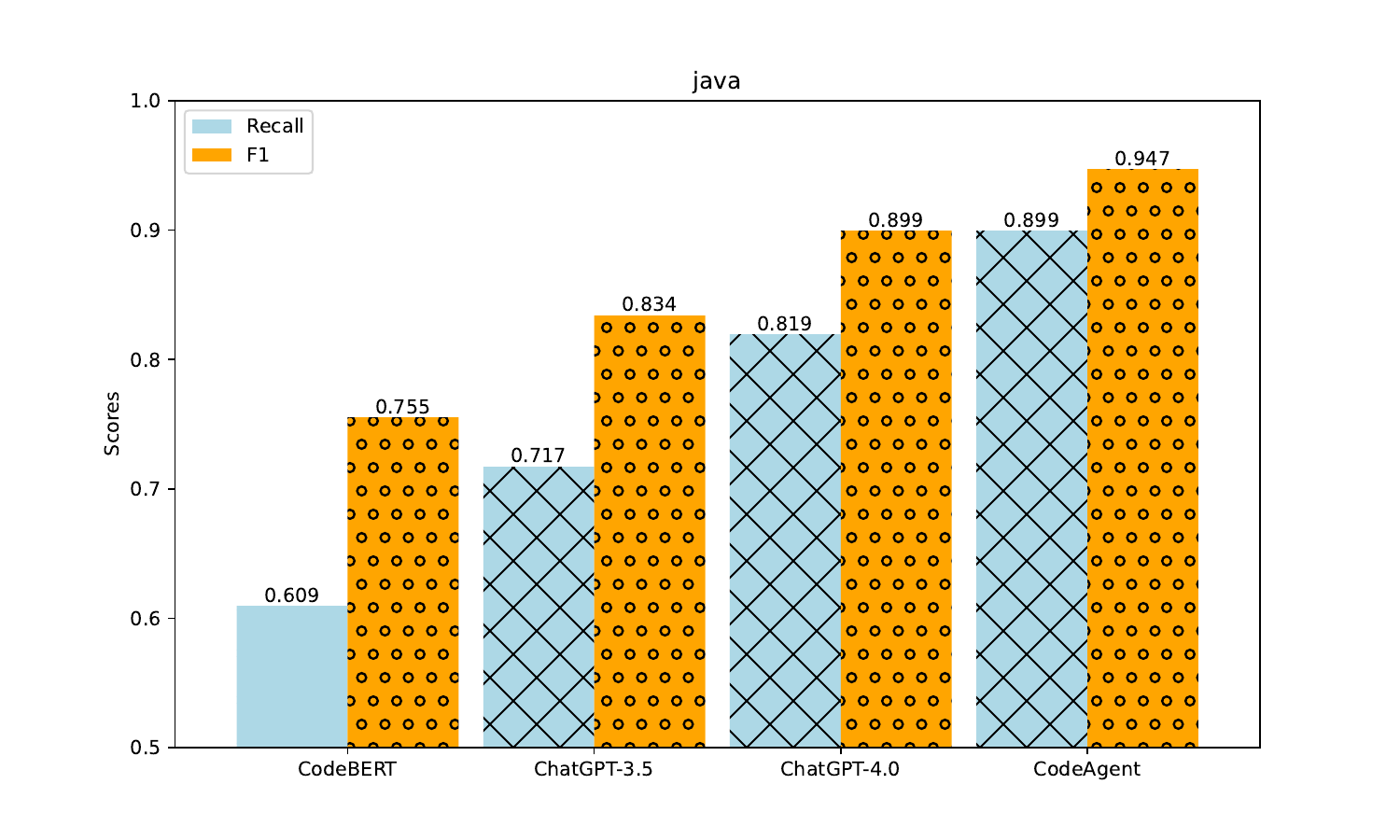}
        \label{fig:famergedjava}
    }
    \subfigure{
        \includegraphics[width=.3\linewidth]{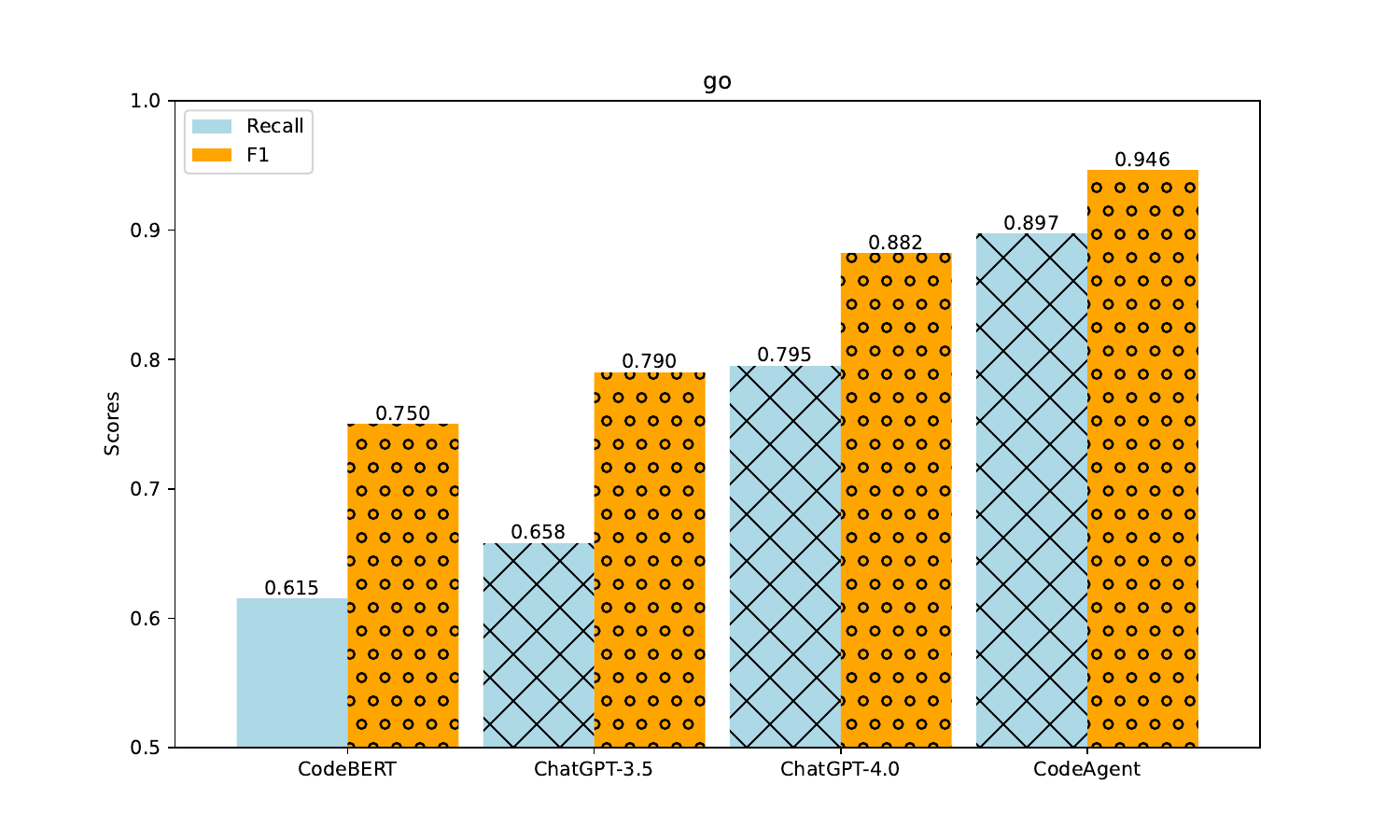}
        \label{fig:famergedgo}
    }
    \subfigure{
        \includegraphics[width=.3\linewidth]{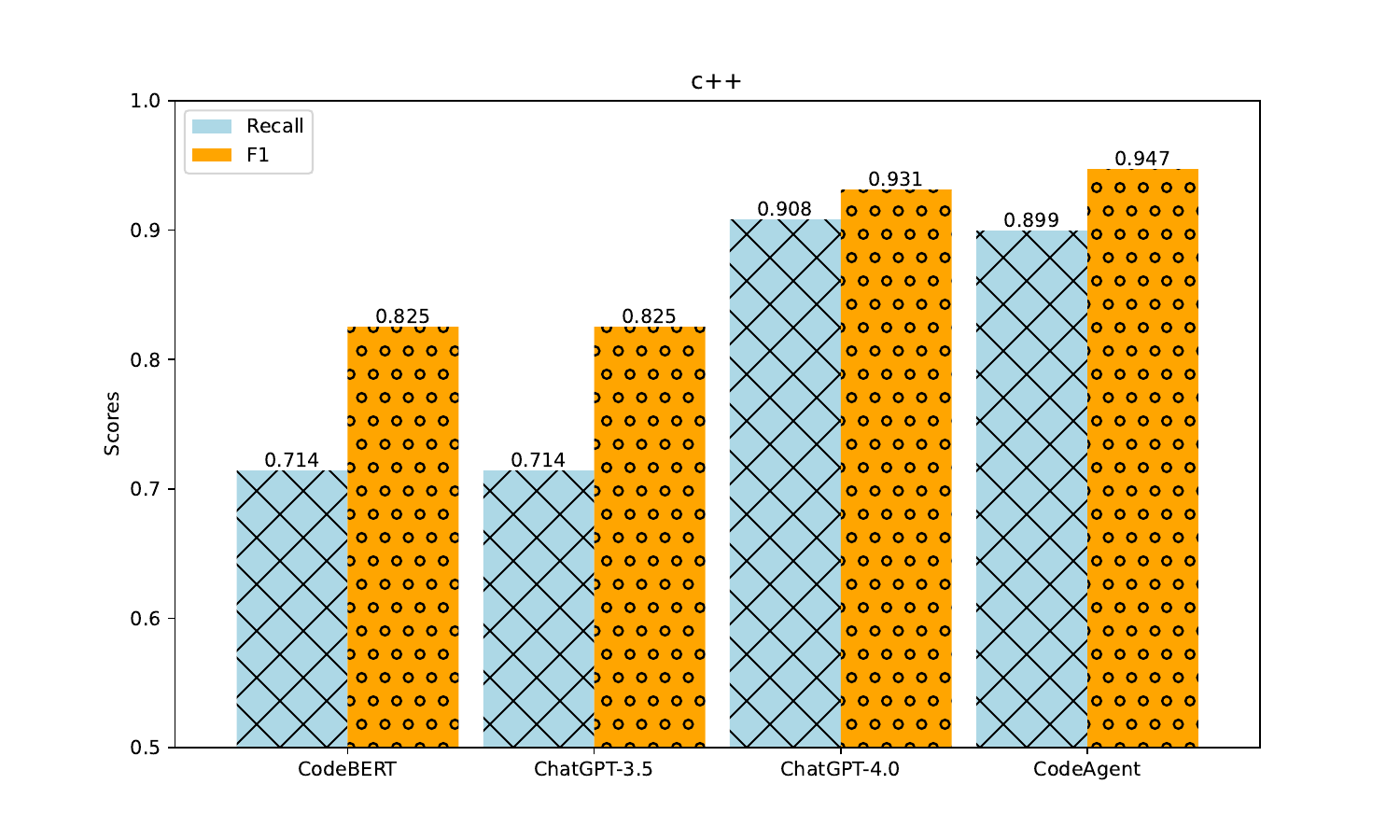}
        \label{fig:famergedc++}
    }
    \subfigure{
        \includegraphics[width=.3\linewidth]{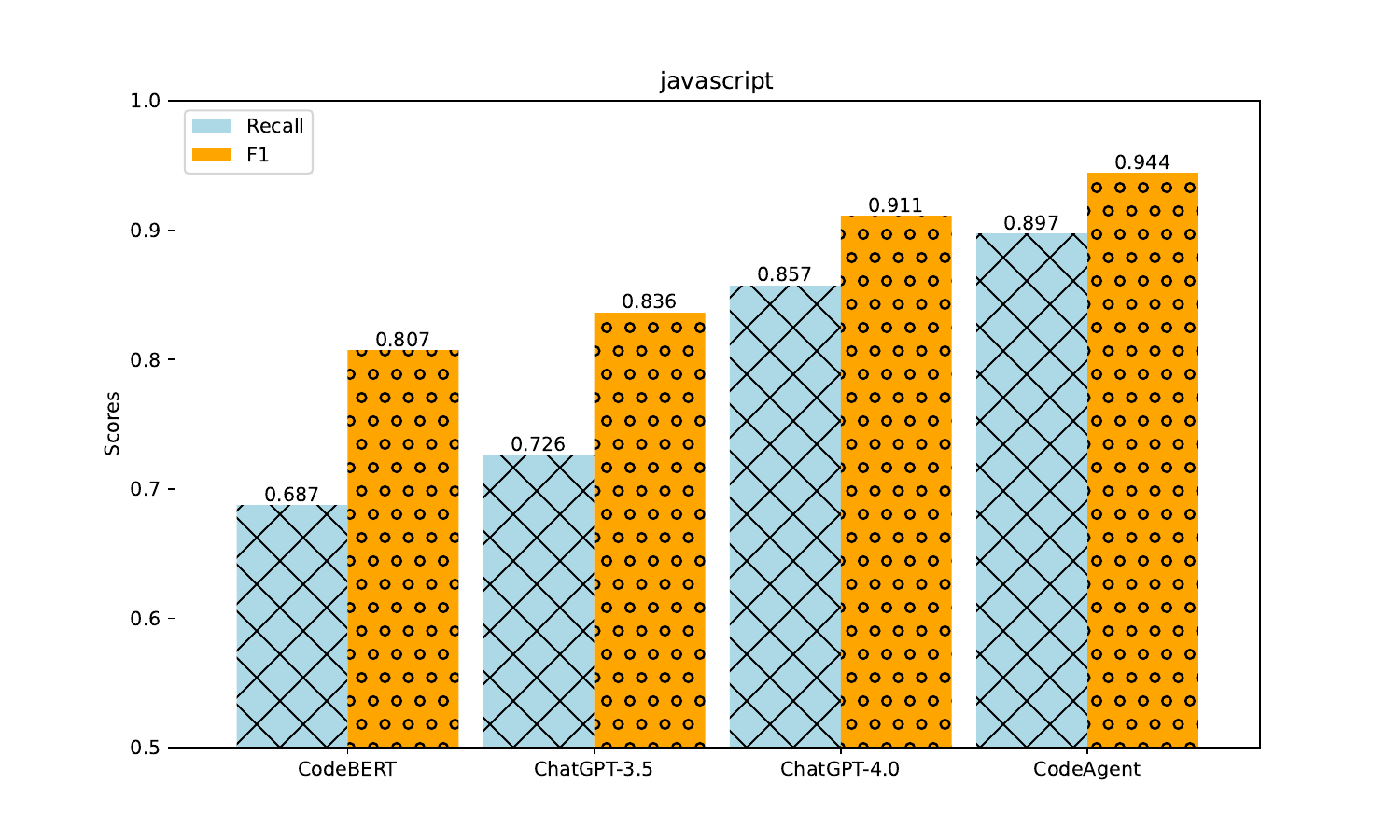}
        \label{fig:famergedjavascript}
    }
    \subfigure{
        \includegraphics[width=.3\linewidth]{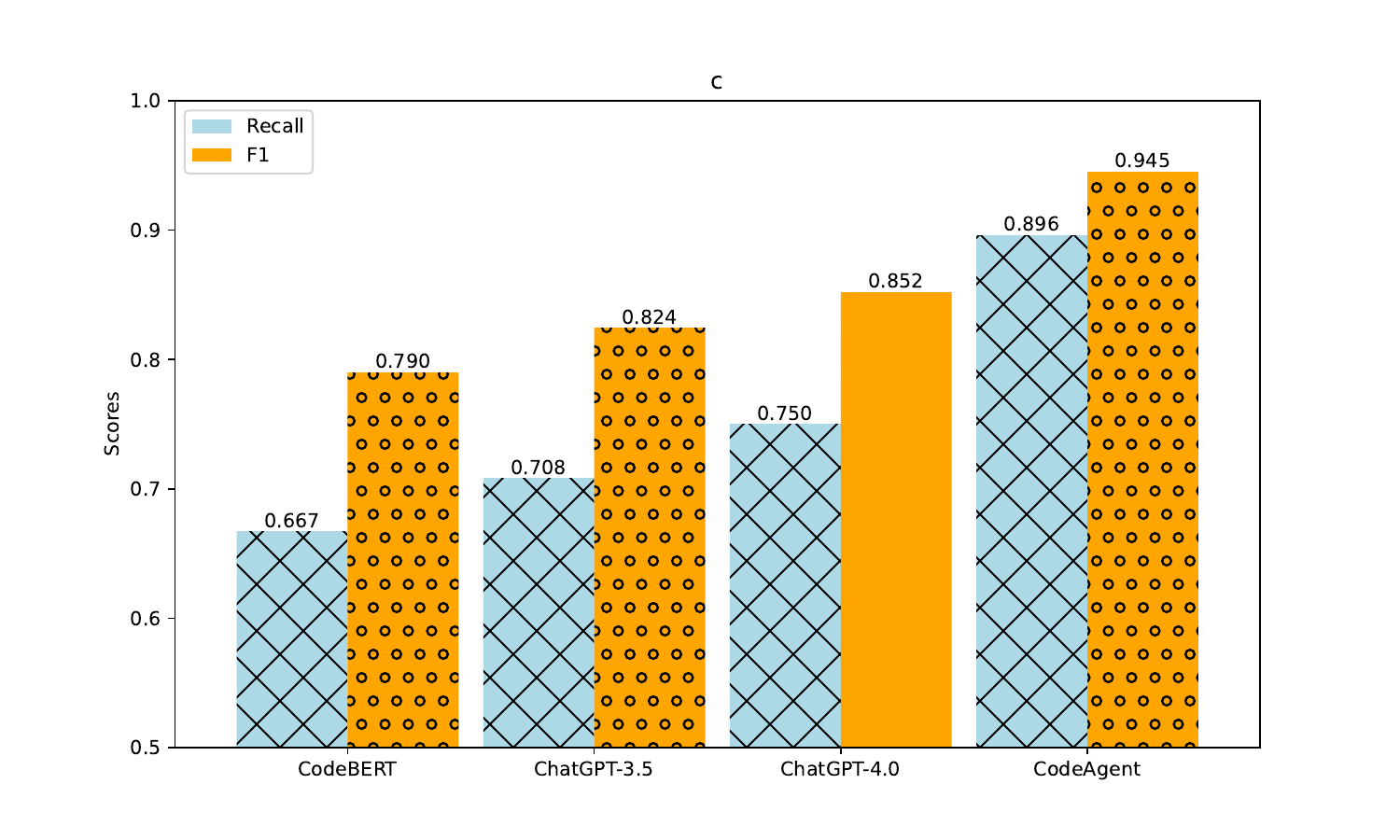}
        \label{fig:famergedc}
    }
    \subfigure{
        \includegraphics[width=.3\linewidth]{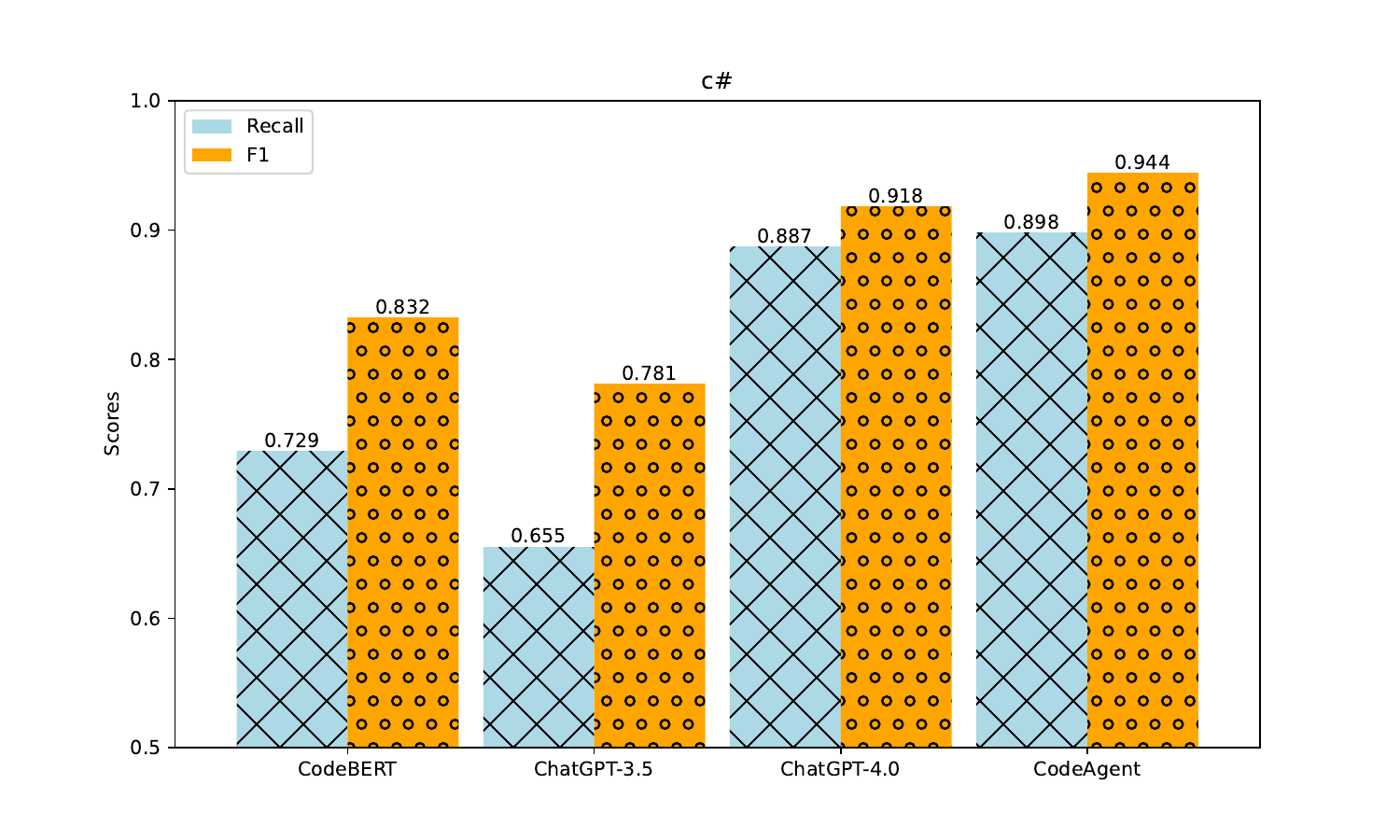}
        \label{fig:famergedcsharp}
    }
    \subfigure{
        \includegraphics[width=.3\linewidth]{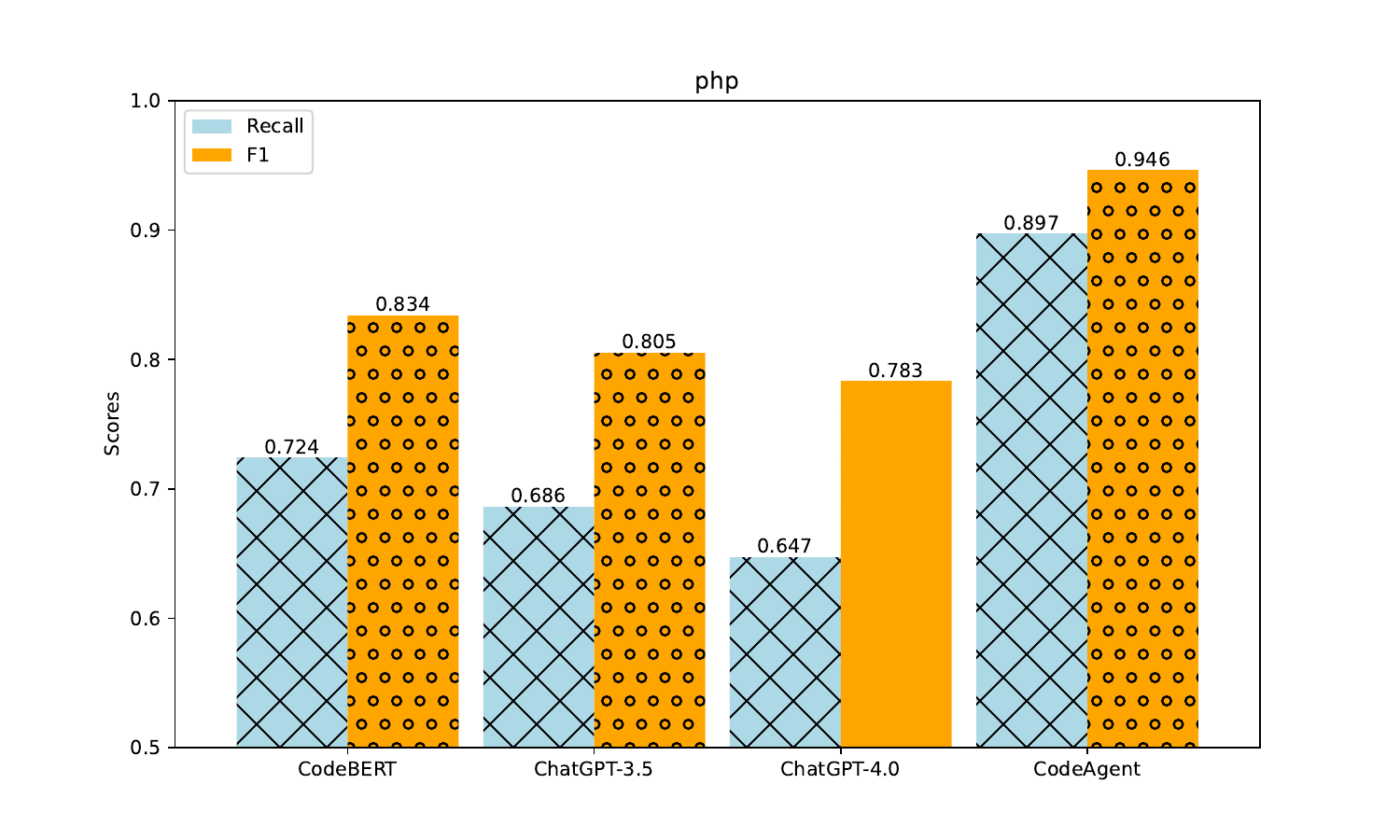}
        \label{fig:famergedphp}
    }
    \subfigure{
        \includegraphics[width=.3\linewidth]{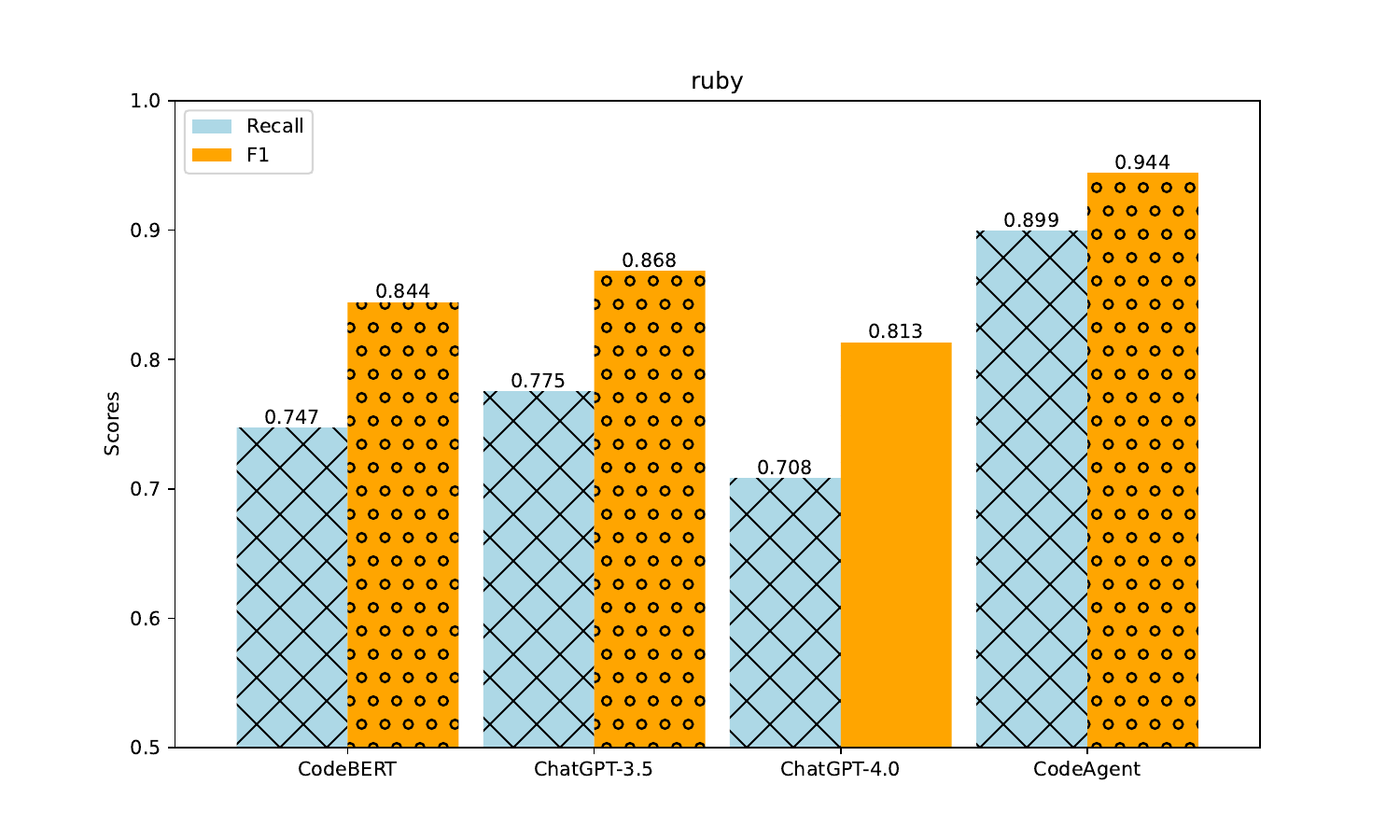}
        \label{fig:famergedruby}
    }
    \caption{Comparison of models on the \textbf{merged} data across 9 languages on \textbf{FA task}.}
    \label{famergedmetricdetail}
\end{figure*}

\begin{figure*}[htbp]
    \centering
    \subfigure{
        \includegraphics[width=.3\linewidth]{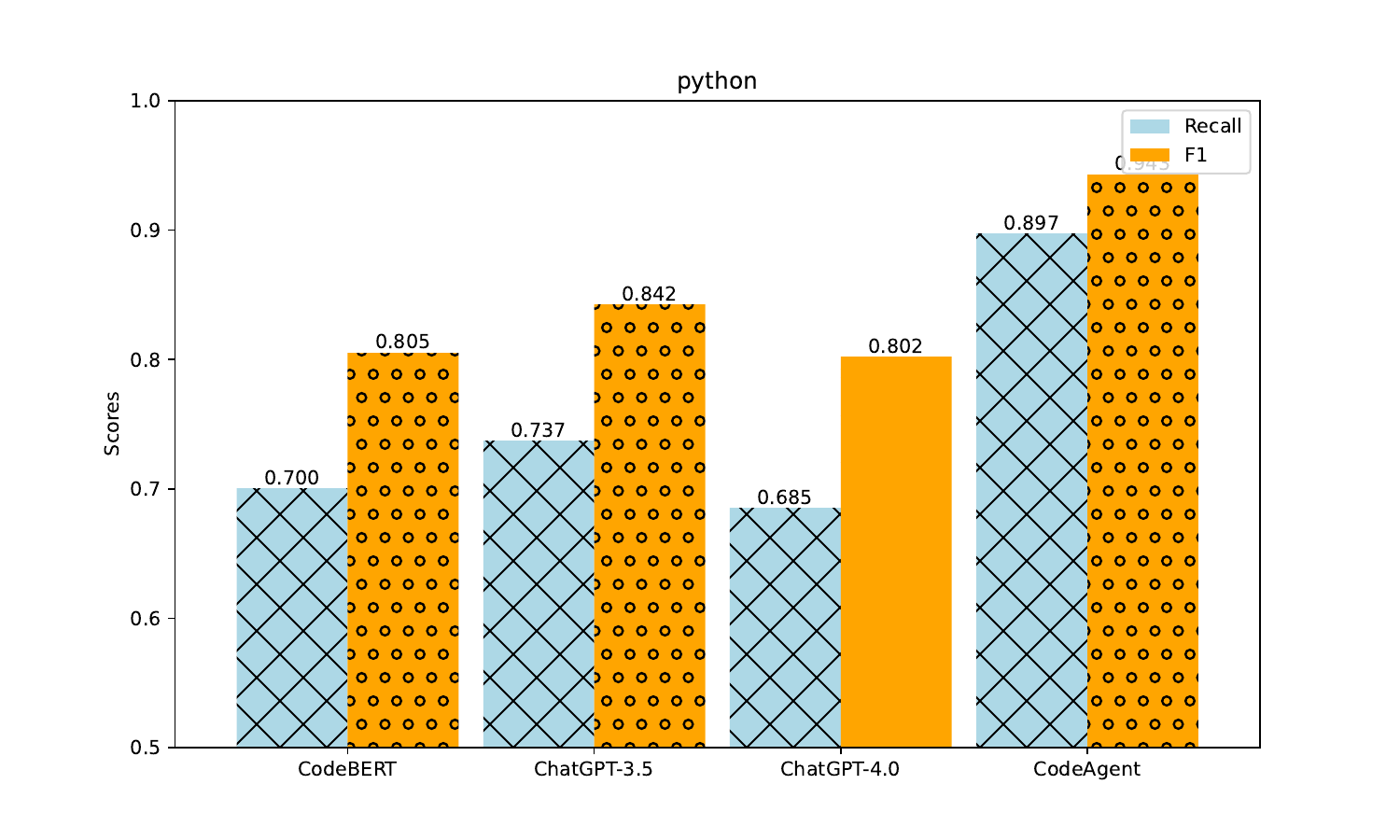}
    }%
    \subfigure{
        \includegraphics[width=.3\linewidth]{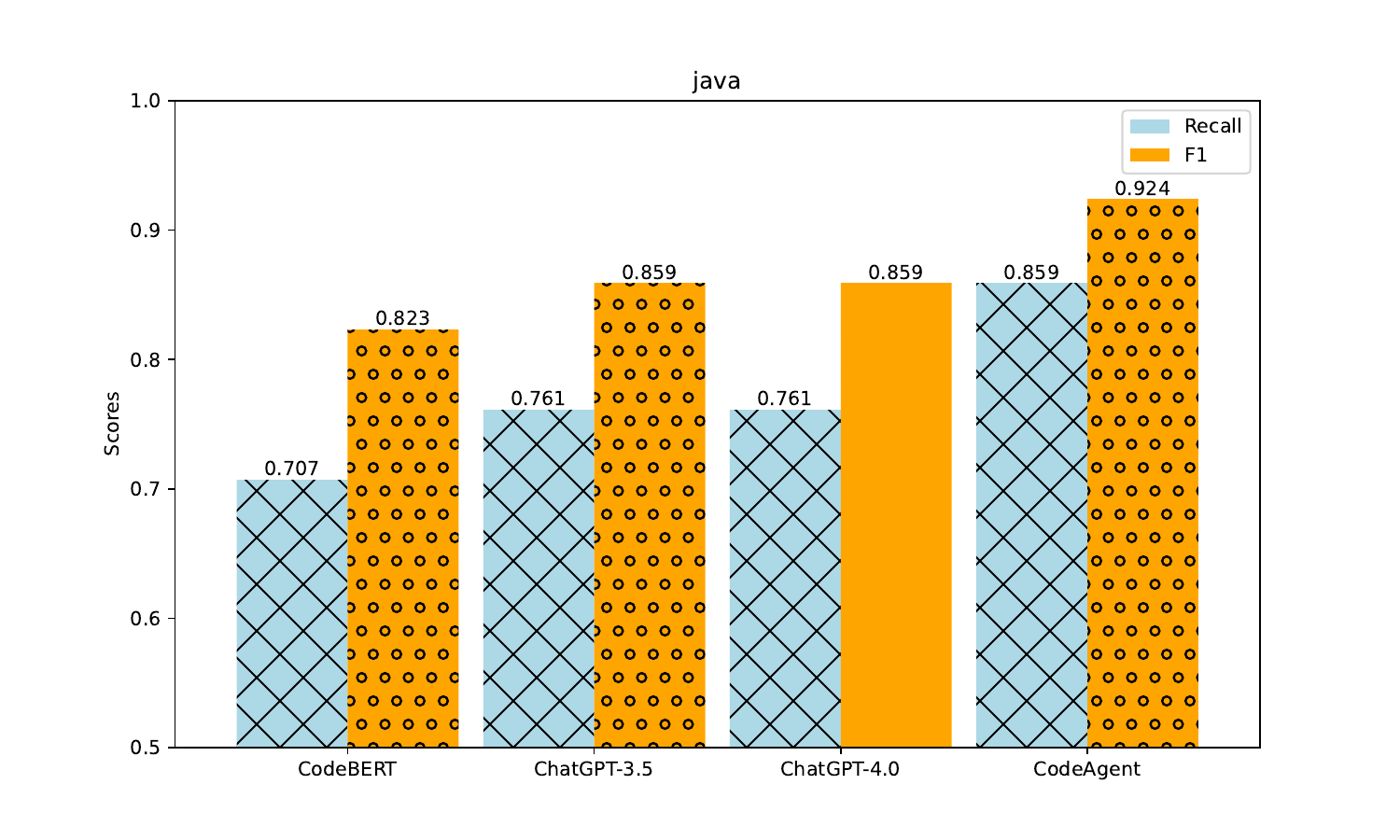}
    }
    \subfigure{
        \includegraphics[width=.3\linewidth]{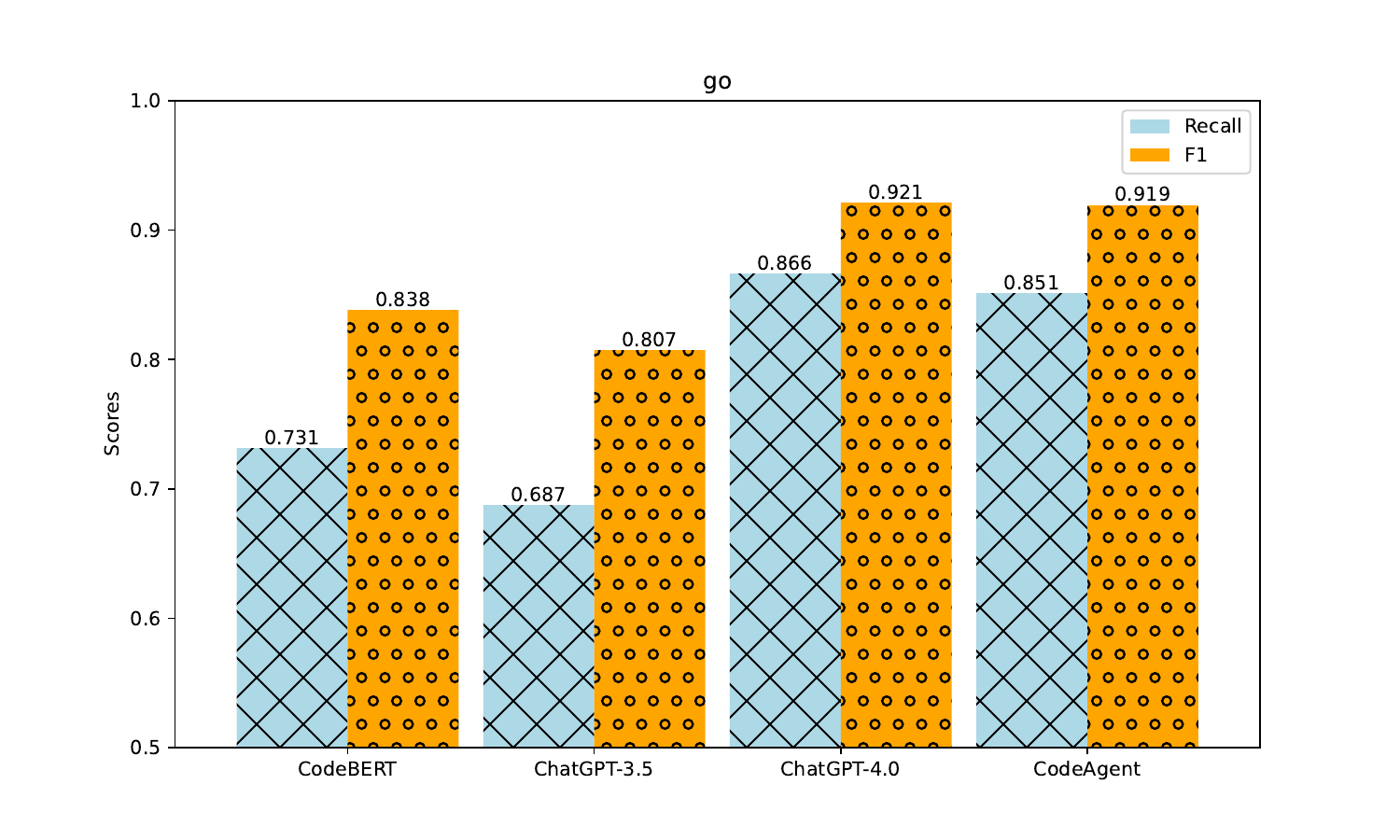}
    }
    \subfigure{
        \includegraphics[width=.3\linewidth]{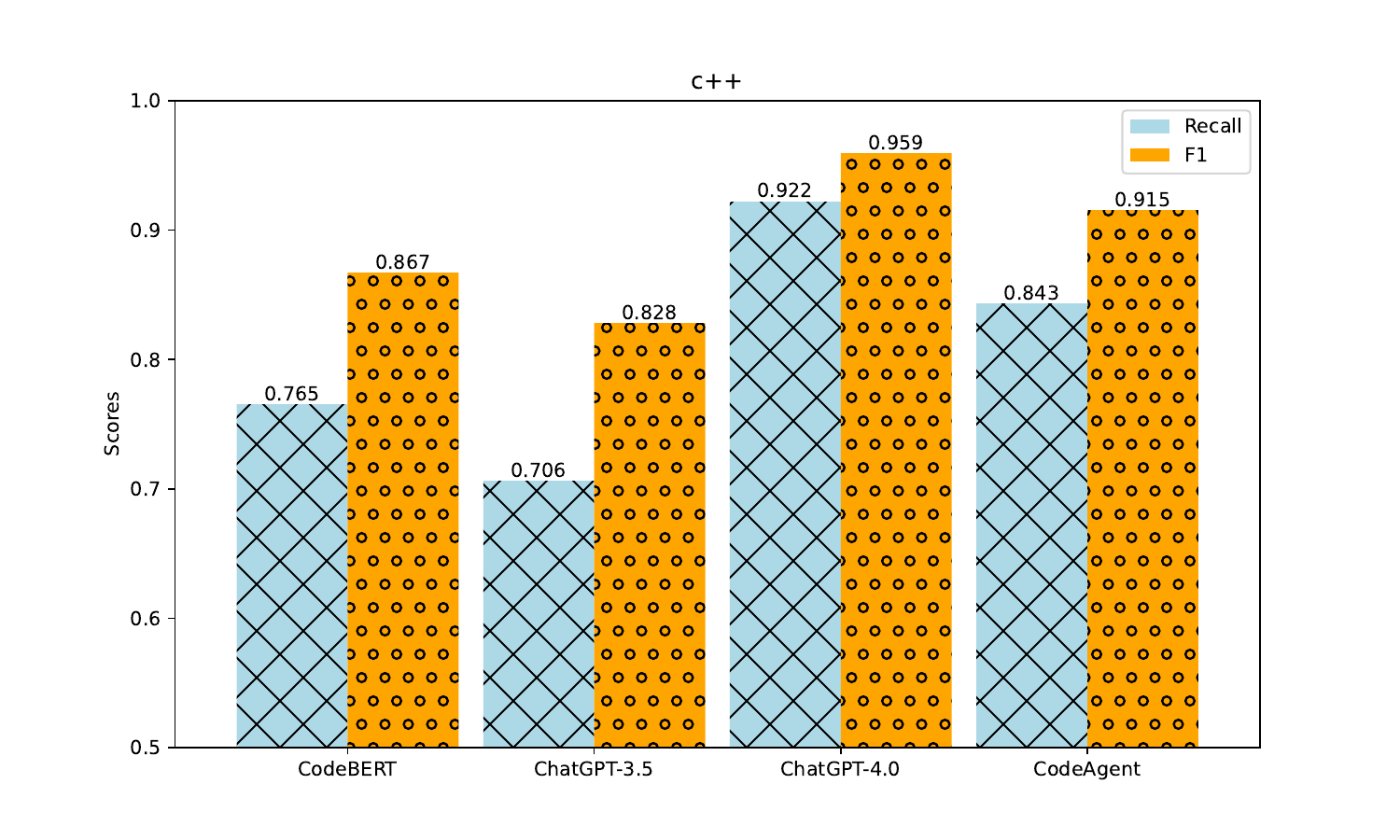}
    }
    \subfigure{
        \includegraphics[width=.3\linewidth]{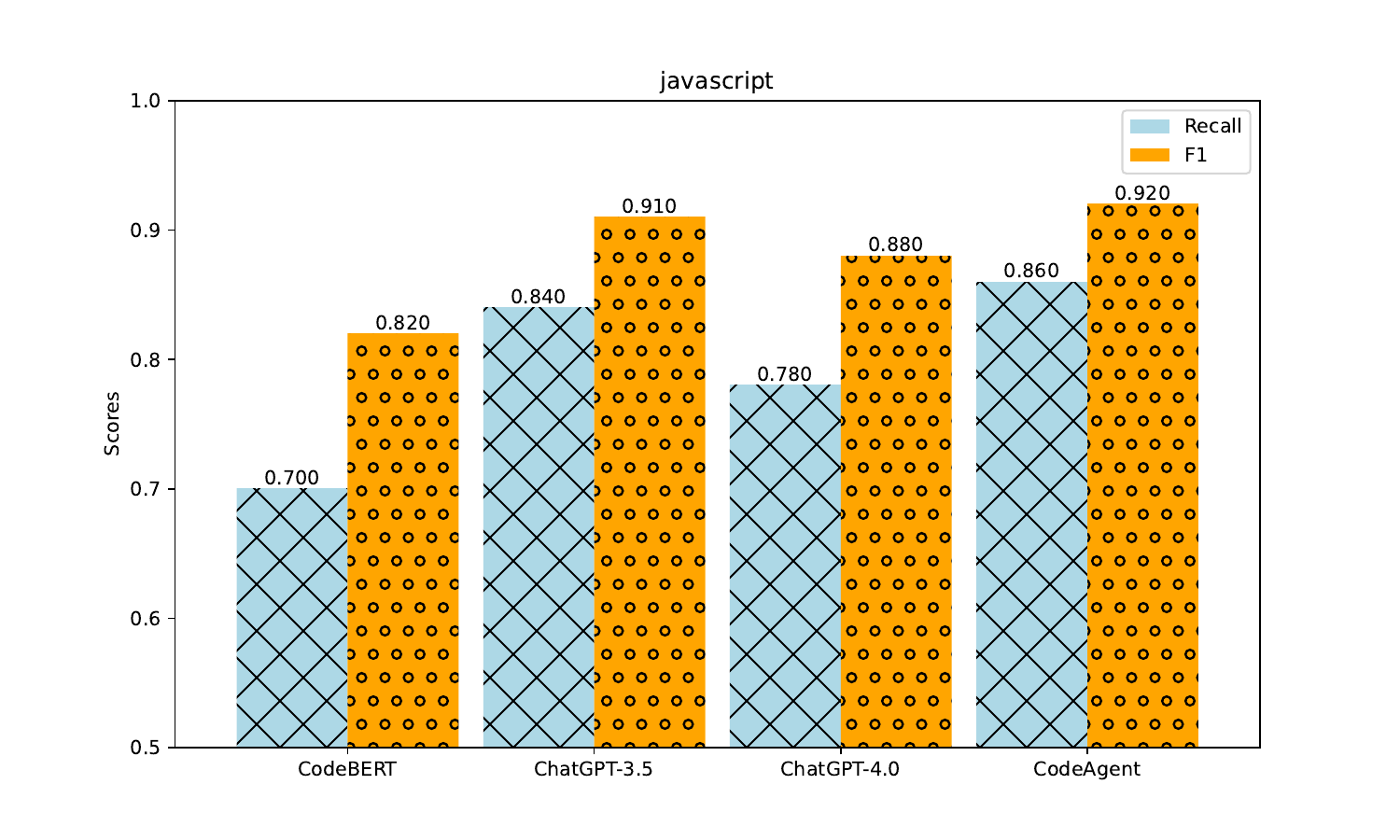}
    }
    \subfigure{
        \includegraphics[width=.3\linewidth]{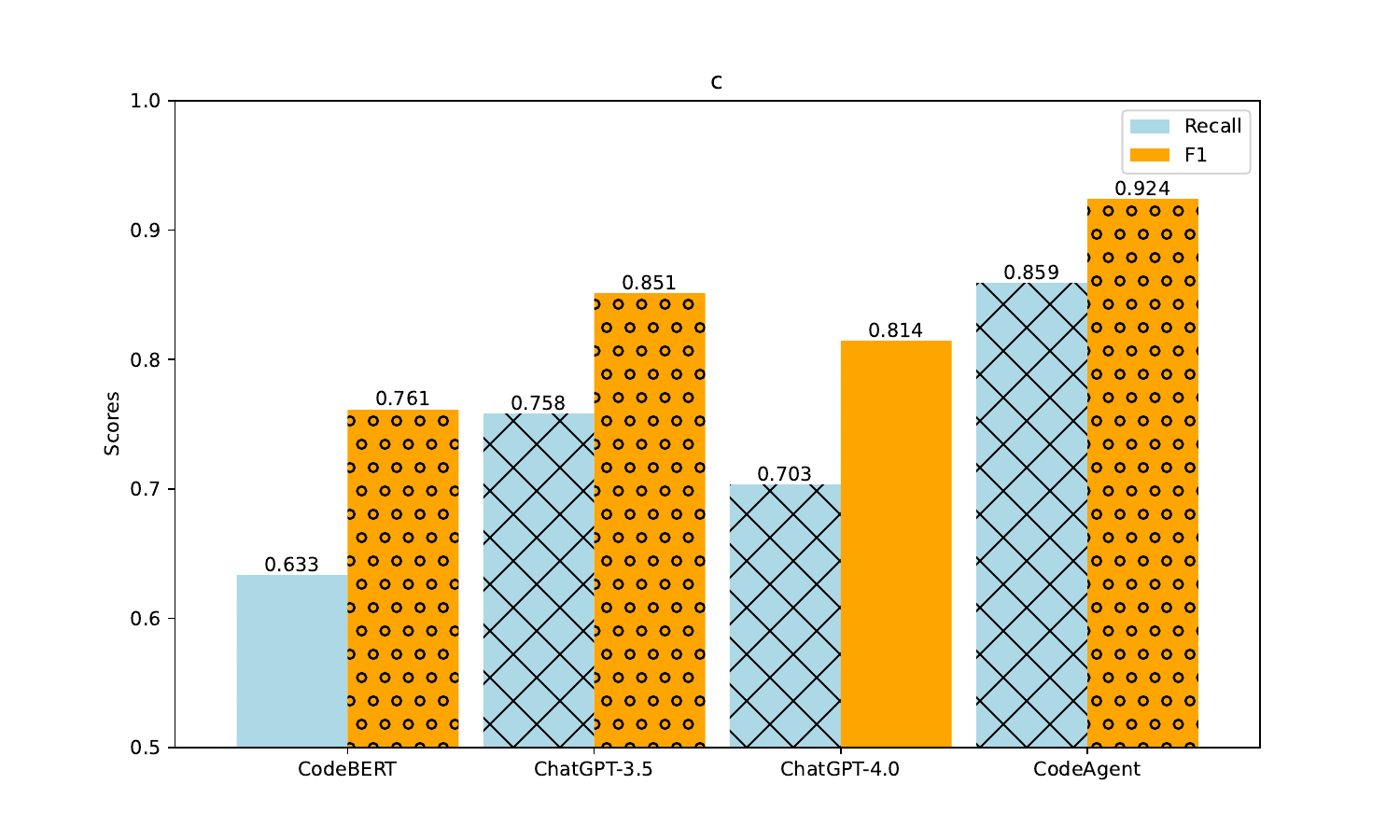}
    }
    \subfigure{
        \includegraphics[width=.3\linewidth]{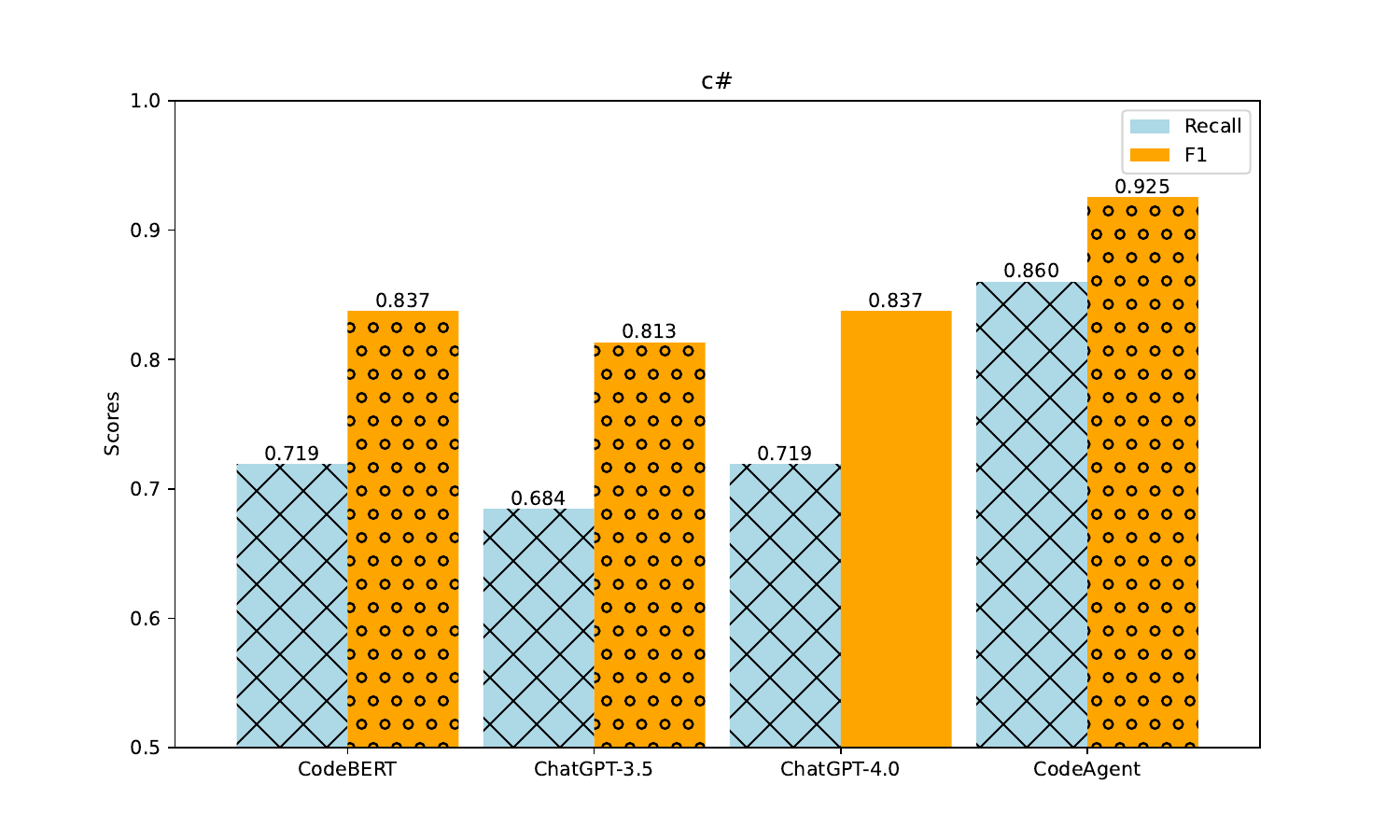}   
    }
    \subfigure{
        \includegraphics[width=.3\linewidth]{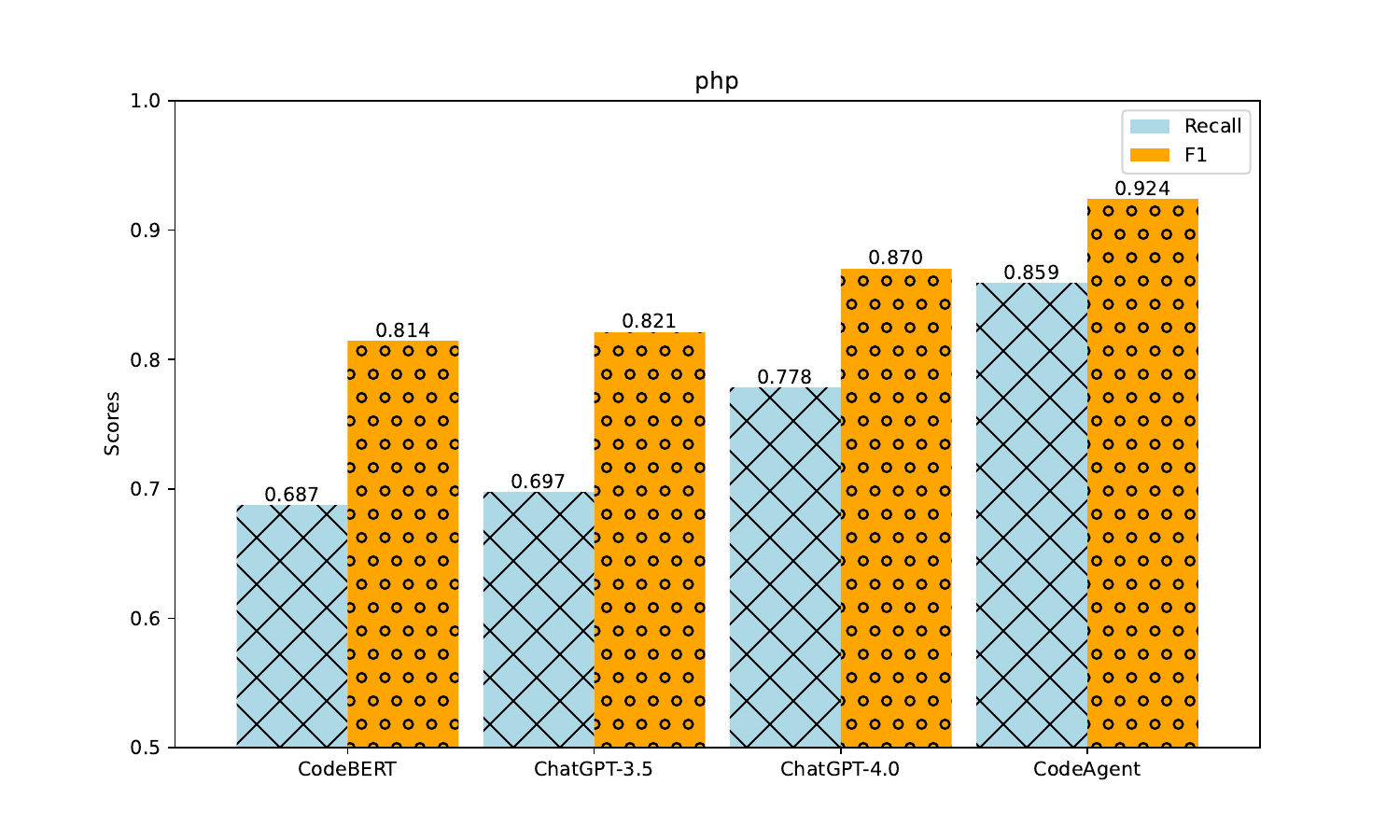}
    }
    \subfigure{
        \includegraphics[width=.3\linewidth]{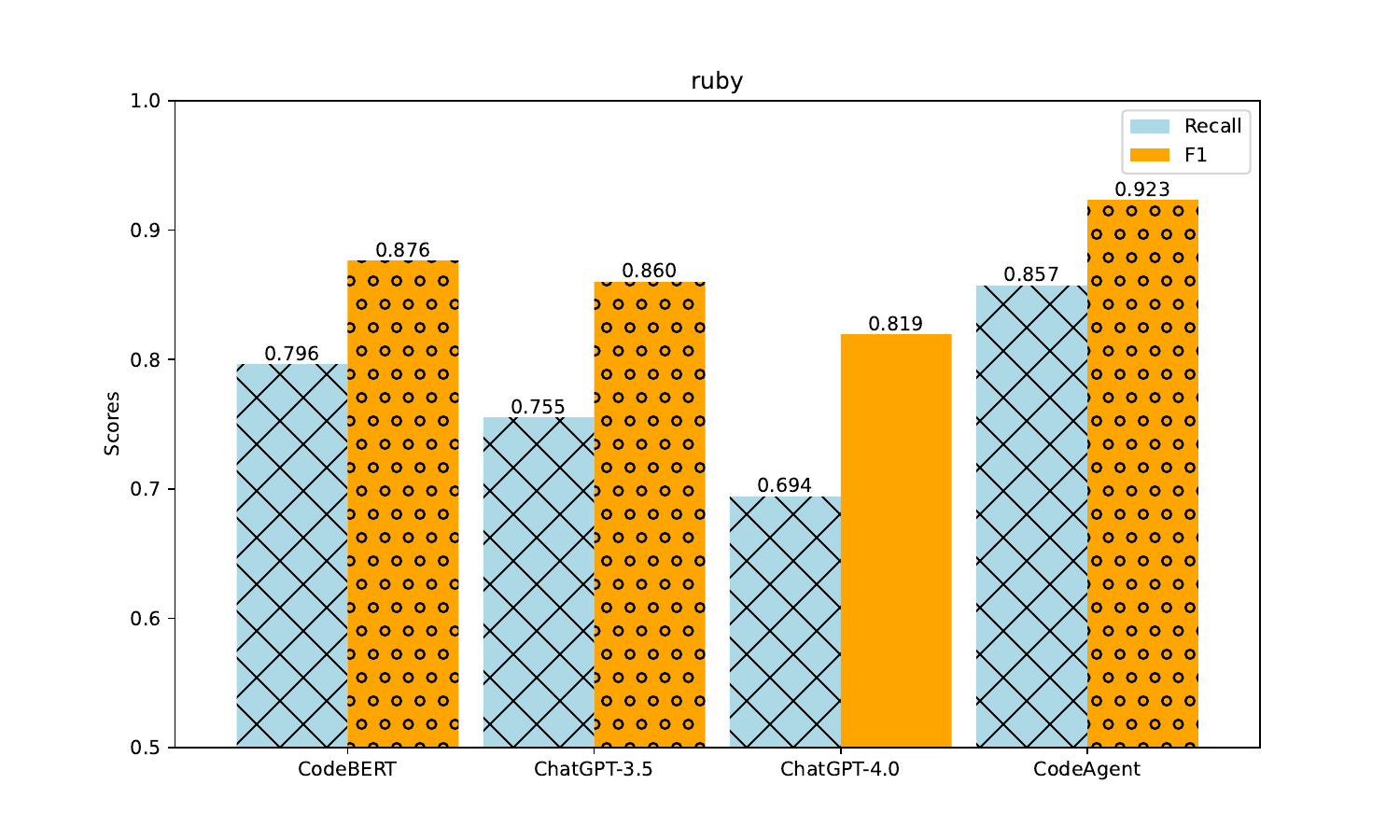}
    }
    \caption{Comparison of models on the \textbf{closed} data across 9 languages on \textbf{FA task}.}
    \label{faclosedmetricdetail}
\end{figure*}

\section{Case Study}
As shown in Table~\ref{lab:correlation}, we can easily localize the figure numbers of case studies for specific programming languages.
\subsection{Performance on 9 languages}
\begin{table}[H]
\centering
\caption{Correlation Table between specific programming language and case study.}
\begin{tabular}{ll}
\hline
\begin{tabular}[c]{@{}l@{}}Programming\\ Language\end{tabular} & Figure No. \\ \hline
Python                                                         & \ref{fig:pythoncase}          \\
Java                                                           & \ref{fig:javacase}          \\
Go                                                             & \ref{fig:gocase}          \\
C++                                                            & \ref{fig:c++case}           \\
JavaScript                                                     & \ref{fig:javascriptcase}           \\
C                                                              & \ref{fig:ccase}           \\
C\#                                                            & \ref{fig:csharpcase}           \\
php                                                            & \ref{fig:phpcase}           \\
Ruby                                                           & \ref{fig:rubycase}           \\ \hline
\end{tabular}
\label{lab:correlation}
\end{table}

\begin{figure*}[htbp]
    \centering
    \includegraphics[width=0.9\linewidth]{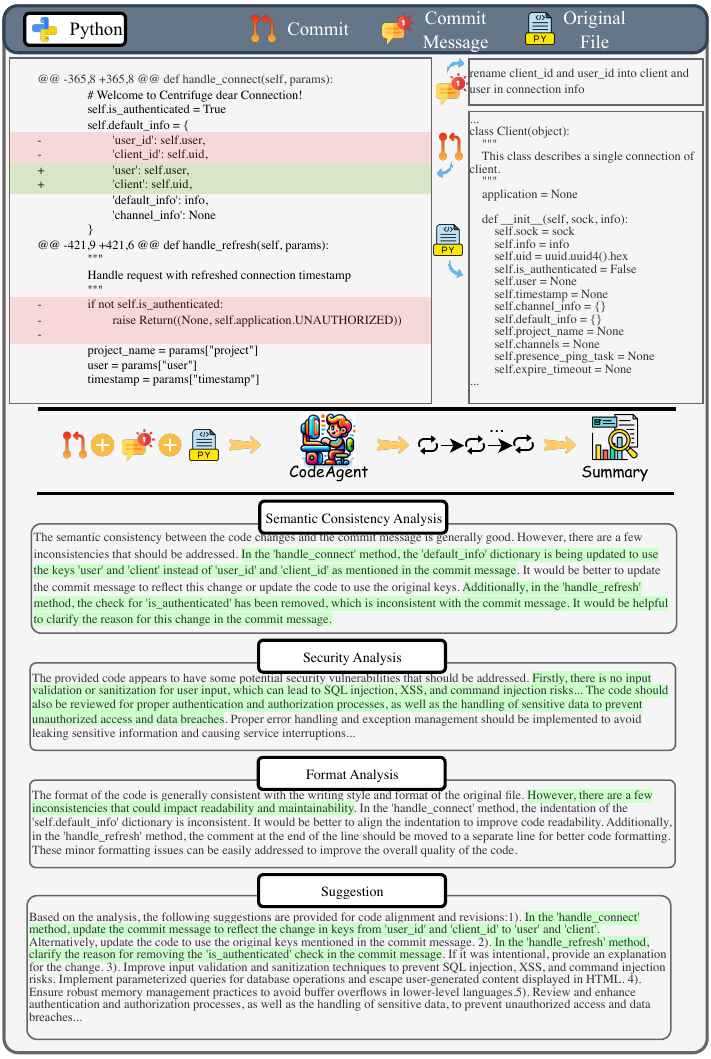}
    \caption{Example in Python project with sha value: ``0d3b94bbd7528ed8649cdab4370a62e508bad82c"}
    \label{fig:pythoncase}
\end{figure*}

\begin{figure*}[htbp]
    \centering
    \includegraphics[width=0.9\linewidth]{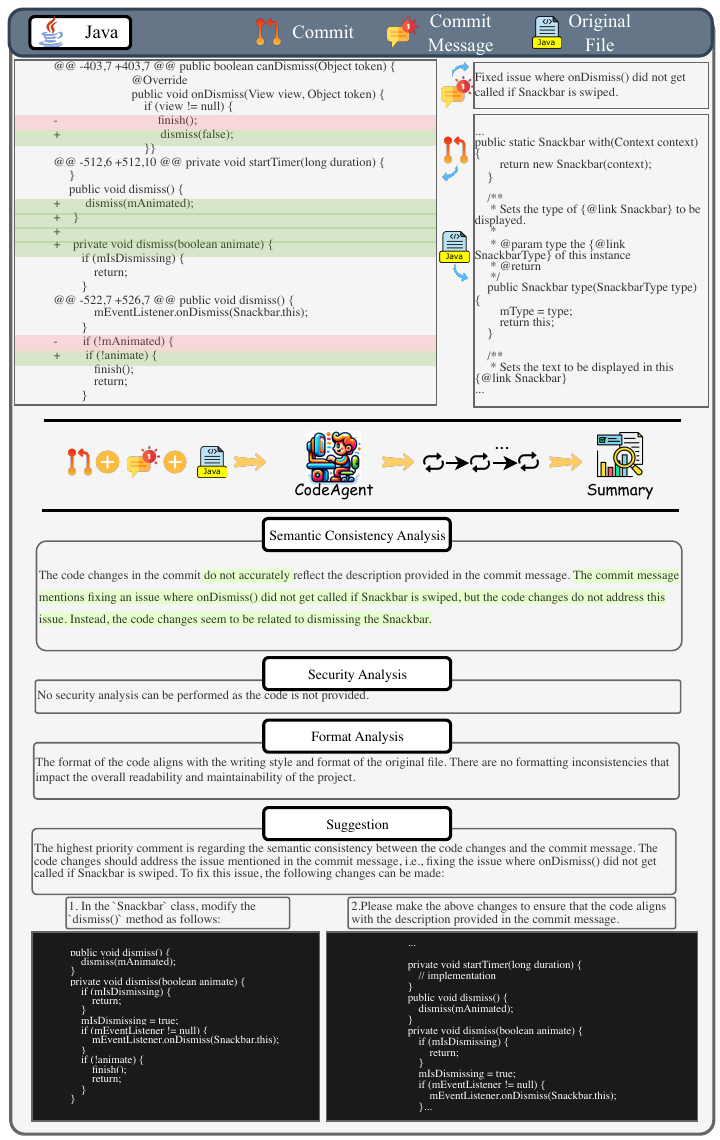}
    \caption{Example in Java project with sha value: ``6e3c6d17d943f5ac70b421653eb167e0c34b119f"}
    \label{fig:javacase}
\end{figure*}

\begin{figure*}[htbp]
    \centering
    \includegraphics[width=0.9\linewidth]{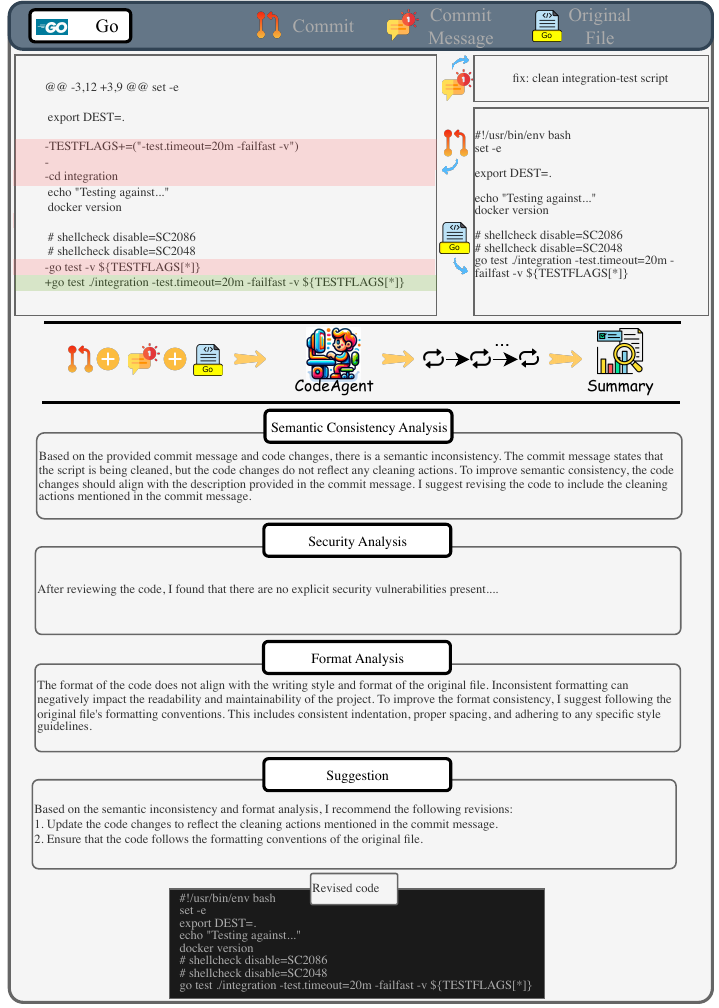}
    \caption{Example in Go project with sha value: ``a6b00b097b1a8140bd8b60177247dc830374c8ed"}
    \label{fig:gocase}
\end{figure*}
\begin{figure*}[htbp]
    \centering
    \includegraphics[width=0.9\linewidth]{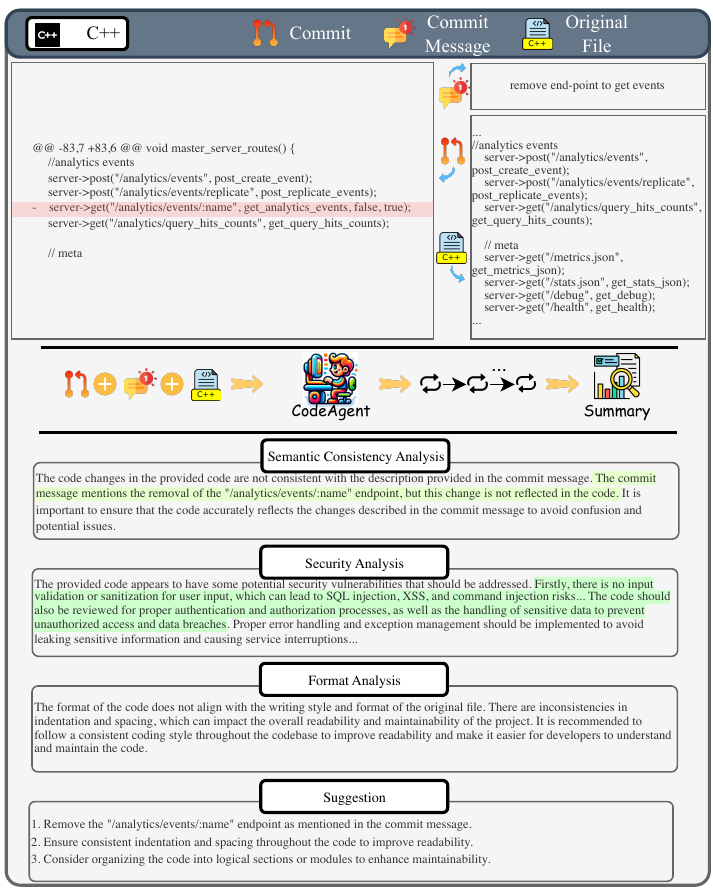}
    \caption{Example in C++ project with sha value: ``09e88c7c20347cdf513323fa8c616aba0a1b9ea7"}
    \label{fig:c++case}
\end{figure*}
\begin{figure*}[htbp]
    \centering
    \includegraphics[width=0.9\linewidth]{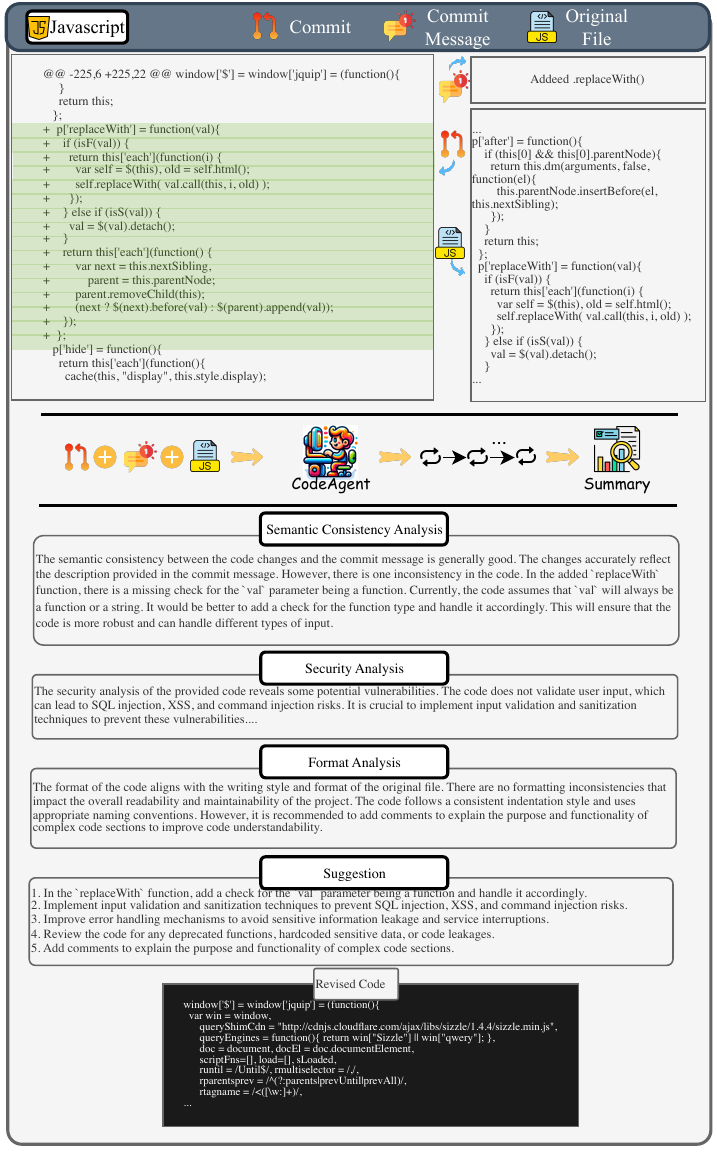}
    \caption{Example in JavaScript project with sha value: ``6e9a80d93b470aef4978677ed0b7c9996fddbe20"}
    \label{fig:javascriptcase}
\end{figure*}
\begin{figure*}[htbp]
    \centering
    \includegraphics[width=0.9\linewidth]{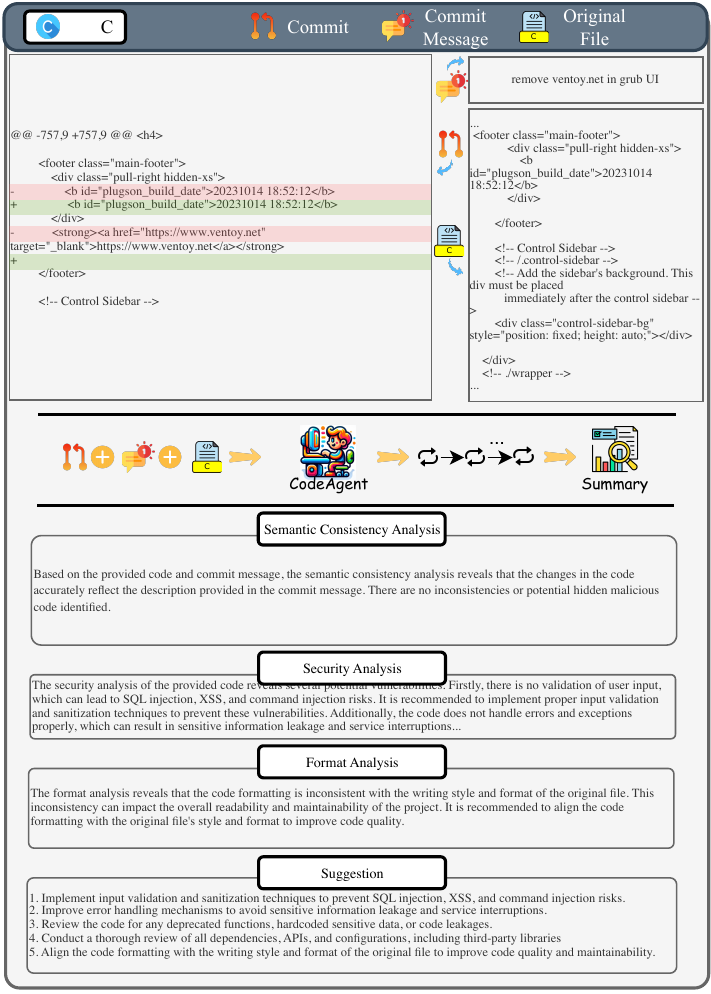}
    \caption{Example in C project with sha value: ``5f93cf43dd6f809d7927abb79884285ad77e8a58"}
    \label{fig:ccase}
\end{figure*}
\begin{figure*}[htbp]
    \centering
    \includegraphics[width=0.9\linewidth]{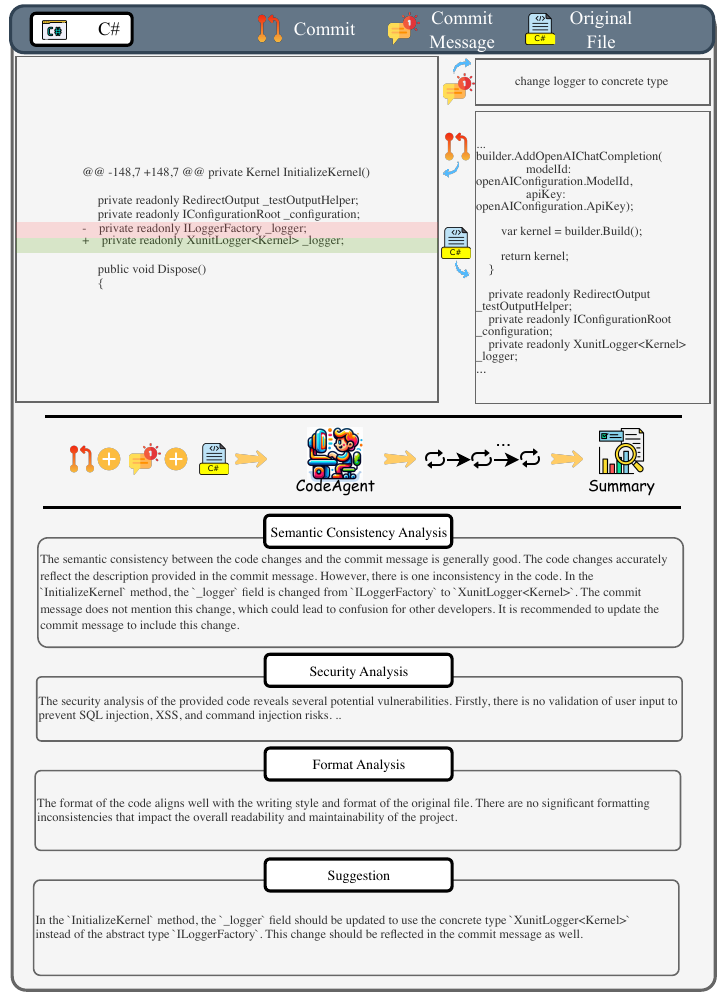}
    \caption{Example in C\# project with sha value: ``0e231c7a81b318e9eade972f7b877e66128ed67d"}
    \label{fig:csharpcase}
\end{figure*}

\begin{figure*}[htbp]
    \centering
    \includegraphics[width=0.9\linewidth]{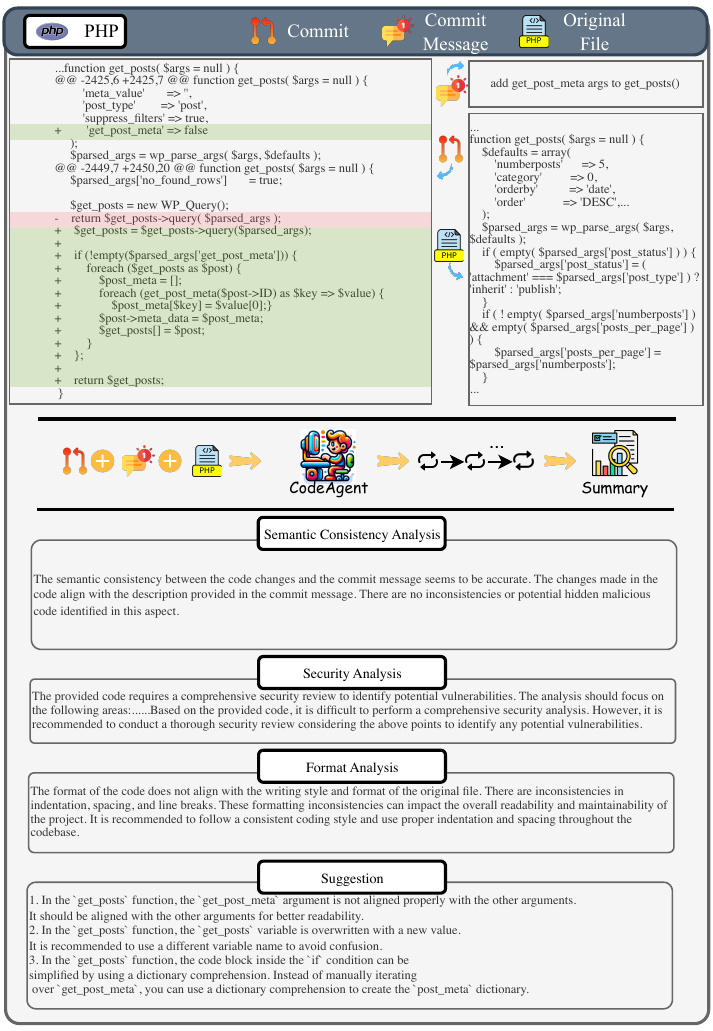}
    \caption{Example in PHP project with sha value: ``6679f059b9a0972a90df198471188da652f5c547"}
    \label{fig:phpcase}
\end{figure*}

\begin{figure*}[htbp]
    \centering
    \includegraphics[width=0.9\linewidth]{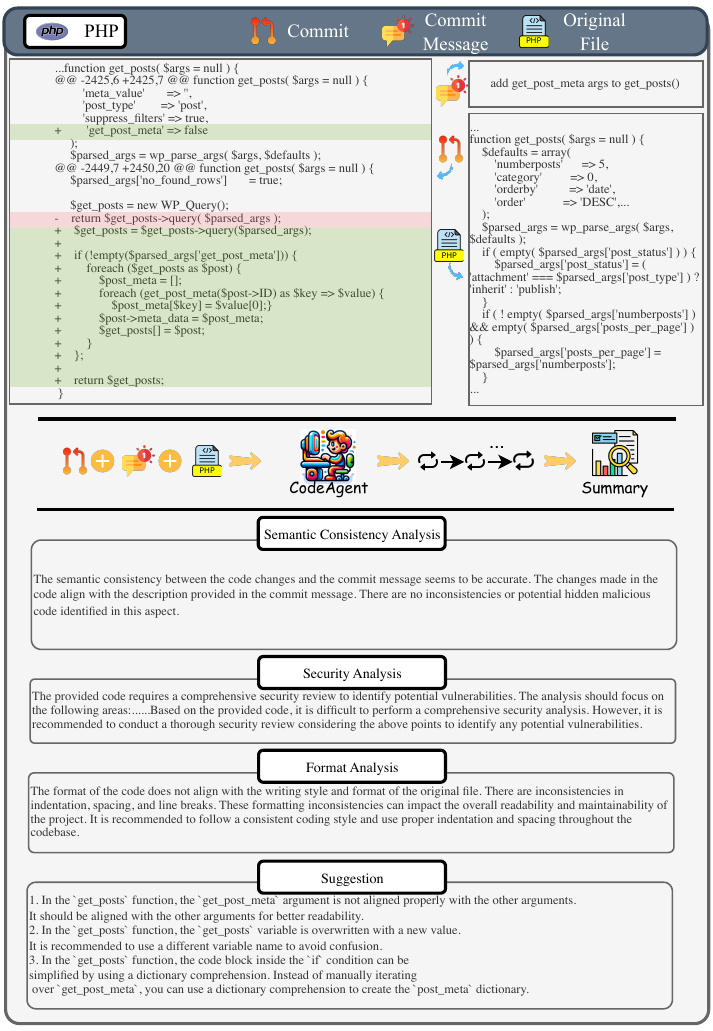}
    \caption{Example in Ruby project with sha value: ``584f72e7f4c65066ccbd2183af76bf380b6eb974"}
    \label{fig:rubycase}
\end{figure*}

\subsection{Difference of \tool{}-3.5 and \tool{}-4.0} \label{sec:agentdiff}
\tool{}-3.5 and \tool{}-4.0 in this paper has no difference in general code review, however, \tool{}-4.0 is more powerful in processing long input sequences and logic reasoning. As shown in Figure~\ref{fig:javacodeagents}, we take one example of consistency detection between commit and commit message and find that \tool{}-4.0 diffs from \tool{}-3.5 in the detailed explanation. \tool{}-3.5 output a report with 15k lines while \tool{}-4.0 outputs a report with more than 17.7k lines. Detailed data is shown in \url{https://zenodo.org/records/10607925}.

\begin{figure*}
    \centering
    \includegraphics[width=0.9\linewidth]{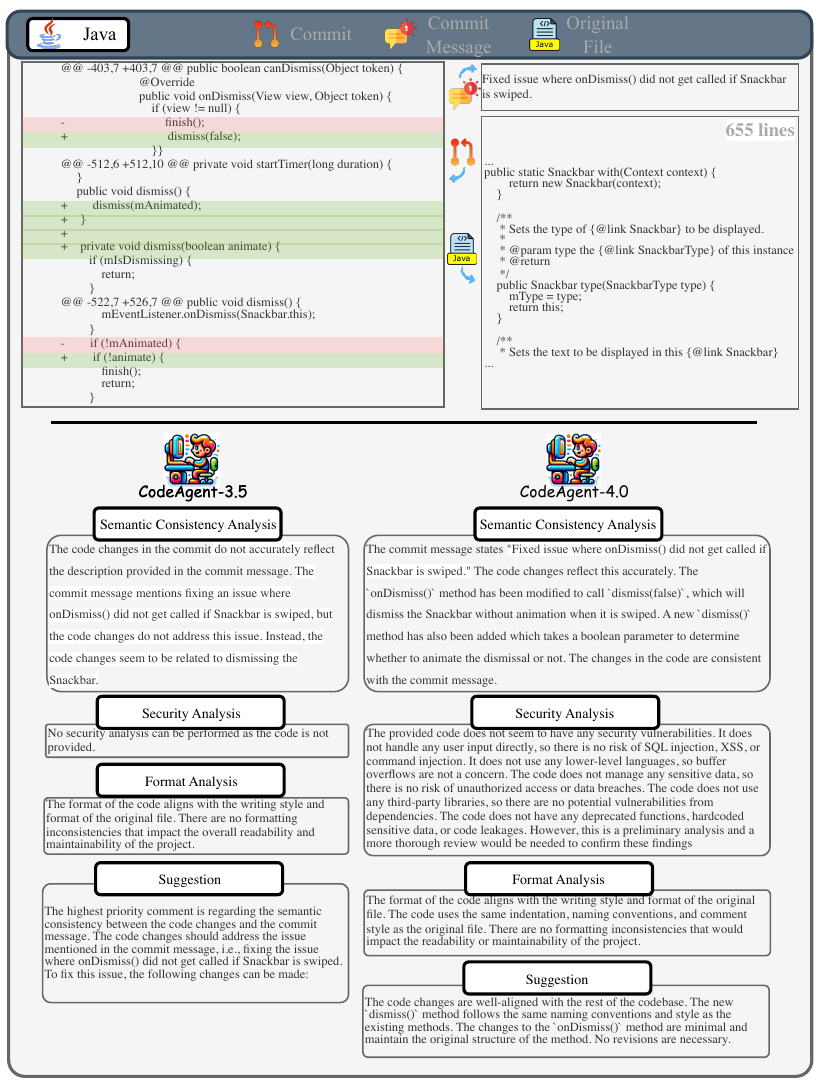}
    \caption{Comparison of code review between \tool{}-3.5 and \tool{}-4.0 on one Java project with sha value: ``6e3c6d17d943f5ac70b421653eb167e0c34b119f"}
    \label{fig:javacodeagents}
\end{figure*}

\section{Ablation study} \label{sec:abs}
In this section, we evaluate the performance of different parts in \tool{} in vulnerability analysis. \tool{} is based on chain-of-thought (COT) and large language model (a.k.a. GPT). As shown in Section~\ref{sec:va}, \tool{} outperforms baselines (a.k.a. CodeBERT, GPT-3.5, GPT-4.0) across 9 different languages. The performance mainly comes from the combination of COT and QA-Checker. Thus, we design an additional version called \tool{}$_{w/o}$, which means \tool{} without QA-Checker. Then, we use \tool{}$_{w/o}$ to do vulnerability analysis and compare with \tool{}. We first discuss about the result of \tool{}$_{w/o}$ and then discuss about comparison between \tool{} and \tool{}$_{w/o}$.

\paragraph{Overview of Vulnerabilities in \tool{}$_{w/o}$}
Table~\ref{tab:languagediffqachecker} presents the findings of \tool{}$_{w/o}$, a variant of the original \tool{}, in identifying vulnerabilities across different programming languages. The table showcases the number of `merged' and `closed' vulnerabilities in languages such as Python, Java, Go, C++, JavaScript, C, C\#, PHP, and Ruby. Notably, Python leads in the `merged' category with a total of 1,057 cases, of which 140 are confirmed, yielding a Rate$_{merge}$ of 13.25\%. In contrast, languages like Go and Ruby show lower vulnerability counts in both `merged' and `closed' categories. The table also includes Rate$_{close}$ and Rate$_{avg}$, providing insights into the effectiveness of vulnerability management across these languages.

\paragraph{Detailed Comparison between \tool{} and \tool{}$_{w/o}$}
Comparing the findings in Table~\ref{tab:languagediffqachecker} with those in Table~\ref{tab:languagediff}, we observe some notable differences in vulnerability detection by \tool{} and \tool{}$_{w/o}$. While the overall trend of higher `merged' vulnerabilities in Python and lower counts in Go and Ruby remains consistent, Table~\ref{tab:languagediffqachecker} shows a slight reduction in the Rate$_{merge}$ for most languages, suggesting a more conservative confirmation approach in \tool{}$_{w/o}$. Similarly, Rate$_{close}$ and Rate$_{avg}$ values in Table~\ref{tab:languagediffqachecker} generally indicate a lower proportion of confirmed vulnerabilities compared to Table~\ref{tab:languagediff}, reflecting potentially different criteria or efficacy in vulnerability assessment. These variations highlight the impact of QA-Checker in \tool{}.

\begin{table*}[htbp]
\centering
\caption{Vulnerable problems (\#) found by \tool{}$_{w/o}$ }
\label{tab:languagediffqachecker}
\resizebox{0.9\textwidth}{!}{%
\begin{tabular}{l|ccccccccc}
\hline

\hline

\hline

\hline
\tool{}&  Python & Java & Go & C++ & JavaScript & C & C\# & PHP & Ruby\\ 
\hline

\hline

merged (total\#) &  1,057 & 287 & 133 & 138 & 280 & 114 & 206 & 173 & 202\\

merged (confirmed\#) &  140 & 17 & 10 & 12 & 28 & 9 & 21 & 28 & 17\\

Rate$_{merge}$ &  13.25\% & 5.92\% & 7.52\% & 8.70\% & 10.00\% & 7.89\% & 10.19\% & \cellcolor{gray!30}16.18\% & 8.42\%\\

closed (total\#)&  248 & 97 & 74 & 56 & 112 & 146 &62 & 105 & 55\\ 

closed (confirmed\#) &  36 & 9 & 5 & 12 & 16 & 26 &7 & 15 & 5\\ 

Rate$_{close}$ &  14.52\% & 9.28\% & 6.76\% & \cellcolor{gray!30}21.43\% & 14.29\% & 17.81\% & 11.29\% & 14.29\% & 9.09\%\\

Total number (\#) &  1,305 & 384 & 207 & 194 & 392 & 260 &268 & 278 & 257\\ 
Total confirmed (\#) & 176 & 26 & 15& 24 &44 &35 & 28 & 43 & 22 \\ 
Rate$_{avg}$ & 13.49\% & 6.77\% & 7.25\% & 12.37\% & 11.22\% & 13.46\% & 10.45\% & \cellcolor{gray!30}15.47\% & 8.56\%\\

\hline

\hline

\hline

\hline
\end{tabular}%
 }
\end{table*}
\section{Cost statement}
As shown in Table \ref{tabcost}, \tool{}-4 has a higher query time and cost compared to \tool{}-3.5 due to its increased complexity and capabilities. We acknowledge that the integration of AI models and multi-agent systems may introduce complexity and require specialized knowledge for implementation.
\begin{table}[H]
\caption{Summarizes the average query time and cost per code review for \tool{}-3.5 and \tool{}-4.}
\label{tabcost}
\centering
\resizebox{0.9\columnwidth}{!}{
\begin{tabular}{lll}
\hline
Model                        & Query Time(min) & Cost in USD \\ \hline
\tool{}-3.5 & 3               & 0.017    \\
\tool{}-4   & 5               & 0.122    \\ \hline
\end{tabular}}
\end{table}

\section{Tool}
We develop a website for \tool{}, which is shown in Figure~\ref{fig:website}, and it is also accessable by visiting following link:

\begin{center}
    \url{https://code-agent-new.vercel.app/index.html}
\end{center}

\begin{figure}[htbp]
    \centering
    \includegraphics[width=1\columnwidth]{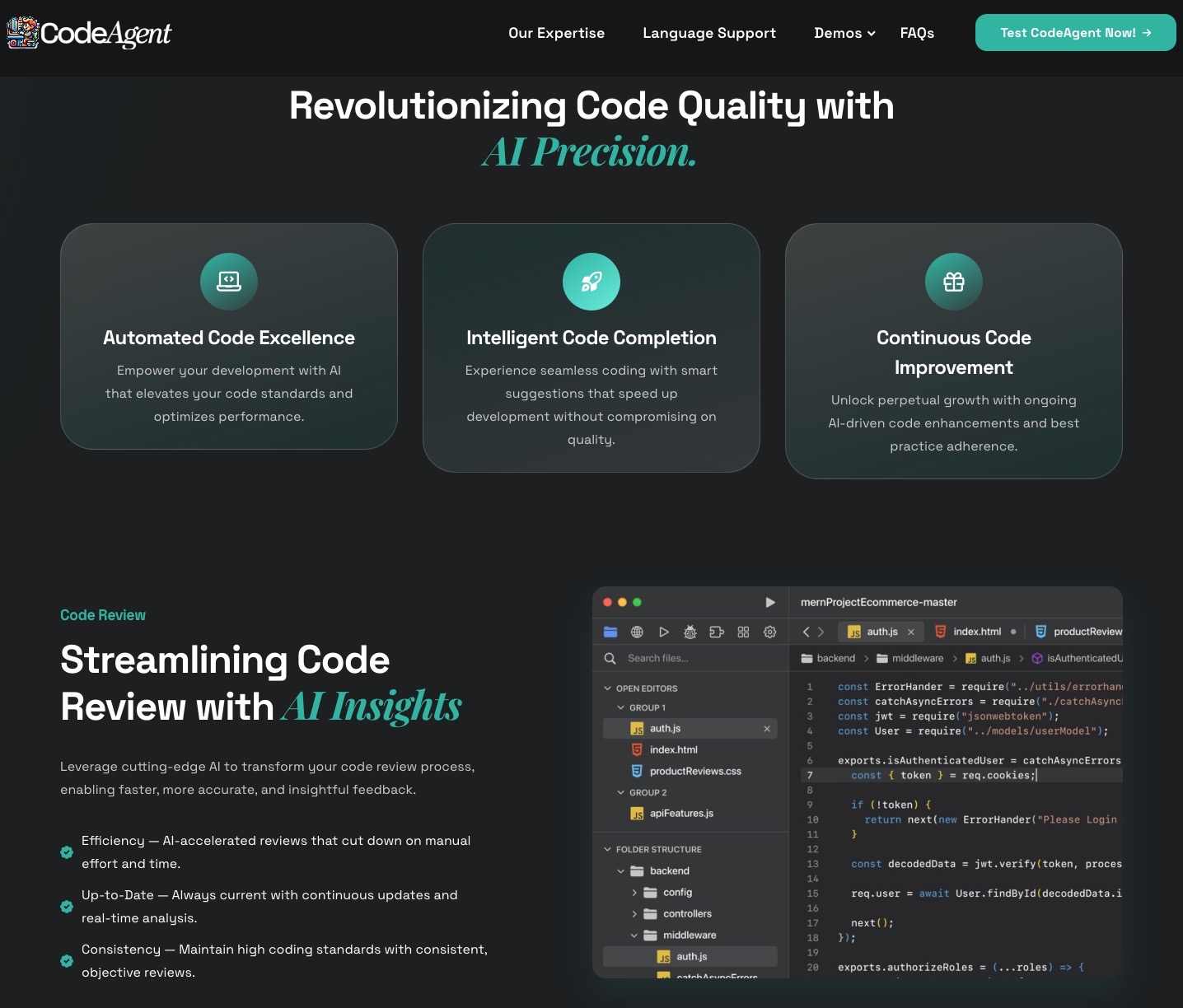}
    \caption{website of \tool{}}
    \label{fig:website}
\end{figure}

\end{document}